\documentclass[12pt]{article}
\usepackage{amsmath}
\usepackage{amssymb}
\usepackage{amsthm}
\usepackage{graphicx}
\usepackage{bbm}
\usepackage{authblk}
\usepackage{natbib}
\newcommand{\IF}{\operatorname{IF}}
\newcommand{\E}{\operatorname{\mathbb{E}}}
\newcommand{\trace}{\operatorname{tr}}

\newtheorem{theorem}{Theorem}
\newtheorem{lemma}{Lemma}
\newtheorem{corollary}{Corollary}

\newtheorem{remark}{Remark}

\newcommand{\vect}[1]{\boldsymbol{#1}}
\newcommand{\modulus}{\mathrm{mod}}
\newcommand{\argmax}{\mathrm{argmax}}
\newcommand{\email}[1]{\texttt{#1}}

\begin{document}

\title{Weighted likelihood methods for robust fitting of wrapped models for $p$-torus data}

\author[1]{Claudio Agostinelli}
\author[2]{Luca Greco} 
\author[3]{Giovanni Saraceno}

\affil[1]{Department of Mathematics, University of Trento, Italy, \email{claudio.agostinelli@unitn.it}}
\affil[2]{University Giustino Fortunato, Benevento, Italy, \email{l.greco@unifortunato.eu}}
\affil[3]{Department of Biostatistics, University of Buffalo, New York, USA, \email{gsaracen@buffalo.edu}}

\maketitle

\begin{abstract}
We consider robust estimation of  wrapped models to multivariate circular data that are points on the surface of a $p$-torus based on the weighted likelihood methodology. 
Robust model fitting is achieved by a set of weighted likelihood estimating equations, based on the computation of data dependent weights aimed to down-weight anomalous values, such as unexpected directions that do not share the main pattern of the bulk of the data. 
Weighted likelihood estimating equations	with weights evaluated on the torus orobtained after unwrapping the data onto the Euclidean space are proposed and compared. 
Asymptotic properties and robustness features of the estimators under study have been studied, whereas their finite sample behavior has been investigated by Monte Carlo numerical experiment and real data examples.


\end{abstract}

\section{Introduction}
\label{sec:0}
Multivariate circular data arise commonly in many different fields, including the analysis of wind directions \citep{lund1999cluster, Agostinelli2007}, animal movements \citep{ranalli2020model, rivest2016general}, handwriting recognition \citep{bahlmann2006directional}, people orientation \citep{baltieri2012people}, cognitive and experimental psychology \citep{warren2017wormholes}, human motor resonance \citep{cremers2018one}, neuronal activity \citep{rutishauser2010human} and protein bioinformatics \citep{mardia2007protein,mardia2012mixtures, Eltzner2018}. The reader is pointed to \cite{MardiaJupp2000, JammalamadakaSenGupta2001,  pewsey2013circular} for a general review.
The data can be thought as points on the surface of a $p$-torus, embedded in a $(p+1)$-dimensional space, whose surface is obtained by revolving the unit circle in a $p-$dimensional manifold. 
A $p$-torus is topologically equivalent to a product of a circle $p$ times by itself, written $\mathbb{T}^p, \ p \ge 1$ \citep{munkres2018elements}. The peculiarity of torus data is periodicity, that reflects in the boundedness of the sample space and often of the parametric space. 

In order to illustrate the nature of torus data, let us consider a bivariate example, concerning
$n=490$ backbone torsion angle pairs $(\phi, \psi)$ for the protein 8TIM. Data are available from the {\tt R} package {\tt BAMBI} \citep{bambi} and are extracted from the vast Protein Data Bank \citep{bourne:2000}.  
The protein is an example of a TIM barrel folded into eight $\alpha$-helices and eight parallel $\beta$-strands, alternating along the protein tertiary structure. It gets its name from the enzyme triose-phosphate isomerase, a conserved metabolic enzyme \citep{tim}. 
The data are shown in Figure \ref{fig:a} according to the Ramachandran plot of the angles over $[0,2\pi)\times [0,2\pi)$, in the left panel, or $[-\pi, \pi)\times [-\pi, \pi)$, in the right panel. Clearly, this type of graphical display is not unique and depends on how the angles are represented. Actually, the Ramachandran plot does not allow to show the intrinsic periodicity of the angles. In order to account for such wraparound nature of the data, one should topologically glue both pairs of opposite edges together with no twists. Then, the resulting surface is that of a torus with one hole (say, of genus one) in three dimensions. The data on the torus are displayed in Figure \ref{fig:b} from two different perspectives.
The limitations of the Ramachandran plot in the two dimensional space can be circumvented by {\it unwrapping} the data on a flat torus, that is the angles are revolved around the unit circle a fixed number of times in each dimension and transformed into linear data,  according to $\vect{x}=\vect{y}+2\pi \vect{j}$, for a given $\vect{j}\in\mathbb{Z}^2$.
This representation is shown in Figure \ref{fig:c} where the data are given for different choices of $\vect{j}\in\mathbb{Z}^2$: then, the same data structure repeats itself to reflect the periodic nature of the data. Dotted lines give multiples of $\pi$. 
 
\begin{figure*}[t]
	\centering
	\includegraphics[scale=0.28]{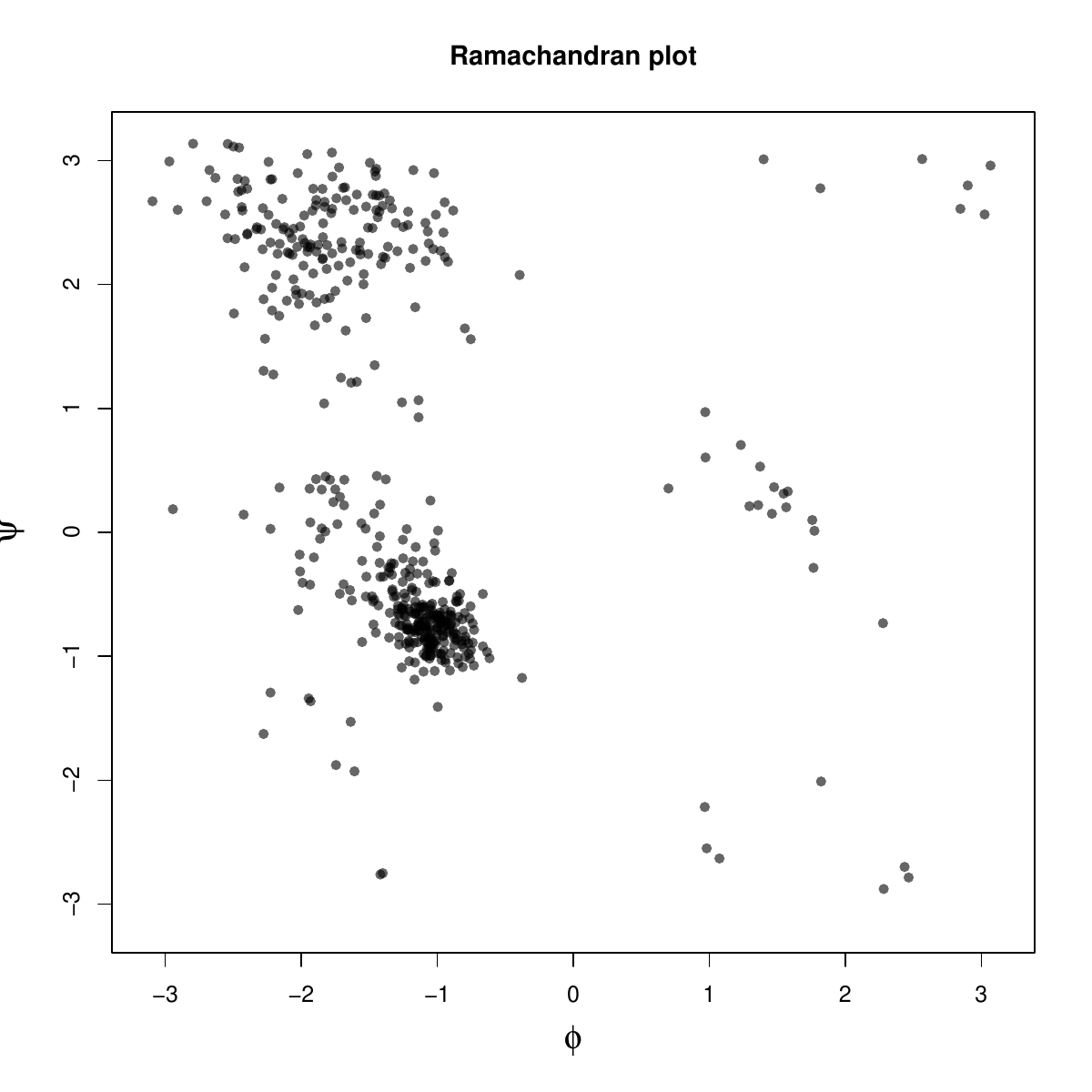}
	\includegraphics[scale=0.28]{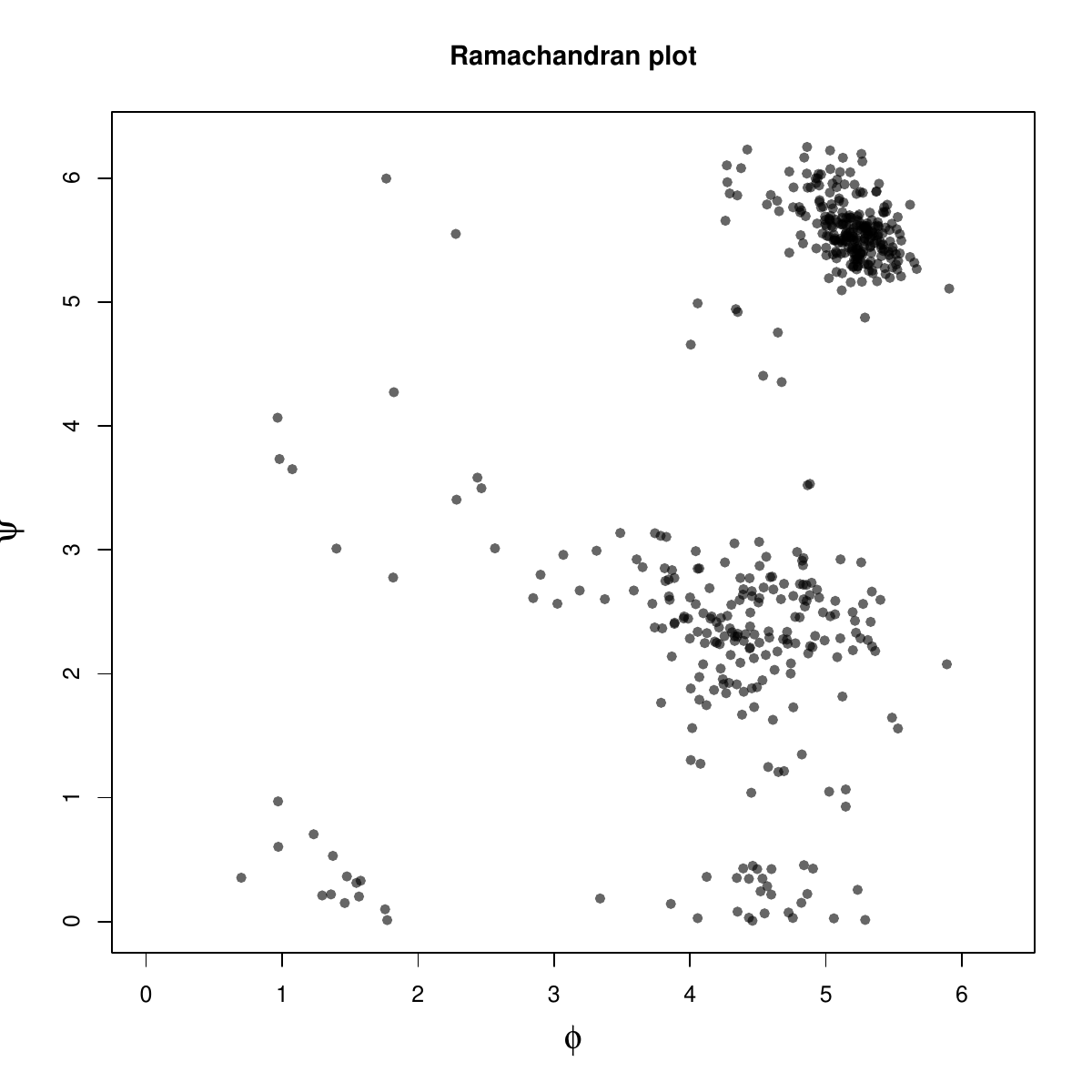}
	\caption{8TIM protein data. Ramachandran plot over $[0, 2\pi)\times [0,2\pi)$ (left) and over $[-\pi, \pi)\times [-\pi, \pi)$.} 
	\label{fig:a}
\end{figure*}

\begin{figure*}[t]
	\centering
	\includegraphics[scale=0.65]{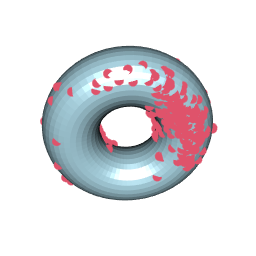}
	\includegraphics[scale=0.65]{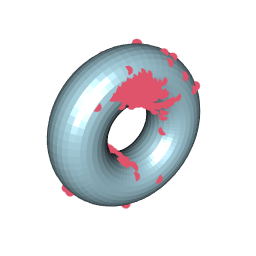}
	\caption{8TIM protein data. Bivariate angles as points on the surface of a torus from two different perspectives.} 
	\label{fig:b}
\end{figure*}

\begin{figure}[t]
	\centering
	\includegraphics[scale=0.3]{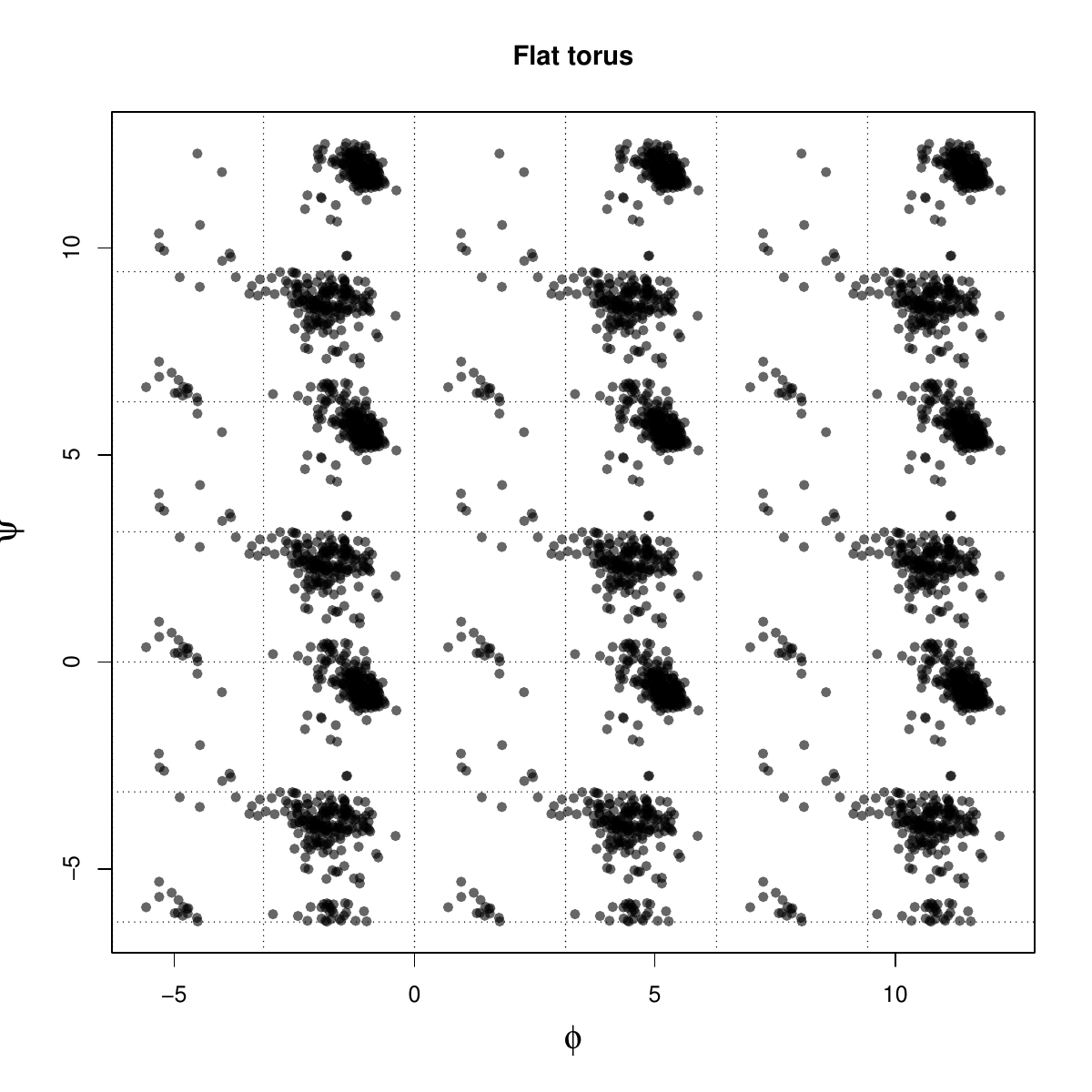}
	\caption{8TIM protein data. Flat torus plot. The dotted lines give multiples of $\mp\pi$.} 
	\label{fig:c}
\end{figure}

The problem of modeling circular data has been tackled through suitable distributions, such as the von Mises \citep{mardia1972statistics}. 
In a different fashion, in this paper, we focus our attention on the family of wrapped distributions \citep{MardiaJupp2000}. Wrapping is a popular method to define distributions for torus data.
Let $\vect{X}=(X_1,X_2,\ldots,X_p)$ be a {\it linear} random vector
with distribution function $M(\vect{x}; \vect{\theta})$ and corresponding probability density function $m(\vect{x};\vect{\theta})$, with $\vect{x}\in\mathbb{R}^p$ and $\vect{\theta}\in\Theta$.
Assume that each component is wrapped around the unit circle, i.e., $Y_d=X_d \ \modulus \ 2\pi, \ d=1,2,\ldots, p$, where $\modulus$ denotes the modulus operator. Then,
the distribution of $\vect{Y} = \vect{X} \ \modulus \ 2\pi$ is a $p-$variate wrapped distribution with distribution function
\begin{equation*}
M^{\circ}(\vect{y}; \vect{\theta})= \sum_{\vect{j} \in \mathbb{Z}^p}\left[ M(\vect{y} + 2 \pi \vect{j}; \vect{\theta}) - M(2 \pi \vect{j}; \vect{\theta})\right]  
\end{equation*}
and probability density function 
\begin{equation}
	m^{\circ}(\vect{y}; \vect{\theta})= \sum_{\vect{j} \in \mathbb{Z}^p} m(\vect{y} + 2 \pi \vect{j}; \vect{\theta}) \ , 
	\label{dens}
\end{equation}
$\vect{y}=(y_1,y_2,\ldots,y_p) \in[0,2\pi)^p$, $\vect{j}=(j_1,j_2,\ldots,j_p)\in\mathbb{Z}^p$.
The $p$-dimensional vector $\vect{j}$ is the vector of wrapping coefficients, that, if it was known, would describe how many times each component of the $p$-toroidal data point was wrapped. 
In other words, if we knew $\vect{j}$ along with $\vect{y}$, we would obtain the unwrapped data $\vect{x}=(x_1,x_2,\ldots,x_p)$ as $\vect{x}=\vect{y}+2\pi \vect{j}$.
Hereafter, we concentrate on unimodal and elliptically symmetric densities of the form
\begin{equation} \label{elliptic}
m(\vect{x}; \vect{\theta}) \propto \vert \Sigma \vert^{-1/2} h\left((\vect{x} - \vect{\mu})^\top \Sigma^{-1} (\vect{x} - \vect{\mu})\right)
\end{equation}
where $h(\cdot)$ is a strictly decreasing and non-negative function, $\vect{\theta} = (\vect{\mu}, \Sigma)$, 
$\vect{\mu}=(\mu_1,\mu_2,\ldots,\mu_p)$ is a location vector and $\Sigma$ is a $p \times p$ positive definite scatter matrix.
When $h(t)=\exp(-t/2)$, the multivariate normal distribution is recovered as a special case. 
Applying the component-wise wrapping of a $p$-variate normal distribution $\vect{X}\sim N_p(\vect{\mu}, \Sigma)$ onto a $p$-dimensional torus, one obtains 
the multivariate wrapped normal (WN), $\vect{Y}\sim WN_p(\vect{\mu},\Sigma)$, with mean vector 
$\vect{\mu}$ and variance-covariance matrix $\Sigma$.
Without loss of generality, we let $\vect{\mu}\in[0,2\pi)^p$ to ensure identifiability.

Torus data are not immune to the occurrence of outliers, that is unexpected values, such as angles or directions, that do not share the main pattern of the bulk of the data. The key to understanding circular outliers lies in the intrinsic periodic nature of the data. In particular, outliers in the circular setting differ from those in the linear case, in that angular distributions have bounded support. For classical {\it linear} data in an Euclidean space, one single outliers can lead the mean to minus or plus infinity. In contrasts, breakdown occurs in directional data when contamination causes the mean direction to change by at most $\pi$ \citep{davies2005breakdown, davies2006addendum}.
Marginally, the occurrence and subsequent detection of anomalous circular data points clearly depends on the concentration of the data around some main direction. The lower the concentration, the more outliers are unlikely to occur and have a little effect on estimates of location or spread. Furthermore, in a multivariate framework, outliers can violate the main correlation structures of the data and lead to misleading associations.
Therefore, when outliers do contaminate the torus data at hand, they can very badly affect likelihood based estimation, leading to unreliable inferences. The problem of robust fitting for directional data has been addressed since the works of \cite{lenth1981robust, ko1988robustness, he1992robust, Agostinelli2007}, mainly for univariate problems. 
A very first attempt to develop a robust parametric technique well suited for $p$-torus data and wrapped models can be found in \cite{saraceno2021robust}. A second approach has been discussed in \cite{greco2021robust}. They are both based on a set of weighted data-augmented estimating equations that are solved using a 
Classification Expectation-Maximization (CEM) algorithm, whose M-step is enhanced by the computation of a set of data dependent weights aimed to down-weight outliers. 

The main contributions of this paper can be summarized as follows. We generalize the approach developed in \cite{saraceno2021robust} building a set of weighted likelihood estimating equations \citep[WLEE,][]{markatou1998weighted} as weighted counterparts of the likelihood equations. 
The technique is developed in a very general framework for unimodal and elliptically symmetric distributions and not limited to the WN model.  
The resulting weighted likelihood estimator (WLE) can be evaluated according to different weighting schemes.
We shed new light on the nature, definition and treatment of torus outliers. In details, it is shown how the different approaches to evaluate weights can be justified in light of the current definition of outliers in use. 
We present and discuss a new strategy to obtain weights for robust fitting based on the unwrapped data, after imputing the vector of wrapping coefficients $\vect{j}$. 
It is shown that the estimating equations based on the unwrapped data can be properly used for sufficiently enough concentrated distributions on the torus. 
Furthermore, this work is meant to be a step forward the existing literature also because it is accompanied by formal theoretical results about the asymptotic behavior and the robustness properties of the proposed estimators. 

The remainder of the paper is organized according to the following structure. Some background on maximum likelihood estimation of wrapped models is given in Section \ref{sec:1}. The concept of outlyingness for torus data is discussed in Section \ref{sec:out}. Methods for weighted likelihood fitting are described in Section \ref{sec:2}.
Theoretical properties are discussed in Section \ref{sec:3}. 
Numerical studies are presented in Section \ref{sec:5}. Real data examples are given in Section \ref{sec:6}.
{\tt R} \citep{cran} code to run the proposed algorithms and replicate the real examples is available as Supplementary Material.

\section{Maximum likelihood estimation}
\label{sec:1}

Given an i.i.d sample $(\vect{y}_1, \vect{y}_2, \ldots, \vect{y}_n)$ from $\vect{Y} \sim m^\circ(\vect{y}; \vect{\theta})$, the maximum likelihood estimate (MLE) is obtained by maximizing the log-likelihood function 
\begin{equation}\label{ell}
\ell^\circ(\vect{\theta}) = \sum_{i=1}^n \log m^\circ(\vect{y}_i; \vect{\theta})
\end{equation}
or solving the corresponding set of estimating equations $\sum_{i=1}^n u^\circ(\vect{y}_i; \vect{\theta}) = \vect{0}$, where $$u^\circ(\vect{y}; \vect{\theta}) = \nabla_{\vect{\theta}} \log m^\circ(\vect{y}; \vect{\theta}) = \frac{\nabla_{\vect{\theta}} m^\circ(\vect{y}; \vect{\theta})}{m^\circ(\vect{y}; \vect{\theta})}$$ is the score function. For a wrapped unimodal elliptically symmetric model, i.e. given by wrapping (\ref{elliptic}) onto the $p$-torus, 
let 
\begin{equation} \label{vijMLE}
v_{i\vect{j}} = v_{i\vect{j}}(\vect{\mu}, \Sigma) = \frac{h^\prime(\vect{y}_i + 2 \pi \vect{j}; \vect{\mu}, \Sigma)}{\sum_{\vect{k}\in\mathbb{Z}^p} h(\vect{y}_i + 2 \pi \vect{k}; \vect{\mu}, \Sigma)}.
\end{equation}
Then, the MLE is the solution to the following fixed point equations
\begin{align} \label{MLE}
\vect{\mu} & = \frac{\sum_{i=1}^n \sum_{\vect{j} \in \mathbb{Z}^p} v_{i\vect{j}} (\vect{y}_i + 2\pi \vect{j})}{\sum_{i=1}^n \sum_{\vect{k} \in \mathbb{Z}^p} v_{i\vect{k}}} \\
\Sigma & = - \frac{2}{n} \sum_{i=1}^n \sum_{\vect{j} \in \mathbb{Z}^p} v_{i\vect{j}}(\vect{y}_i + 2\pi \vect{j}-\vect{\mu} )  (\vect{y}_i + 2\pi \vect{j}-\vect{\mu} )^\top \nonumber \ .
\end{align}
The reader is pointed to Appendix \ref{appendix:mle} for details. Finding the MLE requires an iterative procedure alternating between the computation of (\ref{vijMLE}) based on current parameters values and finding the (updated) solution to (\ref{MLE}). 
An approximate MLE can be obtained using crispy assignments after the computation of (\ref{vijMLE}), that is 
we let 
\begin{equation} \label{Cstep}
\hat{\vect{j}}_i= \argmax_{\vect{j} \in \mathbb{Z}^p} v_{i\vect{j}}
\end{equation}
and solve the estimating equation
\begin{equation} \label{CEM}
\sum_{i=1}^n u(\hat{\vect{x}}_{i}; \vect{\theta}) = \vect{0} 
\end{equation}
based on the \textit{unwrapped} (fitted) linear data $\hat{\vect{x}}_i = \vect{y}_i + 2\pi \hat{\vect{j}}_i$. 

In the special situation given by the WN, the derivation of the MLE through the fixed point equations in (\ref{MLE}) coincides with that obtained from an Expectation-Maximization (EM) algorithm based on a data augmentation procedure \citep{Fisher1994, coles1998inference, jona2012spatial, nodehi2020}. In a similar fashion, the approximate MLE can be obtained from a Classification EM (CEM) algorithm \citep{nodehi2020}. See Appendix \ref{appendix:em}.

\begin{remark}
	The infinite sum over $\mathbb{Z}^p$ makes likelihood inference challenging and hence it is common to replace it by a sum over the Cartesian product $\mathcal{C}_J=\vect{\mathcal{J}}^p$ where $\vect{\mathcal{J}} = (-J, -J+1, \ldots, 0, \ldots, J-1, J)$ for some $J$ providing a good approximation, since the summands of the series converge to zero. The approximation based on the truncated series works when 
	\begin{equation*}
	\Pr\left\{(\vect{Y} - \vect{\mu}) \in (-2 \pi J, 2 \pi J]^p\right\} \le \sum_{d=1}^p \Pr\left\{(Y_k - \mu_k) \in (-2 \pi J, 2 \pi J]\right\}
	\end{equation*}
	is negligible; this is the case when $\left(\mu_d - 4 \Sigma_{dd}^{1/2}, \mu_d + 4 \Sigma_{dd}^{1/2}\right) \subseteq (-2 \pi J, 2 \pi J]$, for $d=1, 2, \ldots, p$ \citep[see also][]{kurz2014efficient}. Actually, in case of the wrapped elliptically symmetric family, the density in (\ref{dens}) tends to that of a uniform distribution as concentration decreases \citep[see also][for the WN case]{mardia2000directional}.
\end{remark}
As noticed in \citep{nodehi2020}, the MLE for location is equivariant under affine transformation of the data in the original (unwrapped) linear space. On the contrary, this is not the case for the scatter matrix estimates. 
Furthermore, it is worth to remark that solving (\ref{CEM}) does not lead to consistent estimates for $\Sigma$ since the $\hat{\vect{j}}_i$ can not be a consistent estimates of the unknown wrapping coefficients. Therefore, there is lack of consistency for $\hat{\vect{x}}_i$, as well. The population estimating equation
\begin{equation} \label{em:population}
\int_{\mathbb{T}^p} u^\circ(\vect{y}; \vect{\mu}, \Sigma) m^\circ(\vect{y}; \vect{\mu_0}, \Sigma_0) \ d\vect{y} = \vect{0}
\end{equation}
is solved by the true values $(\vect{\mu_0}, \Sigma_0)$, hence making the MLE estimator Fisher consistent. In contrasts, the estimating equation (\ref{em:population}) is not the population estimating equations corresponding to (\ref{CEM}). Actually, we can always re-express our observations so that $\vect{z_i} = \vect{y_i} - \vect{\mu} \in (-\pi, \pi]^p$. It is not difficult to see that $\hat{\vect{x}}_i = \vect{z}_i$. Then, the distribution from which the $\hat{\vect{x}}_i$s are sampled is not $m(\vect{x}; \vect{\mu_0}, \Sigma_0)$. However, the distribution is still elliptically symmetric around $\vect{\mu}$ and its support is any hyper-cube of length $2\pi$ and in particular we can take $T(\vect{\mu}) = \times_{k=1}^p(\mu_k -\pi, \mu_k+\pi]$. We call this distribution the unwrapped model and we denote it by
\begin{equation*}    
m^u(\vect{x}; \vect{\mu_0}, \Sigma_0) = m^\circ(\vect{x}; \vect{\mu_0}, \Sigma_0) \mathbb{I}(\vect{x} \in T(\vect{\mu}_0)).
\end{equation*}  
Now, we can define $\Sigma^u_0$ as the solution to the CEM population estimating equation
\begin{equation} \label{cem:population}
\int_{\mathbb{R}^p} u(\vect{x}; \vect{\mu}, \Sigma) m^u(\vect{x}; \vect{\mu_0}, \Sigma_0) \ d\vect{x} = \vect{0} \ .
\end{equation} 
For illustrative purposes, let us consider the following univariate examples.
In Figure \ref{denunwrapped} we compare the unwrapped normal density $m^u(x; 0,\sigma_0^2)$ with the original normal density $m(x; 0, \sigma_0^2)$, for $\sigma_0=3\pi/8 \approx 1.178$ (left panel) and $\sigma_0=\pi/2 \approx 1.571$ (middle panel). We find that $\sigma_0^u \approx 1.163$ and $\sigma_0^u \approx 1.460$ respectively. For small values of $\sigma_0$ the two densities are very similar apart from the truncation of the tails in the range $(-\pi,\pi]$. On the opposite, the difference becomes marked for large values of $\sigma_0$  The relation between $\sigma_0$ and $\sigma_0^u$ is displayed in the right panel of Figure \ref{denunwrapped}. It follows that (\ref{CEM}) can be safely used for $\sigma\leq \pi/2$. However, in most practical cases, distributions characterized by large concentrations are not of interest and the identification of outliers become unfeasible, as  already discussed in Section \ref{sec:0}.

\begin{figure*}[t]
	\centering
	\includegraphics[width=0.3\textwidth]{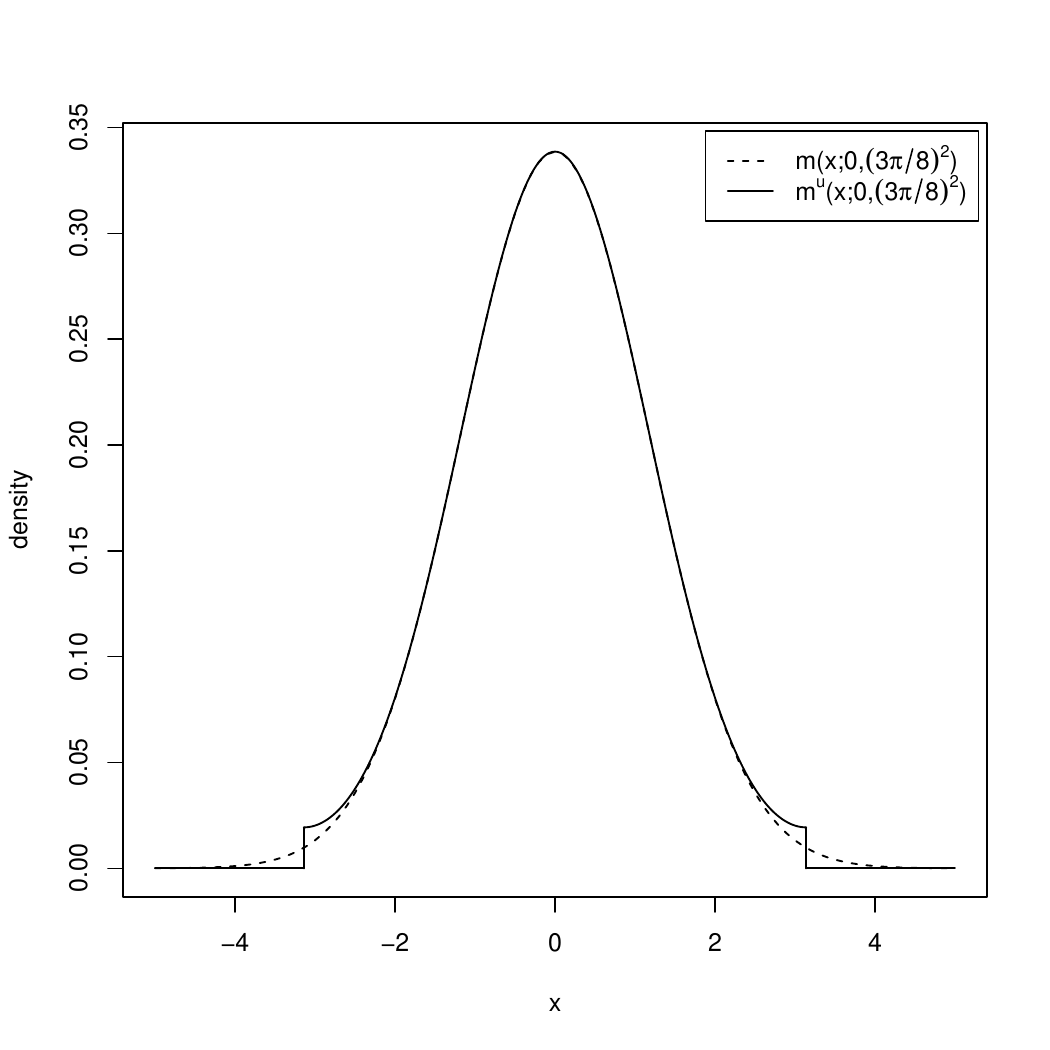}
	\includegraphics[width=0.3\textwidth]{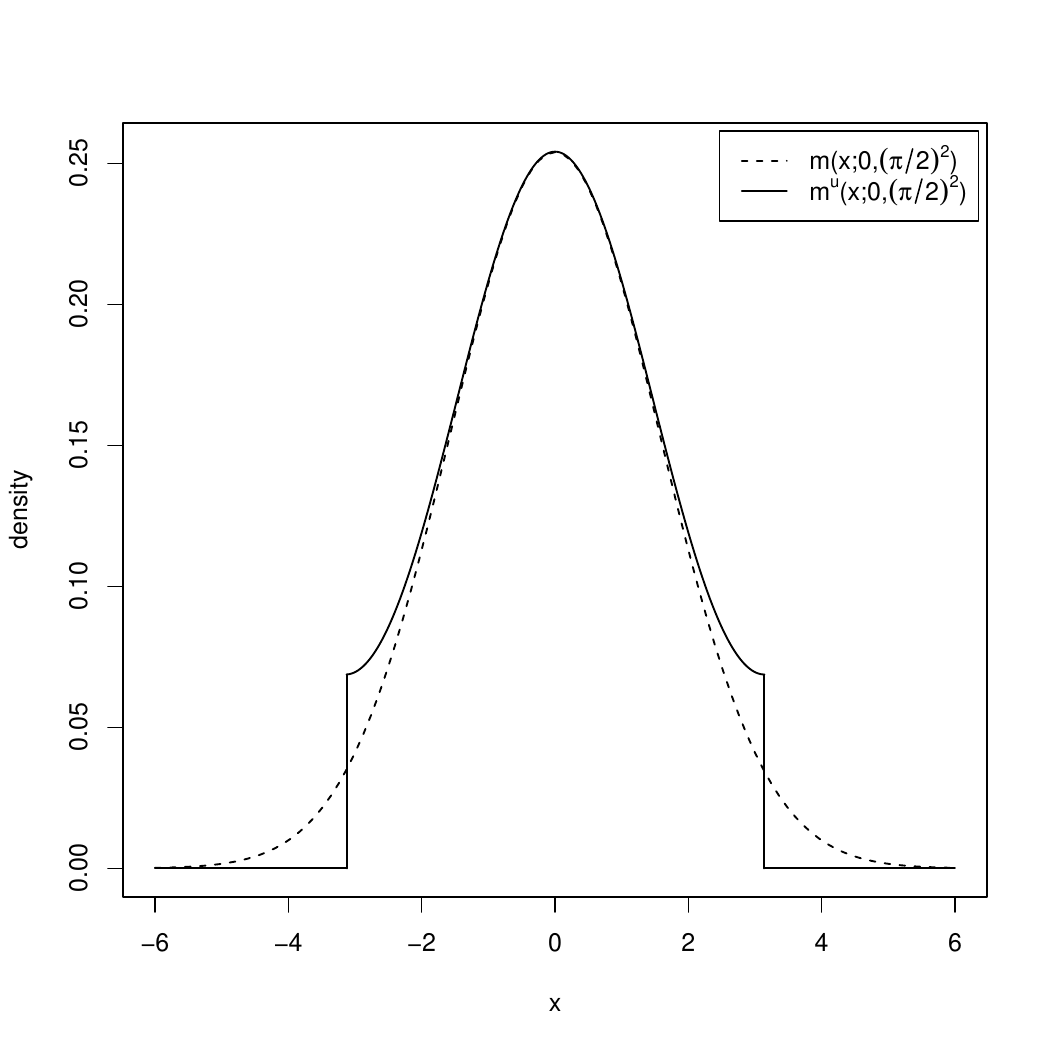}
	\includegraphics[width=0.3\textwidth]{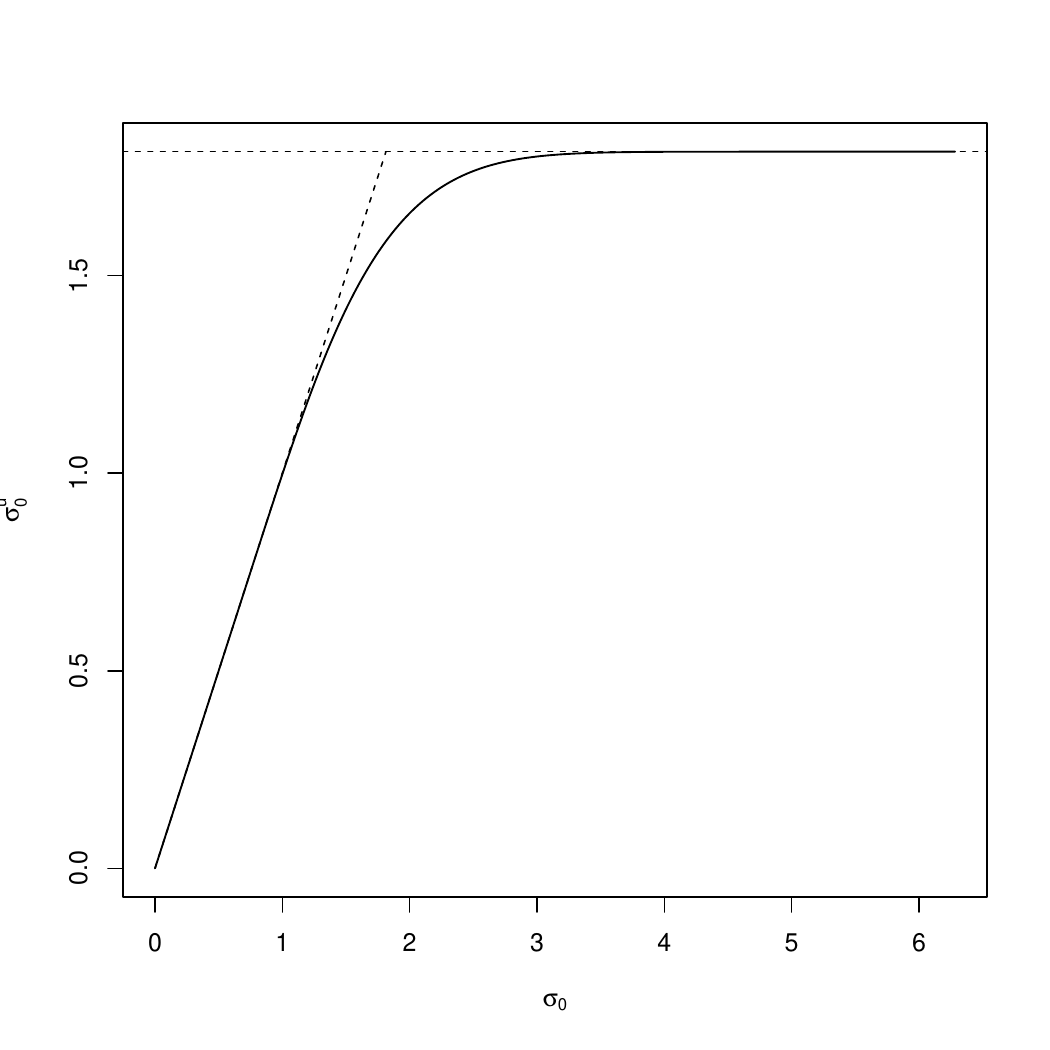}  
	\caption{Unwrapped normal density $m^u(x; 0, \sigma_0^2)$ compared with the original normal density $m(x; 0, \sigma_0^2)$, $\sigma_0=3 \pi/8$, (left panel), $\sigma_0=\pi/2$ (middle panel); $\sigma_0^u$ versus $\sigma_0$ (right panel).} 
	\label{denunwrapped}
\end{figure*}

\section{Outlyingness of torus data}
\label{sec:out}
We distinguish at least two approaches in the definition of outliers. The probabilistic approach is based on the idea that outliers are values {\it that are highly unlikely to occur under the assumed model} \citep{markatou1998weighted, Agostinelli2007}. Under this perspective, outlyingness can be measured according to the degree of agreement between the data and the assumed model, as provided by the Pearson residual \citep{lindsay1994}. In contrasts, according to the geometric approach, outliers are observations {\it which deviate from the pattern set by the majority of the data} \citep{huber2robust, rousseeuw2011robust} with respect to a geometric distance. However, it is not straightforward to define and measure geometric distances on the torus \citep{mardia2012statistics}. This makes the  probabilist point of view very appealing in this framework.

A simple but effective way to introduce outliers on the torus is that of considering the classical gross error model \citep{huber2robust} on the unwrapped linear space. Let $0 \le \epsilon < 0.5$ and $g(\vect{x})$ be an arbitrary density function. Then, the {\it true} density on the Euclidean space is
\begin{equation}\label{gol}
f(\vect{x})= (1-\epsilon)m(\vect{x}; \vect{\theta})+ \epsilon g(\vect{x}) \ ,
\end{equation}
whereas, on the torus, we have that
\begin{align}
f^\circ(\vect{y}) & = \sum_{\vect{j} \in \mathbb{Z}^p} f(\vect{y}+ 2\pi \vect{j}) \nonumber \\
& = (1-\epsilon) \sum_{\vect{j} \in \mathbb{Z}^p} m(\vect{y} + 2\pi \vect{j}; \vect{\theta}) + \epsilon \sum_{\vect{j} \in \mathbb{Z}^p} g(\vect{y} + 2\pi \vect{j}) \nonumber\\
& = (1-\epsilon) m^\circ(\vect{y}; \vect{\theta}) + \epsilon g^\circ(\vect{y}). \label{equ:probabilistic}
\end{align}

A measure of the agreement between the true and assumed model on the probabilistic ground is provided by the Pearson residual function \citep{lindsay1994, basu1994minimum, markatou1998weighted}.
Let $K_H(\vect{y})$ be a smooth family of (circular) kernel functions with bandwidth matrix $H$.
Let  $\hat{f}^\circ(\vect{y})$ and $\hat{m}^\circ(\vect{y}; \vect{\theta})$ be smoothed densities, obtained by convolution between $K_H(\vect{y})$ and $f^\circ(\vect{y})$ and $m^\circ(\vect{y}; \vect{\theta})$, respectively. In \cite{saraceno2021robust} it has been suggested to measure the outlyingness of torus data based on (\ref{equ:probabilistic}) and 
using the Pearson residual function defined on $\vect{y} \in \mathbb{T}^p$ as
\begin{equation}
	\label{pearson}
	\delta^\circ (\vect{y}; \vect{\theta}) = \frac{\hat{f}^\circ(\vect{y})}{\hat{m}^\circ(\vect{y}; \vect{\theta})} - 1 \ ,
\end{equation}
with $\delta^\circ(\vect{y}; \vect{\theta})\in [-1,+\infty)$, see also \citep{Agostinelli2007}.
The same probabilistic definition of outliers can be applied on the unwrapped linear space rather than on the torus, in a dual fashion. Therefore, in a CEM-based framework, one can define outlyingness on the unwrapped rather than circular data, based on (\ref{gol}). Actually, for a given $\vect{x} \in \mathbb{R}^p$, one can define the Pearson residual function
\begin{equation}
\label{pearson:lin}
\delta(\vect{x}; \vect{\theta}) = \frac{\hat{f}(\vect{x})}{\hat{m}(\vect{x}; \vect{\theta})} - 1 \ ,
\end{equation}
where $\hat{f}(\vect{x})$ and $\hat{m}(\vect{x}; \vect{\theta})$ are linear smoothed model densities.
However, according to the results stated in Section \ref{sec:1}, the use of a C-step does not lead to observe data directly from $m(\vect{x}; \vect{\theta})$ but from the wrapped-unwrapped mechanism $m^u(\vect{x}; \vect{\theta})$.  Then, it would be correct to consider the Pearson residual function
\begin{equation}
\label{pearson:unw}
\delta^u(\vect{x}; \vect{\theta}) = \frac{\hat{f}^u(\vect{x})}{\hat{m}^u(\vect{x}; \vect{\theta})} - 1 \ ,
\end{equation}
instead, with $\delta^u(\vect{x}; \vect{\theta})\in [-1,+\infty)$. 

Large Pearson residuals detect points in disagreement with the model. This points are supposed to be down-weighted in the estimation process using a proper weighting function. The evaluation of a proper set of weights requires measuring the outlyingness of each data point with respect to a given (robust) fit of the postulated model. Based on the weighted likelihood methodology \citep{markatou1998weighted}, the weights are obtained from the finite sample counterparts of the Pearson residuals defined in (\ref{pearson}) or (\ref{pearson:unw}). In the former case, we have 
\begin{equation}
\label{residualfs}
\delta_n^\circ(\vect{y}; \vect{\theta}) = \frac{\hat f_n^\circ(\vect{y})}{\hat m^\circ(\vect{y}; \vect{\theta})} - 1 \ ,
\end{equation}
where $\hat f_n^\circ(\vect{y})$ is a circular kernel density estimate on the torus. As well, in the case of unwrapped data, we have that
\begin{equation}
\label{residualfs2}
\delta_n^u(\vect{x}; \vect{\theta}) = \frac{\hat f_n^u(\vect{x})}{\hat m^u(\vect{x}; \vect{\theta})} - 1 \ 
\end{equation}
where $\hat f_n(\vect{x})$ is a kernel density estimate evaluated on the hyperplane over the fitted unwrapped (complete) data $(\hat{\vect{x}}_1, \ldots, \hat{\vect{x}}_n)$. 
In practice, for concentrated circular distributions, the Pearson residuals in (\ref{residualfs2}) can be approximated by 
\begin{equation}
\label{residualfs2bis}
\delta_n(\vect{x}; \vect{\theta}) = \frac{\hat f_n^u(\vect{x})}{\hat m(\vect{x}; \vect{\theta})} - 1 \ .
\end{equation}
Smoothing the model makes the Pearson residuals converge to zero with probability one under the assumed model and it is not required that the kernel bandwidth goes to zero as the sample size increases \citep{markatou1998weighted}. In general, the choice of the kernel is not crucial. \\

\begin{remark}
	When the model is the multivariate WN distribution, we can use a multivariate WN kernel with covariance matrix $H=\operatorname{diag}(h^2)$, since the smoothed model density is still an element of the WN family with covariance matrix $\Sigma+H$.  
\end{remark}

\begin{remark}
	In practice, under the WN model, the distribution of the unwrapped data can be approximated by a multivariate normal variate for {\it concentrated} distributions, that is whenever all the variances are sufficiently {\it small}. In this case, using a multivariate normal kernel with bandwidth matrix $H=\operatorname{diag}(h^2)$ returns a smoothed model that is still normal with variance-covariance matrix $\Sigma+H$. It is worth to stress that the WN distribution inherits this property of closure with respect to convolution from the normal model. The closure to convolution property makes the use of the Gaussian kernel very appealing.
\end{remark}

\begin{remark}
The family of elliptical distributions is not closed under convolution. e.g. see Sec 5.3.4 of \citep{prestele2007}. However, some subfamilies of elliptical distributions are closed under convolution; for example the class of elliptical stable distributions are closed under convolutions.
\end{remark}
  
Despite several weight functions could be used, in the weighted likelihood methodology it is common to consider
\begin{equation}
\label{weightfun}
w(\delta) = \min\left\{ 1, \frac{\left[A(\delta) + 1\right]^+}{\delta + 1} \right\} \ ,
\end{equation}
where $w(\delta)\in [0, 1]$, $[\cdot]^+$ denotes the positive part and $A(\delta)$ is the Residual Adjustment Function (RAF, \cite{lindsay1994,basu1994minimum,park+basu+lindsay+2002}), whose
special role is related to the connections between weighted likelihood estimation and minimum disparity estimation. In practice, the RAF acts by bounding the effect of those points leading to large Pearson residuals. The function $A(\cdot)$ is assumed to be increasing and twice differentiable in $[-1, +\infty)$, with $A(0) = 0$ and $A^\prime(0) = 1$. The weights decline smoothly to zero as $\delta\rightarrow\infty$ (outliers) and depending on the RAF also as $\delta\rightarrow -1$ (inliers). In particular, the weight function (\ref{weightfun}) can involve a RAF based on the Symmetric Chi-squared divergence \citep{markatou1998weighted}, the family of Power divergences \citep{lindsay1994} or the Generalized Kullback-Leibler divergence \citep{park+basu+2003} \citep[see][for details]{saraceno2021robust}. 

\subsection{The geometric approach}
The probabilistic approach allows to identify outliers both on the torus or after unwrapping the data, in a purely dual fashion. On the other hand, the geometric approach can be used only in the latter situation, as described in \cite{greco2021robust}.
By exploiting the methodology developed in \cite{agostinelli2019weighted}, under the elliptically symmetric model in (\ref{ell}) and for a known wrapping coefficient vector $\vect{j}$, Pearson residuals and weights can be based on the squared Mahalanobis distance $d^2 = d^2(\vect{x}; \vect{\theta}) = [(\vect{x} - \vect{\mu})^\top \Sigma^{-1} (\vect{x} - \vect{\mu})]$. In particular, finite sample Pearson residuals are defined as 
\begin{equation}
\label{residualfsGunw}
\delta_n^{du}(\vect{x}; \vect{\theta}) = \frac{\hat f_n^u(d^2)}{\chi_u^2(d^2; p)} - 1 \ ,
\end{equation}
where $\hat f_n(d^2)$ is a (unbounded at the boundary) kernel density estimate evaluated over squared Mahalanobis distances $d^2(\hat{\vect{x}}; \hat{\vect{\theta}})$ and 
$\chi_u^2(d^2; p)$ is the density of the Mahalanobis distance evaluated under the wrapped-unwrapped model $m^u(\cdot; \vect{\theta})$. For concentrated circular distributions, the Pearson residual in (\ref{residualfsGunw}) can be approximated by 
\begin{equation}
\label{residualfsG}
\delta_n^d(\vect{x}; \vect{\theta}) = \frac{\hat f_n(d^2)}{\chi^2(d^2; p)} - 1 \ ,
\end{equation}
where $\chi^2(\cdot; p)$ denotes the (asymptotic) distribution of Mahalanobis distances for the original linear data.
Figure \ref{chiunwrapped} shows two examples of $\chi_u^2(d^2; p)$ for $p=6$ when $\sigma_0 = 3\pi/8$ (left panel) and $\sigma_0=\pi/2$ (right panel). In the first case the support of the distribution is the interval $[0, 42.\bar{6})$ while in the second case is the interval $[0,24)$.

\begin{figure*}[t]
	\centering
	\includegraphics[width=0.45\textwidth]{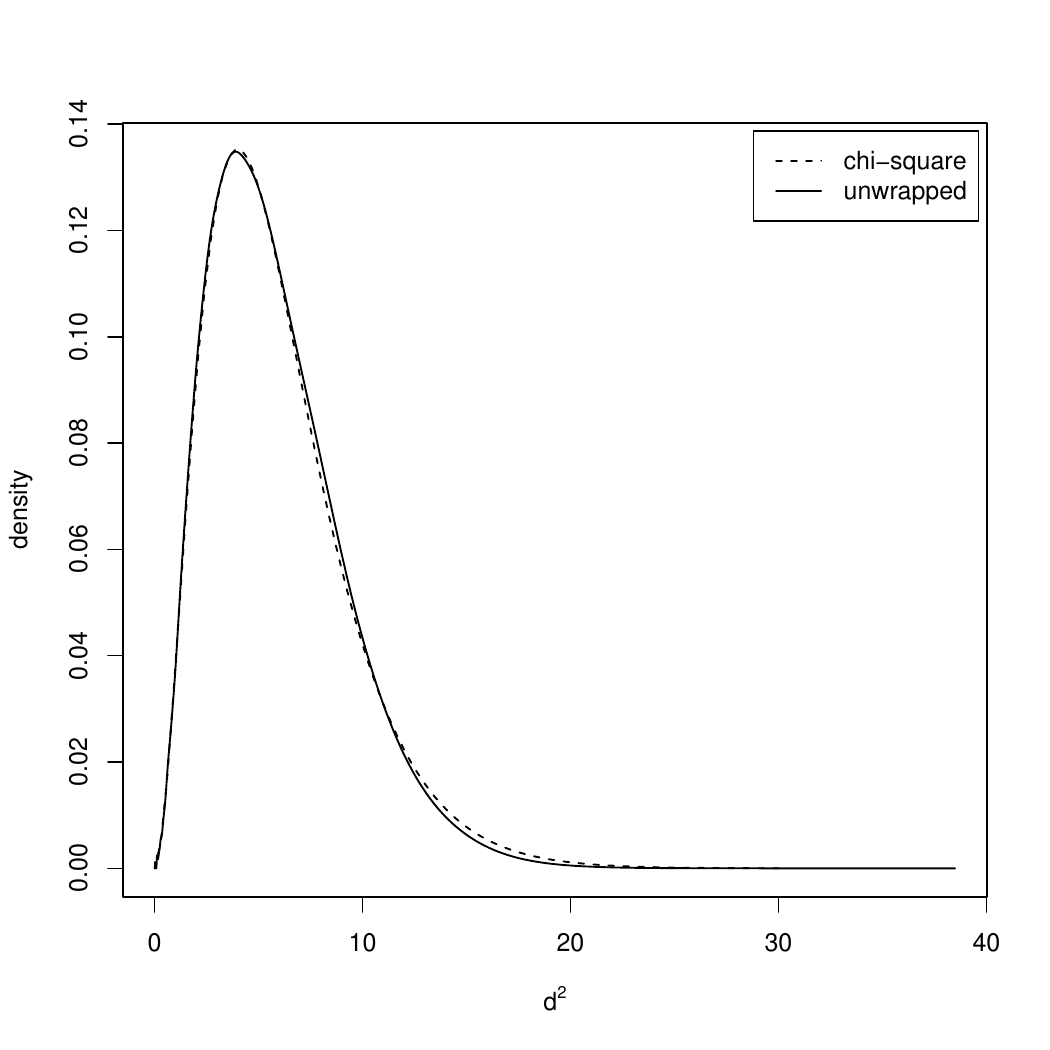}
	\includegraphics[width=0.45\textwidth]{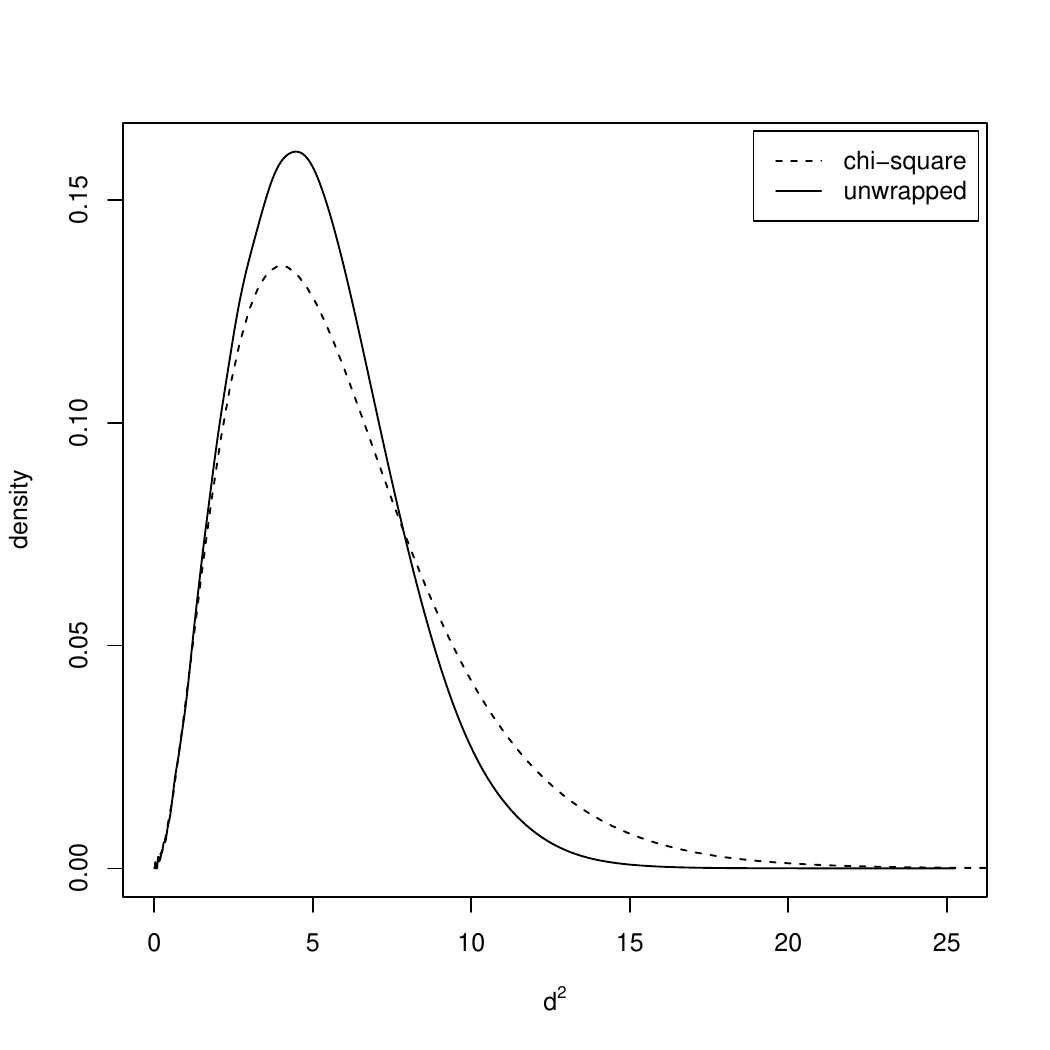}
	\caption{Distribution of the squared Mahalanobis distance for the unwrapped observations from a wrapped normal model with $\sigma_0=3 \pi/8$, (left panel) and $\sigma_0=\pi/2$ (right panel).} 
	\label{chiunwrapped}
\end{figure*}

\section{Robust fitting based on WLEE}
\label{sec:2}

Robust fitting of a multivariate wrapped unimodal elliptically symmetric model to torus data can be achieved according to a weighted version of the population estimating equations (\ref{em:population}), i.e.,
\begin{equation} \label{wem:population}
\int_{\mathbb{T}^p} w^\circ(\vect{y}) u^\circ(\vect{y}; \vect{\mu}, \Sigma) m^\circ(\vect{y}; \vect{\mu_0}, \Sigma_0) \ d\vect{y} = \vect{0}
\end{equation}
where the weight function is given by $w^\circ(\vect{y}) = w\left(\delta^\circ(\vect{y}; \vect{\theta}) \right)$.
We notice that $w^\circ(\vect{y})$ is a periodic function, i.e., $w^\circ(\vect{y}) = w^\circ(\vect{y} + 2\pi\vect{j})$, $\vect{j} \in \mathbb{Z}^p$.  The sample version of (\ref{wem:population}), that is
\begin{equation*}
  \sum_{i=1}^n w(\delta_n^\circ(\vect{y}_i; \vect{\theta})) u^\circ(\vect{y}_i; \vect{\theta}) = \vect{0} \ ,
\end{equation*}
specializes to the following WLEE for unimodal elliptically symmetric distributions 
\begin{align} \label{WLE}
\vect{\mu} & = \frac{\sum_{i=1}^n w(\delta_n^\circ(\vect{y}_i)) \sum_{\vect{j} \in \mathbb{Z}^p} v_{i\vect{j}} (\vect{y}_i + 2\pi \vect{j})}{\sum_{i=1}^n  w(\delta_n^\circ(\vect{y}_i)) \sum_{\vect{k} \in \mathbb{Z}^p} v_{i\vect{k}}} \\
\Sigma & = - \frac{2}{\sum_{i=1}^n  w(\delta_n^\circ(\vect{y}_i))} \sum_{i=1}^n  w(\delta_n^\circ(\vect{y}_i)) \sum_{\vect{j} \in \mathbb{Z}^p} v_{i\vect{j}}(\vect{y}_i + 2\pi \vect{j}-\vect{\mu}) (\vect{y}_i + 2\pi \vect{j}-\vect{\mu} )^\top \nonumber .
\end{align}
with $w(\delta_n^\circ(\vect{y}_i)) = w(\delta_n^\circ(\vect{y}_i; \vect{\mu}, \Sigma))$. 
The WLEE can be solved by a suitable modification of the iterative procedure depicted in Section \ref{sec:1} to find the MLE. 
At iteration $(s)$, based on current $v_{i\vect{j}}^{(s)}$ obtained as in (\ref{vijMLE}),  
a set of data dependent weights $w_i^{(s)}=w\left(\delta^\circ_n(\vect{y}_i;\vect{\mu}^{(s)},\Sigma^{(s)})\right)$ is computed, whose effect is that of down-weighting the contribution of those points with large Pearson residuals based on the current fit. Then, updated estimates from iteration $(s)$ to $(s+1)$ are obtained by solving the WLEE in (\ref{WLE}).
In practice, the summation over $\mathbb{Z}^p$ is replaced by a summation over $C_J$.

According to a similar reasoning, we can consider a weighted counterpart of the population estimating equation (\ref{cem:population}), that is
\begin{equation} \label{wcem:population}
\int_{\mathbb{R}^p} w(\vect{x}) u(\vect{x}; \vect{\mu}, \Sigma) m^u(\vect{x}; \vect{\mu_0}, \Sigma_0) \ d\vect{x} = \vect{0} \ .
\end{equation}
 We notice that, in this situation, the use of (\ref{pearson}) or (\ref{pearson:lin}) leads to the same estimator. Hence, one can build a WLEE based on the fitted unwrapped linear data $\hat{\vect{x}_i}$, with weights whose evaluation can be now based on  (\ref{residualfs}), (\ref{residualfs2}) or (\ref{residualfsGunw}).  
At iteration $(s)$, estimates are updated according to
\begin{align} \label{WCEM}
  \hat{\vect{\mu}}^{(s+1)} & = \frac{\sum_{i=1}^n w_i^{(s)} \dot{h}_i^{(s)} \hat{\vect{x}}_i^{(s)}}{\sum_{i=1}^n w_i^{(s)} \dot{h}_i^{(s)}} \\
\hat\Sigma^{(s+1)} & = - \frac{2}{\sum_{i=1}^n w_i^{(s)}} \sum_{i=1}^n w_i^{(s)}  \dot{h}_i^{(s)} \left(\hat{\vect{x}}_i^{(s)}-\hat{\vect{\mu}}^{(s+1)}\right)\left(\ \hat{\vect{x}}_i^{(s)}-\hat{\vect{\mu}}^{(s+1)}\right)^\top  \nonumber .
\end{align}
where $ \dot{h}_i^{(s)} = h^\prime(d(\hat{\vect{x}}_i^{(s)};\hat{\vect{\mu}}^{(s)}, \hat{\Sigma}^{(s)}))/h(d(\hat{\vect{x}}_i^{(s)};\hat{\vect{\mu}}^{(s)}, \hat{\Sigma}^{(s)}))$.
We stress that the derivation of the WLEE for torus data generalizes the approach introduced in \cite{saraceno2021robust}, that was confined to a data augmentation perspective rather than on genuine maximum likelihood estimation. Therefore, here it is possible to derive a WLE that is the weighted counterpart of the MLE (and of its approximated version) and we are not limited to a CEM-type algorithm. 

\begin{remark}
	For a fixed bandwidth matrix $H$, the newly established weighting approach based on (\ref{residualfs2}) requires that a multivariate kernel density estimate is computed at each iteration. The same is also true when using the weights in (\ref{residualfsGunw}). In contrasts, the procedure based on (\ref{residualfs}) requires the evaluation of a more demanding torus kernel density estimate only once. However, computing a new kernel density estimate for linear data at each iteration adds no computational burden. 
\end{remark}

\subsection{Bandwidth selection}
The finite sample robustness of the WLE depends on the selection of the smoothing parameter $h$, whatever the type of Pearson residuals among those listed above. 
Large values of $h$ lead to smooth kernel density estimates that are stochastically close to the postulated model. As a result, Pearson residuals are all close to zero, weights all close to one, the WLE gains efficiency at the model but is less robust. On the opposite, small values of $h$ make the kernel estimate more sensitive to the occurrence of outliers. Then, Pearson residuals become large where the data are in disagreement with the model and such points are properly down-weighted: the WLE looses efficiency at the model but recover robustness to outliers contamination.

The selection of $h$ is still an open issue in weighted likelihood estimation. From a practical point of view, selecting a too small value for $h$ can lead to an undue excess of down-weighting and hide relevant features in the data. In contrasts, a too large value could provide an insufficient down-weighting and misleading estimates, as well as the MLE. One strategy relies on a monitoring approach \citep{agostinelli2018discussion, greco2020weighted, greco2020weighted1} in the selection of the bandwidth. It is suggested
to run the procedure for different values of the smoothing parameter $h$ and monitor the behavior of estimates and/or weights as $h$ varies in a reasonable range. Monitoring the weights as $h$ varies is expected to 
describe a transition from a robust to a non robust fit, since for increasing values of $h$ all the weights approach one and the methodology does not allow to discriminate between the genuine part of te data and the outliers, anymore. As well, one can monitor a summary of the weights, such as the empirical down-weighting level $1- \bar w$, where $\bar w$ denotes the average of the weights. It can be considered as a rough estimate of the amount of down-weighting.
The approach of monitoring unveils patterns and substructures otherwise hidden that can aid the comprehension of the phenomenon under study and the sources of contamination. 

\subsection{Initialization}
The iterative algorithm to solve the WLEE in (\ref{WLE}) or (\ref{WCEM}) can be initialized using subsampling. 
The subsample size is expected to be as small as possible in order to increase the probability to get an outliers free initial subset but large enough to guarantee estimation of the unknown parameters. 

The initial value for the mean vector $\vect{\mu}$ is set equal to the circular sample mean. Initial diagonal elements of $\Sigma$ can be obtained as $\Sigma^{(0)}_{rr}=-2\log(\hat{\rho}_r)$, where $\hat{\rho}_r$ is the sample mean resultant length, whereas its off-diagonal elements are given by
$\Sigma^{(0)}_{rs}=\rho_c(\vect{y}_r, \vect{y}_s) \sigma_{rr}^{(0)} \sigma_{ss}^{(0)}$ ($r \neq s$), where $\rho_c(\vect{y}_r, \vect{y}_s)$ is the circular correlation coefficient, $r,2=1,2,\ldots,p$ \citep{JammalamadakaSenGupta2001}.
It is suggested to run the algorithm from several starting points.
The {\it best} solution can be selected by minimizing the probability to observe a small Pearson residual \citep{agostinelli2019weighted, saraceno2021robust}.
According to the experience of the authors, a small number of subsamples is sufficient and very often they led to the same solution.

\subsection{Outliers detection}
\label{sec:4}
The objective of a robust analysis is twofold: from the one hand we protect model fitting from the adverse effect of anomalous values, from the other hand it is of interest to provide effective tools to identify outliers based on formal rules and the robust fit. The process of outliers detection allows to investigate deeply their source and nature and unveil hidden and unexpected sub-structures in the data that are worth studying and may not have been considered otherwise \citep{farcomeni2016robust}. 
The inspection of weights provides a first approach for the task of outliers detection: points whose weight is below a fixed, and opportunely low, threshold (see also \cite{greco2020weighted} in a different framework) could be declared as outlying. 
However, it would be desirable to base outliers detection on an appropriate statistic to test outlyingness of each data point. 
In this respect, at least when robust fitting relies on (\ref{WCEM}), it is suggested to 
build a decision rule based  on the fitted unwrapped linear data at convergence, treating them as a proper sample from a multivariate {\it linear} variate with density function as in (\ref{elliptic}). This approximation is supposed to work as long as torus data show a sufficiently high concentrated distribution.
Therefore, one can pursue outliers detection looking at the squared robust  distances $d^2(\hat{\vect{x}_i}; \hat\theta)$. Outlying data are those whose distance exceeds a fixed threshold corresponding to the $(1-\alpha)$-level quantile of a chi-square distribution with $p$ degrees of freedom \citep{greco2021robust}.

\section{Properties}
\label{sec:3}
Here, the asymptotic behavior of the proposed estimators and their robustness properties are investigated. The reader is pointed to \cite{agostinelli2019weighted} for details on the asymptotic behavior of the WLE in a general setting. 
Hereafter, we assume broad regularity conditions  for consistency and asymptotic normality of the MLE to hold.

\subsection{Asymptotic distribution under the model}
The following Lemma give the conditions to ensure the required asymptotic behavior of the Pearson residuals in (\ref{residualfs}), (\ref{residualfs2}) and (\ref{residualfsGunw}) and the corresponding weights at the assumed model. Henceforth, $\hat{f}(\vect{y}) = \hat{m}(\vect{y}, \vect{\theta}_0)$ (a.s.) and
\begin{equation*}
\delta(\vect{y};\vect{\theta}) = \frac{\hat{m}(\vect{y}, \vect{\theta}_0)}{\hat{m}(\vect{y}, \vect{\theta})} - 1
\end{equation*}
where $\hat{m}(\vect{y}; \vect{\theta}) = \int k(\vect{y}-\vect{t}) m^*(\vect{t}; \vect{\theta}) \ d\vect{t}$ is the smoothed model involved in the definition of Pearson residuals in use, i.e. $m^*(\vect{y})$ can be $m^\circ(\vect{y})$, $m^u(\vect{x})$ or $\chi_u^2(d^2)$, respectively. Moreover, let $\delta_n$ be the Pearson residuals defined as either in (\ref{residualfs}), (\ref{residualfs2}) or (\ref{residualfsGunw}),  
and $\hat{f}_n$ be a kernel density estimator with kernel $K_H(\cdot)$ and bandwidth  matrix $H$, corresponding to $\hat{f}_n^\circ$, $\hat{f}_n^u$ or $\hat{f}_n^{du}$, respectively, according to the definition of $\delta_n$ in use.

\begin{lemma} \label{lemma:deltan}
Assume that: (i) the kernel $K_H(\cdot)$ is of bounded variation; (ii) the model is correctly specified, that is, there exists $\vect{\theta}_0 \in \Theta$ such that $f^\circ(\vect{y}) = m^\circ(\vect{y}; \vect{\theta}_0)$ (a.s.); (iii) the model density is positive over the support $ \mathcal{Y}$, that is, there exists $K > 0$ such that $\sup_{\vect{y} \in \mathcal{Y}, \vect{\theta} \in \Theta} m^\circ(\vect{y}; \vect{\theta}) \geq K$; (iv) $A(0)=0$, $A^\prime(0)=1$ and $A^{\prime\prime}(\delta)$ is a bounded and continuous function w.r.t. $\delta$. Then
\begin{align*}
\sup_{\vect{y} \in \mathcal{Y}, \vect{\theta} \in \Theta} \vert \delta_n(\vect{y};\vect{\theta}) - \delta(\vect{y};\vect{\theta}) \vert \stackrel{a.s.}{\rightarrow} 0 \\
\sup_{\vect{y} \in \mathcal{Y}, \vect{\theta} \in \Theta} \vert w (\delta_n(\vect{y};\vect{\theta})) -  w (\delta(\vect{y};\vect{\theta})) \vert \stackrel{a.s.}{\rightarrow}  0 \\
\sup_{\vect{y} \in \mathcal{Y}, \vect{\theta} \in \Theta} \vert w^\prime(\delta_n(\vect{y};\vect{\theta})) -  w^\prime(\delta(\vect{y};\vect{\theta})) \vert \stackrel{a.s.}{\rightarrow}  0 \ .
\end{align*}
\end{lemma}

\begin{proof}
 Under assumptions (i) and (ii) we have that $\hat{f}_n(\vect{y}) \stackrel{a.s.}{\rightarrow} \hat{m}(\vect{y}; \vect{\theta}_0)$ uniformly w.r.t. $\vect{y}$ as a result of the Glivenko-Cantelli theorem \citep{rao2014nonparametric}. Under (iii) we obtain
\begin{align*}
\sup_{\vect{y} \in \mathcal{Y}, \vect{\theta} \in \Theta} \left\vert \delta_n(\vect{y}; \vect{\theta}) - \delta(\vect{y}; \vect{\theta}) \right\vert & = \sup_{\vect{y} \in \mathcal{Y}, \vect{\theta} \in \Theta} \left\vert \frac{\hat{f}(\vect{y}) - \hat{m}(\vect{y}; \vect{\theta}_0)}{\hat{m}(\vect{y}; \vect{\theta})}  \right\vert \\
& \leq \frac{\sup_{\vect{y} \in \mathcal{Y}, \vect{\theta} \in \Theta} \left\vert \hat{f}(\vect{y}) - \hat{m}(\vect{y}; \vect{\theta}_0) \right\vert}{K}   \\
& \stackrel{a.s.}{\rightarrow} 0 \ .
\end{align*}
The second and third statements follows from equation (\ref{weightfun}), assumption (iv) and the continuous mapping theorem.
\end{proof}

\begin{remark}
	Assumption (iii) in Lemma \ref{lemma:deltan} is plausible in the case of toroidal densities. It allows to relax the mathematical device of evaluating the supremum of the Pearson residuals, since it avoids the occurrence of small (almost null) densities in the tails that would affect the denominator of Pearson residuals \cite{agostinelli2019weighted}. It is satisfied for wrapped models obtained from (\ref{elliptic}) under e.g. the assumption that $h(\cdot)$ is strictly positive in the hyper-cube $\times_{i=1}^p (\mu_i - \pi, \mu_i + \pi]$ and $\Sigma$ is positive definite.    
\end{remark}

\begin{lemma}  \label{wconsistency1}
Assume that for all $\vect{y}$ and $\vect{\theta}$, $\Psi(\vect{y};\vect{\theta}) = w(\delta(\vect{y};\vect{\theta})) u(\vect{y};\vect{\theta})$ is differentiable and the matrix $\dot{\Psi}(\vect{y};\vect{\theta})$ with elements $i,j$ be $\partial \Psi_i/\partial \theta_j$ is positive definite and $\E_{\vect{\theta}_0}(\dot{\Psi}(\vect{Y};\vect{\theta}))$ is finite, then
\begin{enumerate}
\item[i.] for every $n$, if there exists a solution $\check{\vect{\theta}}_n$ of $\sum_{i=1}^n \Psi(\vect{Y}_i;\vect{\theta}) = \vect{0}$ this solution is unique;
\item[ii.] let $\check{\vect{\theta}}_n$ be the sequence of solutions, then $\check{\vect{\theta}}_n \stackrel{a.s.}{\rightarrow} \vect{\theta}_0$ as $n \rightarrow \infty$.
\end{enumerate}
\end{lemma}

\begin{proof}
  Part \textit{i.} is an application of Theorem 10.9 in \cite{maronna2019robust}. For part \textit{ii.} notice that $\Psi(\vect{y};\vect{\theta}_0) = u(\vect{y};\vect{\theta}_0)$ and by a first order Taylor expansion around $\vect{\theta}_0$ of $\Psi(\vect{y};\vect{\theta})$ we have
\begin{equation*}
0 = \sum_{i=1}^n \Psi(\vect{Y}_i;\check{\vect{\theta}}_n) = \sum_{i=1}^n u(\vect{Y}_i;\vect{\theta}_0) + \sum_{i=1}^n \dot{\Psi}(\vect{Y}_i;\vect{\theta}_i) (\check{\vect{\theta}}_n - \vect{\theta}_0)
\end{equation*}
hence
\begin{equation*}
\check{\vect{\theta}}_n - \vect{\theta}_0 = \left[ \frac{1}{n} \sum_{i=1}^n \dot{\Psi}(\vect{Y}_i;\vect{\theta}_i)\right]^{-1} \frac{1}{n} \sum_{i=1}^n u(\vect{Y}_i;\vect{\theta}_0)
\end{equation*}
On the right hand side, the first term is bounded almost surely, while the second term goes to zero almost surely by the strong law of large numbers for i.i.d. random variables. Hence, $\check{\vect{\theta}}_n \stackrel{a.s.}{\rightarrow} \vect{\theta}_0$ as $n \rightarrow \infty$.
\end{proof}  

\begin{theorem}[Consistency] \label{wconsistency2}
Under the assumptions of Lemmas \ref{lemma:deltan} and \ref{wconsistency1}. Assume $\Psi_n(\vect{y};\vect{\theta}) = w(\delta_n(\vect{y};\vect{\theta})) u(\vect{y};\vect{\theta})$ is differentiable and the matrix $\dot{\Psi}_n(\vect{y};\vect{\theta})$ with elements $i,j$ be $\partial \Psi_{n,i}/\partial \theta_j$ is positive definite, then
\begin{enumerate}
\item[i.] for every $n$, if there exists a solution $\hat{\vect{\theta}}_n$ of $\sum_{i=1}^n \Psi_n(\vect{Y}_i;\vect{\theta}) = \vect{0}$ this solution is unique;
\item[ii.] let $\hat{\vect{\theta}}_n$ be the sequence of solutions, then $\hat{\vect{\theta}}_n \stackrel{a.s.}{\rightarrow} \vect{\theta}_0$ as $n \rightarrow \infty$.
\end{enumerate}
\end{theorem}

\begin{proof}
For each $n$ consider a first order Taylor expansion around $\check{\vect{\theta}}_n$ of $\Psi_n(\vect{Y}_i; \vect{\theta})$ and hence
\begin{equation*}
\sum_{i=1}^n \left( \Psi_n(\vect{Y}_i;\hat{\vect{\theta}}_n) - \Psi_n(\vect{Y}_i;\check{\vect{\theta}}_n) \right) = \sum_{i=1}^n \dot{\Psi}_n(\vect{Y}_i; \vect{\theta}_{n,i}) (\hat{\vect{\theta}}_n - \check{\vect{\theta}}_n) \ ,
\end{equation*}  
and since
\begin{align*}
\vect{0} = \sum_{i=1}^n \Psi_n(\vect{Y}_i; \hat{\vect{\theta}}_n) & = \sum_{i=1}^n \left( \Psi_n(\vect{Y}_i;\hat{\vect{\theta}}_n) - \Psi_n(\vect{Y}_i;\check{\vect{\theta}}_n) \right) + \sum_{i=1}^n \left( \Psi_n(\vect{Y}_i;\check{\vect{\theta}}_n) - \Psi(\vect{Y}_i; \check{\vect{\theta}}_n) \right) \\
& = \sum_{i=1}^n \dot{\Psi}_n(\vect{Y}_i; \vect{\theta}_{n,i}) (\hat{\vect{\theta}}_n - \check{\vect{\theta}}_n) + \sum_{i=1}^n \left( \Psi_n(\vect{Y}_i;\check{\vect{\theta}}_n) - \Psi(\vect{Y}_i; \check{\vect{\theta}}_n) \right) \ ,
\end{align*}
we have
\begin{equation*}
\hat{\vect{\theta}}_n - \check{\vect{\theta}}_n = - \left[ \frac{1}{n} \sum_{i=1}^n \dot{\Psi}_n(\vect{Y}_i; \vect{\theta}_{n,i})\right]^{-1} \frac{1}{n} \sum_{i=1}^n \left( \Psi_n(\vect{Y}_i;\check{\vect{\theta}}_n) - \Psi(\vect{Y}_i; \check{\vect{\theta}}_n) \right) \ .
\end{equation*}
The first term is bounded almost surely, while for the second term, we notice that
\begin{align*}
\left\| \frac{1}{n} \sum_{i=1}^n \left( \Psi_n(\vect{Y}_i;\check{\vect{\theta}}_n) - \Psi(\vect{Y}_i; \check{\vect{\theta}}_n) \right) \right\| & = \left\| \frac{1}{n} \sum_{i=1}^n \left( w(\delta_n(\vect{Y}_i;\check{\vect{\theta}}_n)) -  w(\delta(\vect{Y}_i;\check{\vect{\theta}}_n)) \right) u(\vect{Y}_i;\check{\vect{\theta}}_n) \right\| \\
& \le \sup_{\vect{y} \in \mathcal{Y}, \vect{\theta} \in \Theta} \left\vert w(\delta_n(\vect{Y}_i;\check{\vect{\theta}}_n)) -  w(\delta(\vect{Y}_i;\check{\vect{\theta}}_n)) \right\vert \\
& \times \frac{1}{n} \sum_{i=1}^n \left\Vert u(\vect{Y}_i;\check{\vect{\theta}}_n) \right\Vert  
\end{align*}
the first term goes to zero almost surely by Lemma \ref{lemma:deltan}, while the second term is bounded almost surely by assumption on the second moment of the score function. Hence $\hat{\vect{\theta}}_n - \check{\vect{\theta}}_n \stackrel{a.s.}{\rightarrow} \vect{0}$. On the other hand, by Lemma \ref{wconsistency1} we have $\check{\vect{\theta}}_n - \vect{\theta}_0 \stackrel{a.s.}{\rightarrow} \vect{0}$ and this concludes the proof.
\end{proof}  

\begin{remark}
	We stress again that the WLE is consistent for $\vect{\theta}_0=(\vect{\mu}_0, \Sigma_0^u)$ in the case of (\ref{wcem:population}), but the differences between the solutions to the population estimating equations (\ref{wem:population}) and (\ref{wcem:population}) are negligible for concentrated circular distributions, as well as for (\ref{em:population}) and (\ref{cem:population}).  
\end{remark}

\begin{theorem}[Asymptotic distribution]
  Under the assumptions of Theorem \ref{wconsistency2}. Assume, for each $n$, $\Psi_n$ be twice differentiable with respect to $\vect{\theta}$ with bounded derivatives; let $\dot{\Psi}_{n,jk} = \partial \Psi_{n,j}/\partial \theta_k$ assume, for all $\vect{y}$, $\vect{\theta}$ $\vert \dot{\Psi}_{n,jk} \vert \le K(\vect{y})$ with $\E(K(\vect{Y})) < \infty$. Then
\begin{equation*}  
\sqrt{n} (\hat{\vect{\theta}}_n - \vect{\theta}_0) \stackrel{d}{\rightarrow} N(\vect{0}, A^{-1})  
\end{equation*}
where $A = \E_{\vect{\theta}_0}(\Psi(\vect{y}; \vect{\theta}_0)\Psi(\vect{y}; \vect{\theta}_0)^\top)$.
\end{theorem}  

\begin{proof}
  The proof is similar to Theorem 10.11 of \citet{maronna2019robust}. Let $\epsilon(\vect{\theta}) = \E_{\vect{\theta}_0}\Psi(\vect{Y};\vect{\theta})$ and $B$ the matrix of derivatives with elements $\partial \epsilon_j/\partial\theta_k \vert_{\vect{\theta} = \vect{\theta}_0}$. For each $n$ and $j$ call $\ddot{\Psi}_{n,j}$ be the matrix with elements $\partial \Psi_{n,j}/\partial \theta_k \partial \theta_h$. Let $A_n = \frac{1}{n} \sum_{i=1}^n \Psi_n(\vect{Y}_i; \vect{\theta}_0)$, $B_n = \frac{1}{n} \sum_{i=1}^n \dot{\Psi}_n(\vect{Y}_i; \vect{\theta}_0)$ and $C_n$ is the matrix with its $j$th row equals to $(\hat{\vect{\theta}}_n - \vect{\theta}_0)^\top \frac{1}{2n} \sum_{i=1}^n \ddot{\Psi}_{n,j}(\vect{Y}_i; \vect{\theta}_i)$. We notice that $\frac{1}{n} \sum_{i=1}^n \ddot{\Psi}_{n,j}(\vect{Y}_i; \vect{\theta}_i)$ is bounded and since $\hat{\vect{\theta}}_n - \vect{\theta}_0 \stackrel{a.s.}{\rightarrow} \vect{0}$ by Theorem \ref{wconsistency2}, this implies that $C_n \stackrel{a.s.}{\rightarrow} \vect{0}$. From a second order Taylor expansion around $\vect{\theta}_0$ of $\Psi_n(\vect{y};\vect{\theta})$ it is easy to see that
\begin{equation*}
\sqrt{n} (\hat{\vect{\theta}}_n - \vect{\theta}_0) = - (B_n + C_n)^{-1} \sqrt{n} A_n \ . 
\end{equation*}
From the proof of Theorem \ref{wconsistency2} we have $A_n - \frac{1}{n} \sum_{i=1}^n u(\vect{Y_i}; \vect{\theta}_0) \stackrel{a.s.}{\rightarrow} \vect{0}$. In a similar way, using Lemma \ref{lemma:deltan} we have $B_n -  B \stackrel{a.s.}{\rightarrow} \vect{0}$. Since $u(\vect{Y_i}; \vect{\theta}_0)$ are i.i.d and finite second moments, by multivariate central limit theorem we have $\frac{1}{\sqrt{n}} \sum_{i=1}^n u(\vect{Y_i}; \vect{\theta}_0) \stackrel{d}{\rightarrow} N(\vect{0}, A)$ and hence $\sqrt{n} A_n$ has the same limit. We notice that $B$ coincides with the second derivatives of the log-likelihood and we had assume it positive definite. So, by the multivariate Slutsky's lemma, see, e.g. \citet[Theorem 10.10]{maronna2019robust} we have
\begin{equation*}  
\sqrt{n} (\hat{\vect{\theta}}_n - \vect{\theta}_0) \stackrel{d}{\rightarrow} N(\vect{0}, B^{-1} A B^{-1\top}) \ , 
\end{equation*}
on the other hand, under Bartlett's assumption we have $A=B$ and the result holds.
\end{proof}

In the next corollary we provide a set of assumptions so that the previous results can be applied to wrapped unimodal elliptical symmetric models.
\begin{corollary}
Consider a wrapped unimodal elliptically symmetric model as in (\ref{elliptic}). Let $\vect{\theta}_0 = (\vect{\mu}_0$, $\Sigma_0$) be the true values with $\Sigma_0$ be a non singular covariance matrix, i.e., the sample $\vect{y}_1, \ldots, \vect{y}_n$ is i.i.d. from $m(\cdot; \vect{\theta}_0)$. Let $h(\cdot)$ be a strictly decreasing, non-negative function with uniformly bounded third derivatives and $h(\cdot)$ is positive in the region $T(\vect{\mu}_0)$. Assumptions in Lemma \ref{lemma:deltan} hold. Then,
\begin{itemize}
\item[i.] the sequence $\hat{\theta}_n$ solutions of (\ref{WLE}) is strongly consistent for $\vect{\theta}_0$ and 
\begin{equation*}  
\sqrt{n} (\hat{\vect{\theta}}_n - \vect{\theta}_0) \stackrel{d}{\rightarrow} N(\vect{0}, I(\vect{\theta}_0)^{-1}) \ , 
\end{equation*}
where $I = \E_{\vect{\theta}_0}(\nabla_{\vect{\theta}\vect{\theta}^\top} m^\circ(\vect{Y}; \vect{\theta})) \vert_{\vect{\theta} = \vect{\theta}_0}$ is the expected Fisher information matrix.
\item[ii.] the sequence $\tilde{\theta}_n$ solutions of (\ref{WCEM}) is strongly consistent for $\vect{\theta}_0^u = (\vect{\mu}_0, \Sigma_0^u)$ and 
\begin{equation*}  
\sqrt{n} (\tilde{\vect{\theta}}_n - \vect{\theta}_0^u) \stackrel{d}{\rightarrow} N(\vect{0}, I^u(\vect{\theta}_0^u)^{-1}) \ , 
\end{equation*}
where $I^u = \E_{\vect{\theta}_0}(\nabla_{\vect{\theta}\vect{\theta}^\top} m^u(\vect{Z}; \vect{\theta})) \vert_{\vect{\theta} = \vect{\theta}_0^u}$ is the expected Fisher information matrix based on $m^u(\cdot; \vect{\theta}_0)$. 
\end{itemize}  
\end{corollary}

\subsection{Influence Function}
The influence function (IF) plays a very important role in the evaluation of local robust properties of estimators in a classic robust framework \citep{huber2robust}. For a class of minimum distance estimators and weighted likelihood estimators \citep{beran1977minimum, lindsay1994}, under broad regularity conditions and the assumed model, the IF coincides with that of the MLE. This feature suggests their high efficiency from one side, but a lack of local robustness on the other.  
The IF was used to investigate the robustness of some estimators for the circular mean direction in \citet{he1992robust} but its use was unsatisfactory.
Here, we discuss the IF of the proposed WLE in a more general setting. 

Given a distribution function $F$, let $\vect{T}: F \mapsto \vect{T}(F) \in \Theta$ be a statistical functional that admits a von Mises expansion \citep{serfling2009approximation}. 
Given the gross error neighborhood $F_\epsilon(\vect{x}) = (1 - \epsilon) F(\vect{x}) + \epsilon \mathbbm{1}_{\vect{z}}(\vect{x})$ we define the influence function of $\vect{T}$ at $\vect{z}$ as
\begin{equation*} \label{if}  
\IF(\vect{z}; \vect{T}, F)  = \lim_{\epsilon \downarrow 0} \frac{\vect{T}(F_\epsilon) - \vect{T}(F)}{\epsilon}
= \frac{\partial}{\partial \epsilon} \vect{T}(F_\epsilon) \big\vert_{\epsilon = 0} \ .
\end{equation*}
Let $M_{\vect{\theta}}=M(\vect{x};\vect{\theta})$ be the assumed model and $u(\vect{x};\vect{\theta})$ the corresponding score function. Let $\vect{T}_F = \vect{T}(F)$ be the statistical functional solution of the weighted likelihood estimating equations
\begin{equation*}
\label{equ:funzionale}  
\int w(\vect{x}; \vect{T}(F), F) u(\vect{x}; \vect{T}(F)) \ dF(\vect{x}) = \vect{0} \ ,
\end{equation*}
where we have $\vect{T}(M_{\vect{\theta}}) = \vect{\theta}$. The derivation of the IF for such functional is  similar to the case of M-estimators \citep{huber2robust}. We have that 
$$\frac{\partial}{\partial \delta} w(\delta) = \left( \frac{\partial}{\partial \delta} A(\delta) - w(\delta) \right) \left( \delta + 1 \right)^{-1}$$ and
\begin{equation*}
  \frac{\partial}{\partial \epsilon} \delta(\vect{x}; \vect{T}(F_\epsilon),F_\epsilon) \vert_{\epsilon=0} = - \frac{k(\vect{x}; \vect{z}, H)}{\hat{m}(\vect{x}; \vect{T}(F))} + (\delta(\vect{x}; \vect{T}(F),F) + 1) (1 - \hat{u}(\vect{x}; \vect{T}(F)) \IF(\vect{z}; \vect{T}, F)) \ ,
\end{equation*}  
where $\hat{m}(\vect{x}; \vect{\theta})$ is the smoothed model and $\hat{u}(\vect{x};\vect{\theta}) = \frac{\partial}{\partial \vect{\theta}} \log \hat{m}(\vect{x}; \vect{\theta})$. Then, we obtain
\begin{equation*} \label{equ:funinf1}
IF(\vect{z}; \vect{T}, F) = D(F)^{-1} N(\vect{z},F)
\end{equation*}
where
\begin{align*}
N(\vect{z},F) & = w(\vect{z}; \vect{T}(F), F) u(\vect{z}; \vect{T}(F)) \\
& +  \int w'(\delta(\vect{x}; \vect{T}(F), F)) \frac{k(\vect{x};\vect{z},H)}{\hat{m}(\vect{x}; \vect{T}(F))} u(\vect{x}; \vect{T}(F)) \ dF(\vect{x}) \\
& - \int w'(\delta(\vect{x}; \vect{T}(F), F)) (\delta(\vect{x}; \vect{T}(F), F) + 1) u(\vect{x}; \vect{T}(F)) \ dF(\vect{x}) \\
& = w(\vect{z}; \vect{T}(F), F) u(\vect{z}; \vect{T}(F)) \\
& + \int (A'(\delta(\vect{x}; \vect{T}(F),F)) - w(\delta(\vect{x}; \vect{T}(F),F))) \\
& \times \left(\frac{k(\vect{x}; \vect{z}, H)}{\hat{f}(\vect{x})} - 1 \right) u(\vect{x}; \vect{T}(F)) \ dF(\vect{x}) \\
\end{align*}
and
\begin{align*}
D(F) & = \int w'(\delta(\vect{x}; \vect{T}(F),F))(\delta(\vect{x}; \vect{T}(F),F) + 1) \hat{u}(\vect{x}; \vect{T}(F)) u(\vect{x};\vect{T}(F))^\top \ dF(\vect{x}) \\
& - \int w(\vect{x}; \vect{T}(F),F) u'(\vect{x}; \vect{T}(F)) \ dF(\vect{x}) \\
& = \int (A'(\delta(\vect{x}; \vect{T}(F),F)) - w(\delta(\vect{x}; \vect{T}(F),F))) \hat{u}(\vect{x}; \vect{T}(F)) u(\vect{x}; \vect{T}(F))^\top \ dF(\vect{x}) \\
& - \int w(\delta(\vect{x};\vect{T}(F),F)) u'(\vect{x};\vect{T}(F)) \ dF(\vect{x}) \ .
\end{align*}
where $u'(\vect{x};\vect{\theta}) = \frac{\partial^2}{\partial \vect{\theta}\partial \vect{\theta}^\top} \log m(\vect{x}; \vect{\theta})$. 
Under the model, we obtain the classical IF, that for the WLE corresponds to that of the MLE, i.e.
\begin{equation*}
\IF(\vect{z}; \vect{T}, M_{\vect{\theta}_0}) = I(\vect{\theta}_0)^{-1} \ u(\vect{z};\vect{\theta}_0) \ ,
\end{equation*} 
where $I(\vect{\theta}) = -\mathbb{E}_{\vect{\theta}}(u'(\vect{x};\vect{\theta}))$ is the expected Fisher information matrix.
However, the behavior of the IF under a distribution other than the postulated model is very different. 
As an example let us consider a simple setting in which $m^\circ(\vect{y}; \mu, \sigma^2)$ is the univariate WN and we are interested in evaluating the IF for the location functional when the data are from a two components mixture $f^\circ(\vect{y}) = (1 - \varepsilon) m^\circ(\vect{y}; 0, \sigma_0^2) + \varepsilon m^\circ(\vect{y}; \pi/2, (\pi/16)^2)$. In Figure \ref{IFwrapped} we show the IF of the location functional $\mu(F^\circ)$ defined as the solution to the estimating equation (\ref{wem:population}) for $\sigma_0=\pi/8$ (left panel) and $\sigma_0=\pi/4$ (right panel). In this setting, the IF is a periodic function and in a region of high probability for the contaminating distribution the influence of a point is almost null. On the opposite, the behavior of the IF outside that region is similar to that of the maximum likelihood functional. We also notice the change in sign at the antimode ($\pm\pi$).  
When we consider the location functional $\mu(F)$ associated to the WLE defined by (\ref{wcem:population}) with Pearson residuals as in (\ref{pearson}) or (\ref{pearson:lin}), the IF is not periodic and it is zero outside the interval $(\mu - \pi, \mu+\pi)$. Inside the interval, the behavior of the IF is similar to that of $\mu(F^\circ)$, as it is shown in Figure \ref{IFeuclidean}. 
in contrasts, the IF of  $\mu(F)$ with Pearson residuals built according to the geometric approach is symmetric, since only the magnitude of the outliers plays a role in the Mahalanobis distance, as shown in Figure \ref{IFeuclidean-chisq}.

\begin{figure*}[t]
	\centering
	\includegraphics[width=0.45\textwidth]{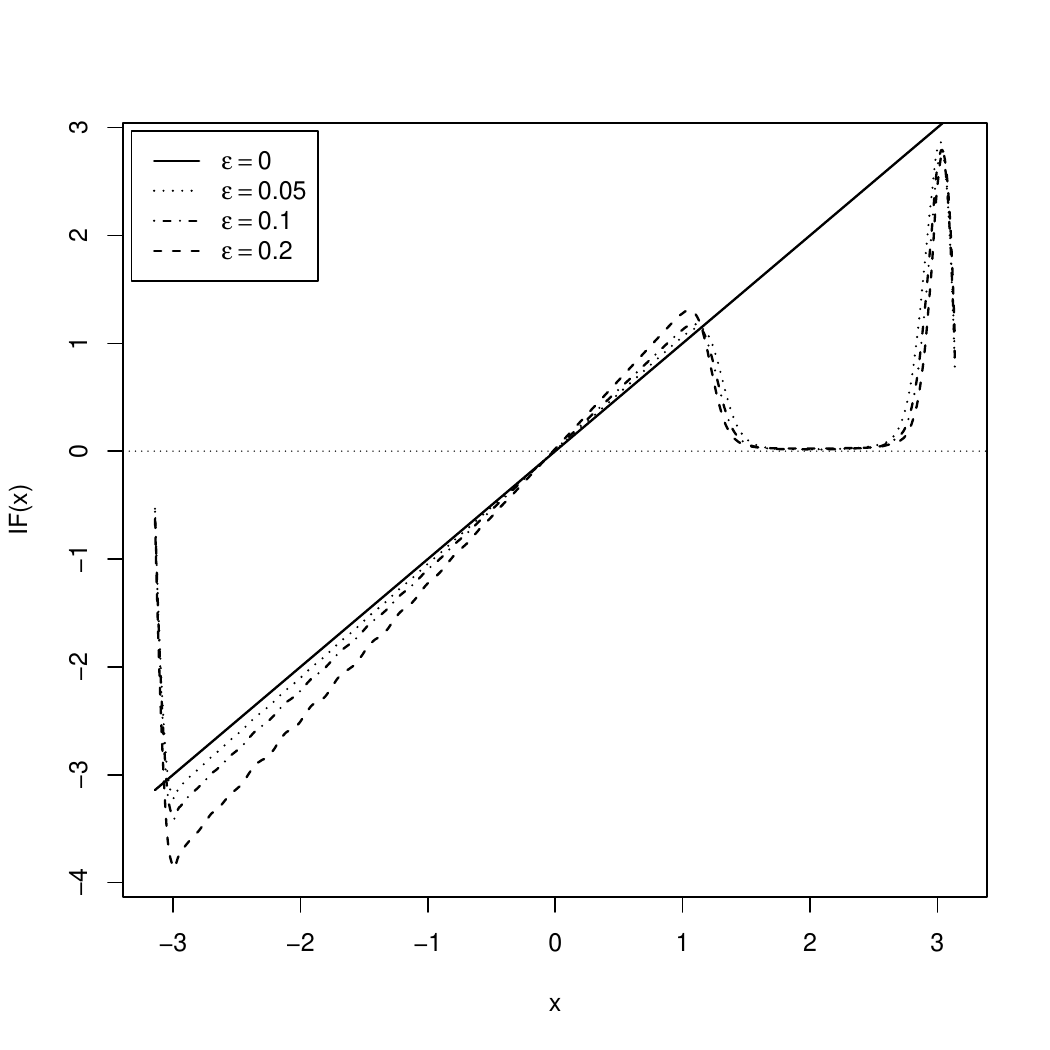}
	\includegraphics[width=0.45\textwidth]{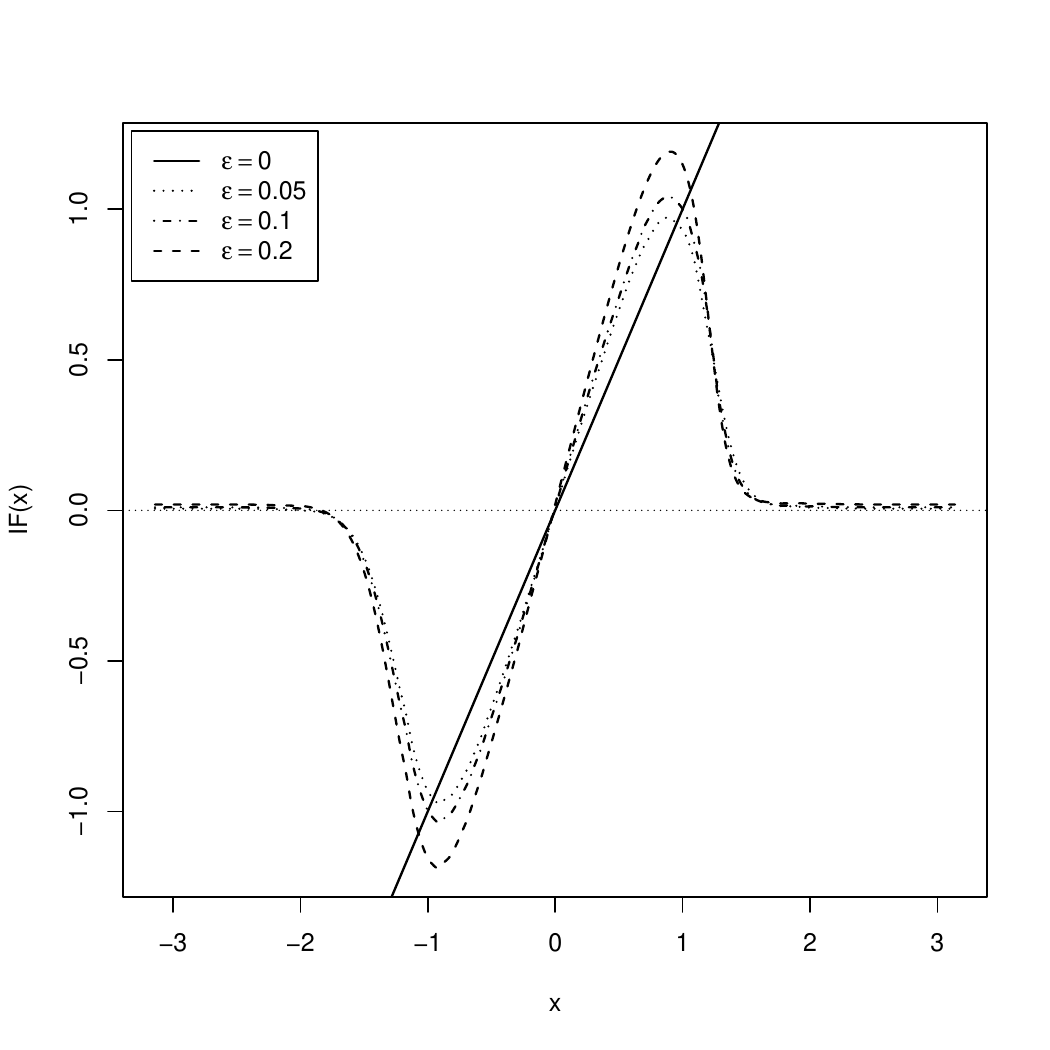}
	\caption{WEM. Influence function for the location functional $\mu(F)$ with $f^\circ(\vect{y}) = (1 - \varepsilon) m^\circ(\vect{y}; 0, \sigma_0^2) + \varepsilon m^\circ(\vect{y}; \pi/2, (\pi/16)^2)$, for $\varepsilon=0, 0.05, 0.10, 0.20$ and $\sigma_0=\pi/8$ (left panel) and $\sigma_0=\pi/4$ (right panel).} 
	\label{IFwrapped}
\end{figure*}

\begin{figure*}[t]
	\centering
	\includegraphics[width=0.45\textwidth]{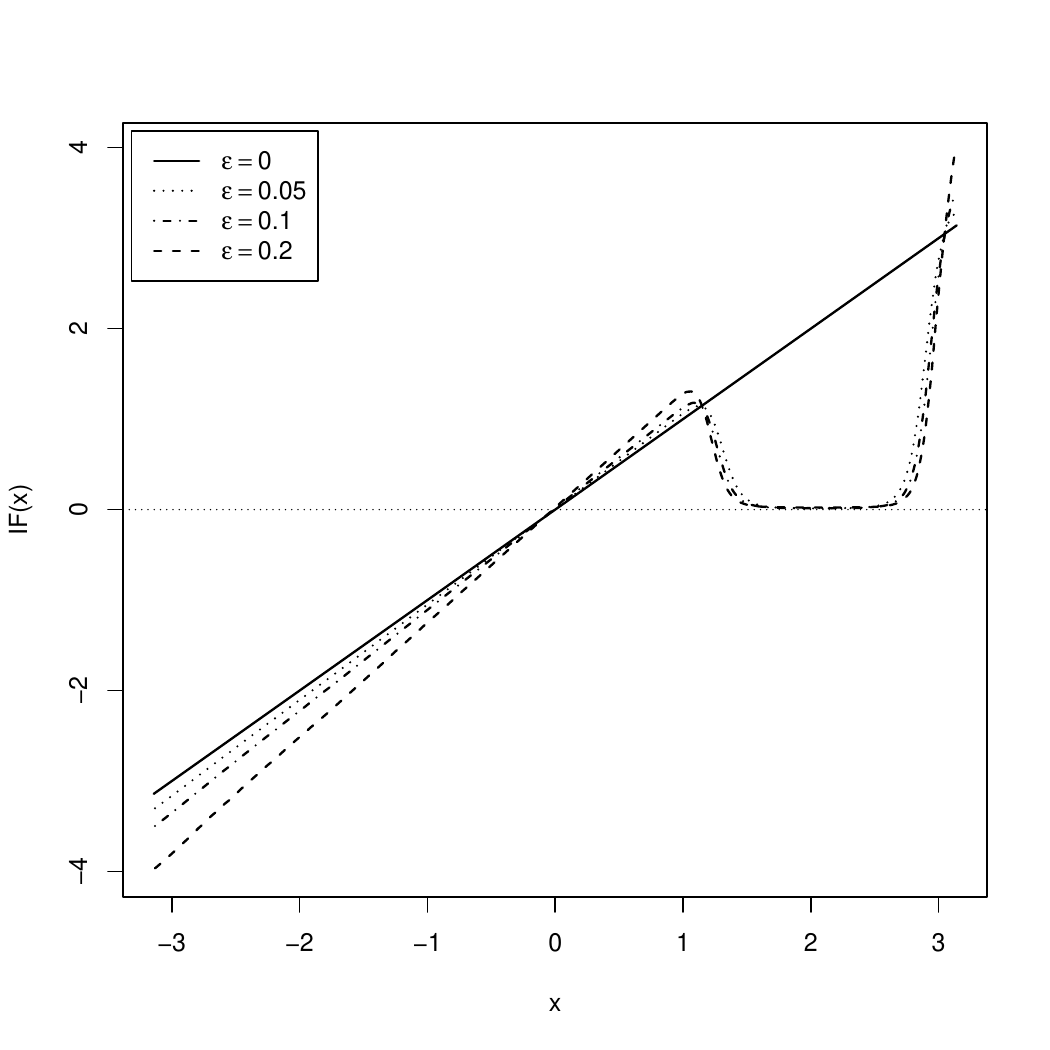}
	\includegraphics[width=0.45\textwidth]{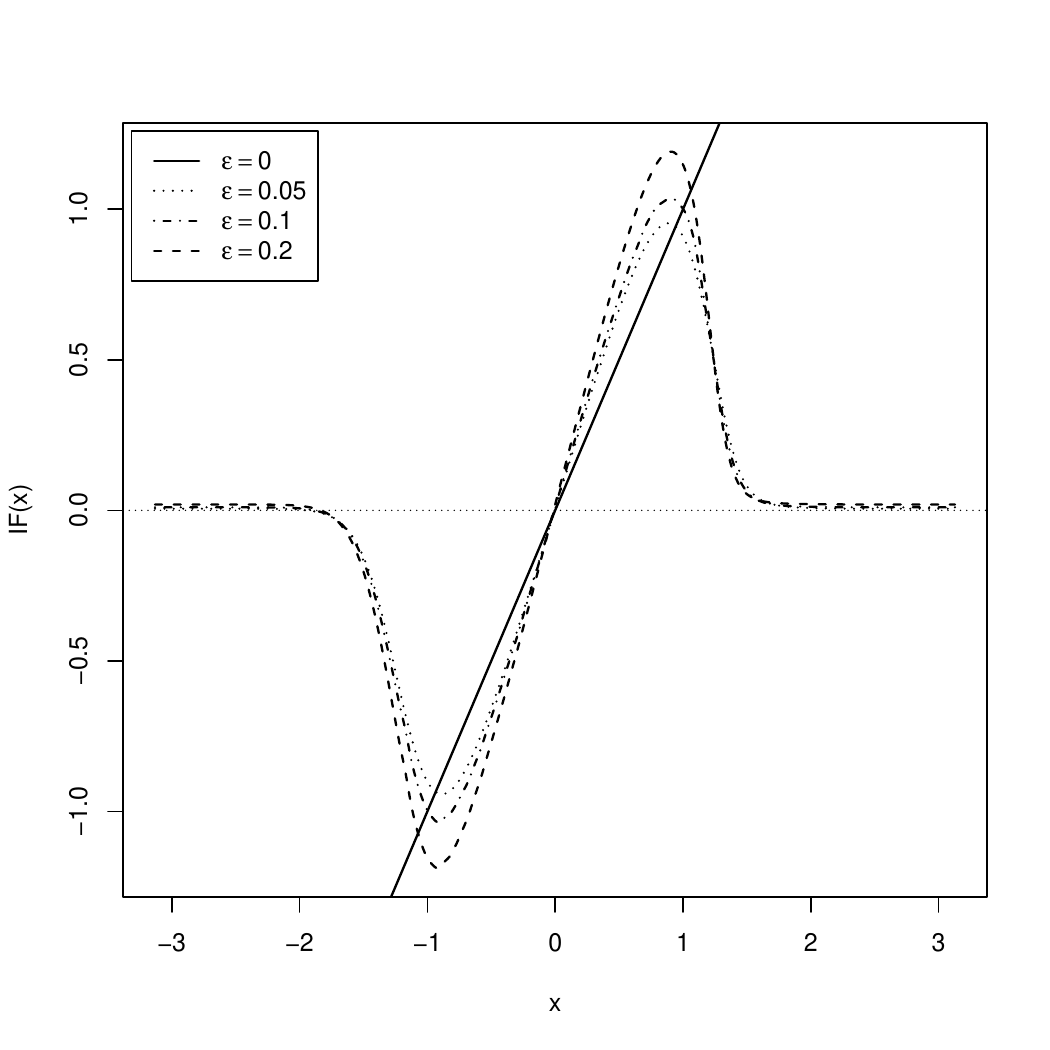}
	\caption{WCEM using Pearson residuals as in (\ref{residualfs}) or (\ref{residualfs2}). Influence function for the location functional $\mu(F)$ with $f^\circ(\vect{y}) = (1 - \varepsilon) m^\circ(\vect{y}; 0, \sigma_0^2) + \varepsilon m^\circ(\vect{y}; \pi/2, (\pi/16)^2)$, for $\varepsilon=0, 0.05, 0.10, 0.20$ and $\sigma_0=\pi/8$ (left panel) and $\sigma_0=\pi/4$ (right panel).} 
	\label{IFeuclidean}
\end{figure*}

\begin{figure*}[t]
	\centering
	\includegraphics[width=0.45\textwidth]{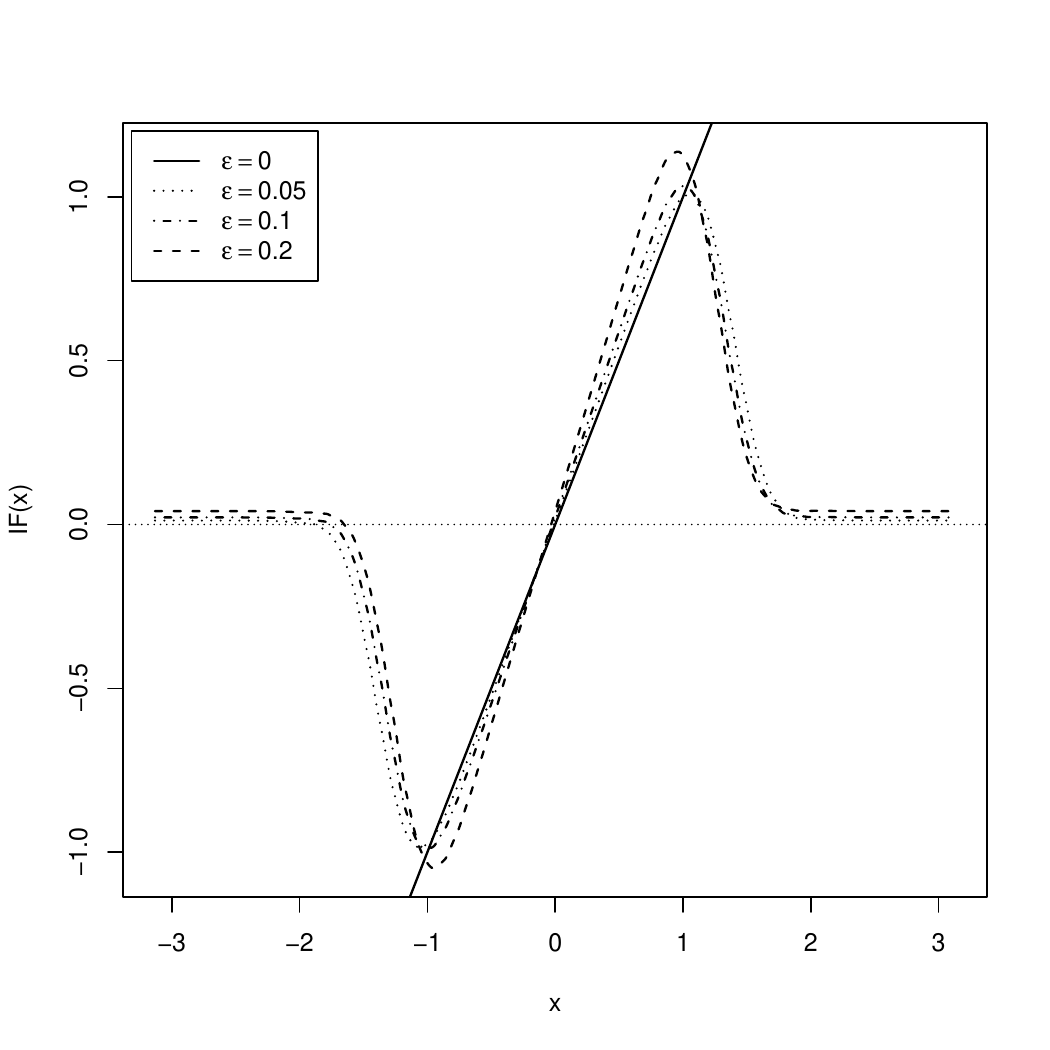}
	\includegraphics[width=0.45\textwidth]{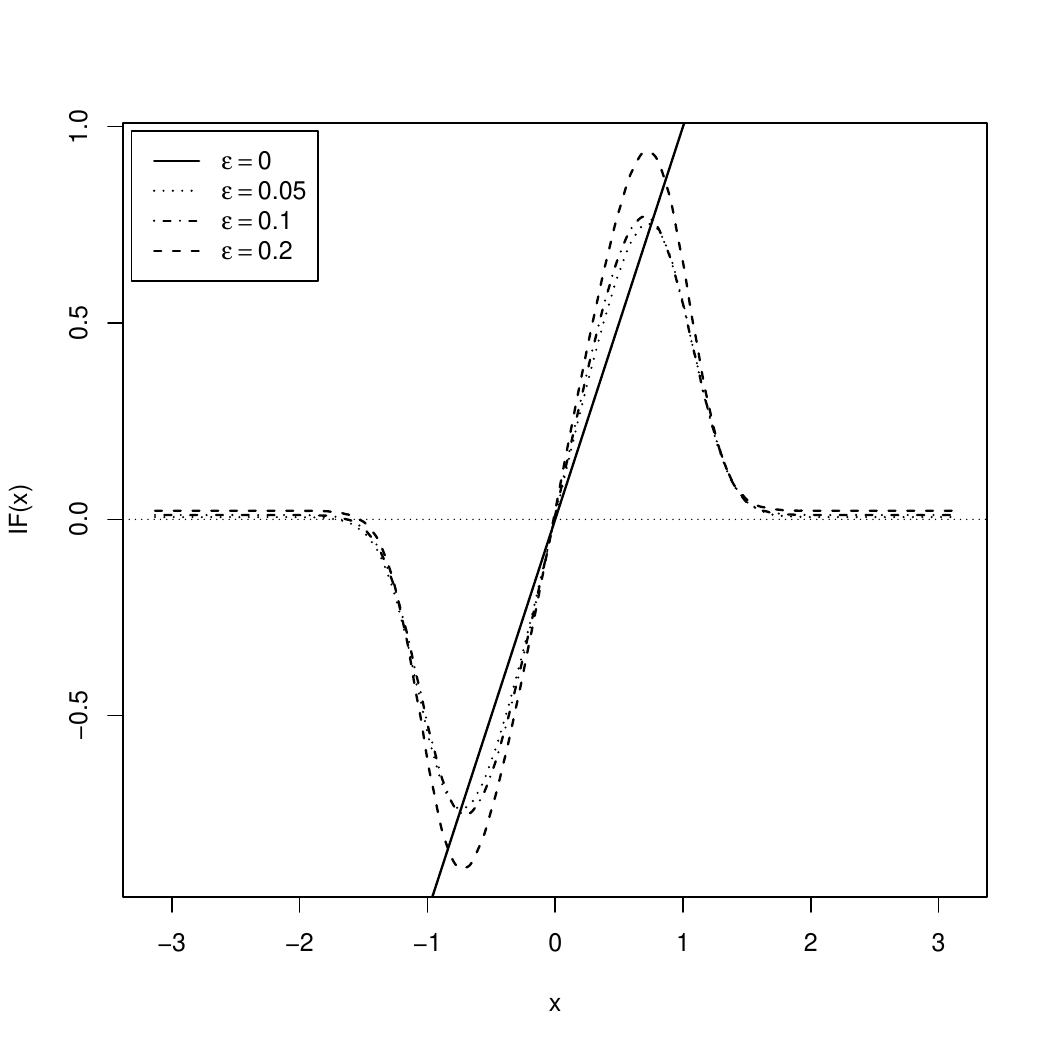}
	\caption{WCEM using Pearson residuals as in (\ref{residualfsGunw}). Influence function for the location functional $\mu(F)$ with $f^\circ(\vect{y}) = (1 - \varepsilon) m^\circ(\vect{y}; 0, \sigma_0^2) + \varepsilon m^\circ(\vect{y}; \pi/2, (\pi/16)^2)$, for $\varepsilon=0, 0.05, 0.10, 0.20$ and $\sigma_0=\pi/8$ (left panel) and $\sigma_0=\pi/4$ (right panel).} 
	\label{IFeuclidean-chisq}
\end{figure*}

\section{Numerical studies}
\label{sec:5}

In this section, we investigate the finite sample behavior of the proposed WLEs given by the WLEE in (\ref{WLE}) and (\ref{WCEM}), for the different weighting schemes considered. The numerical studies are limited to the WN case. Since solving the WLEE in this case is equivalent to consider a weighted counterpart of the EM or CEM algorithms,  in order to make it easier to read the results, we denote the WLE solution to (\ref{WLE}) as WEM and  the approximate WLE solution to (\ref{WCEM}) as WCEM-torus, WCEM-unwrap and WCEM-dist, depending on whether weights are based on residuals in (\ref{residualfs}), (\ref{residualfs2bis}) or (\ref{residualfsG}), respectively.
The MLE and its approximated version have been also taken into account and are denoted by EM and CEM, respectively.
We consider numerical studies based on $N = 500$ Monte Carlo trials. 
Data are sampled from a $p-$variate WN with null mean vector and variance-covariance matrix $\Sigma=D^{1/2}RD^{1/2}$, where $R$ is a random correlation matrix with condition number set equal to $20$ and $D=\sigma\mathrm{I}_p$. Contamination has been added by replacing a proportion $\epsilon$ of randomly selected data points. Those observations are shifted by an amount $k_\epsilon$ in the direction of the smallest eigenvector of $\Sigma$ and perturbed by adding some noise from a $p-$variate wrapped normal with independent components and marginal scale $\sigma_\epsilon$. 
We considered a sample size $n=250$, number of dimensions $p=2,5$, $\sigma=\pi/8, \pi/4$, $\epsilon=0, 0.10, 0.20$, $k_\epsilon=\pi/2, \pi$, $\sigma_\epsilon=0.05$, $J=2$. The case $\epsilon=0$ concerns the situation without contamination and allows to investigate the behavior of the proposed robust methods at the true model. When $p=5$, contamination only affects the first two dimensions. 
The bandwidths have been chosen so that all the WLEs return an empirical downweighting level close to the nominal contamination size to make a fair comparison. The weights are based on a GKL RAF.
Initialization is based on subsampling with twenty subsamples of size $p+p(p+1)/2+5$. This choice did not represent an issue. Moreover, very often the different starting values led to the same solution. 
All the algorithms are assumed to reach convergence when 
$$
\max\left( g(\hat{\vect{\mu}}^{(s+1)}-\hat{\vect{\mu}}^{(s)}),
\|\hat\Sigma^{(s+1)}-\hat\Sigma^{(s)}\|
\right) < 10^{-6} \ ,
$$
where $g(\vect{\mu})=\sqrt{2(1-\cos(\vect{\mu}))}$.
Fitting accuracy is evaluated according to
\begin{itemize}
	\item[(i)] the square root average angle separation
	\begin{equation*}
		\sqrt{AS(\hat{\vect{\mu}}) }= \sqrt{\frac{1}{p} \sum_{j=1}^p(1-\cos(\hat{\vect{\mu}}_j ))},
	\end{equation*}
	\item[(ii)] the divergence:
	\begin{equation*}
		\Delta(\hat{\Sigma}) = \textrm{trace}(\hat{\Sigma}\Sigma^{-1})-\log(\textrm{det}(\hat{\Sigma}\Sigma^{-1}))-p.
	\end{equation*}
\end{itemize}
The effectiveness of the outliers detection rules described in Section \ref{sec:2} is assessed in terms of swamping and power, that is evaluating the rate of genuine observation wrongly declared outliers and that of outliers correctly detected, respectively, for an overall significance level $\alpha=1\%$.
Both univariate and multivariate kernel density estimation involved in the computation of Pearson residuals in (\ref{residualfs2bis}) and (\ref{residualfsG}), respectively, has been performed using the functions available from package {\tt pdfCluster} \citep{pdfclus}. The numerical studies are based on non optimized {\tt R} code and have been run on a 3.4 GHz Intel Core i5 quad-core. Codes are available as supplementary material. 

Figure \ref{figtab1} displays the results under the true model, for $p=2, 5$: the robust methods all provide accurate results in this scenario and the observed differences with respect to the MLE are tolerable. 

Figure \ref{figtab2a} and Figure \ref{figtab2b} give the empirical distributions of the four WLEs in the presence of contamination when $p=2$ and $\sigma=\pi/8$ or $\sigma=\pi/4$, respectively. As well,  Figure \ref{figtab3a} and Figure \ref{figtab3b} concern the case with $p=5$. 
The MLE becomes unreliable and it is not shown. In contrasts, the robust techniques always provide resistant estimates, as expected. We do not observe relevant differences among the robust proposals in terms of fitting accuracy. For what concerns the task of outliers detection, all the suggested WLEs return an average rate of swamping close to the nominal level and a power almost always equal to one, for all considered scenario and they do not exhibit different performances. 

Computational time was always in a feasible range. However, based on the current codes, there is a remarkable time saving from the use of WCEM-unwrap or WCEM-dist with respect to WCEM-torus and WEM. One main reason could be the use of the functions from {\tt pdfCluster} in the former two methods. For instance, when $p=2$, $\sigma=\pi/4$, $\epsilon=20\%$ the median elapsed time was about 12 seconds for the WEM and the WCEM-torus, but only 1.3 seconds for the WCEM-unwrap and slightly larger (still less than two) for the WCEM-dist. The advantage of using the WCEM combined with Pearson residuals in (\ref{residualfs2bis}) was overwhelming for $p=5$: with $\sigma=\pi/4$ and $\epsilon=20\%$ the WEM and WCEM-torus took a median time of about 75 and 80 seconds, respectively for $k_\epsilon=\pi/2$, whereas the WCEM-unwrap took about 9 seconds and  the WCEM-dist about 35 seconds. The case with $k_\epsilon=\pi$ was less computationally demanding but still the differences were noticeable: about 55 seconds for the WEM and WCEM-torus, about 27 seconds for the WCEM-dist and only about 4 seconds for the WCEM-unwrap. The ability to evaluate weights on the unwrapped data rather than on the torus reduced the computational time, indeed. 

\begin{figure*}[t]
	\centering
	\includegraphics[scale=0.28]{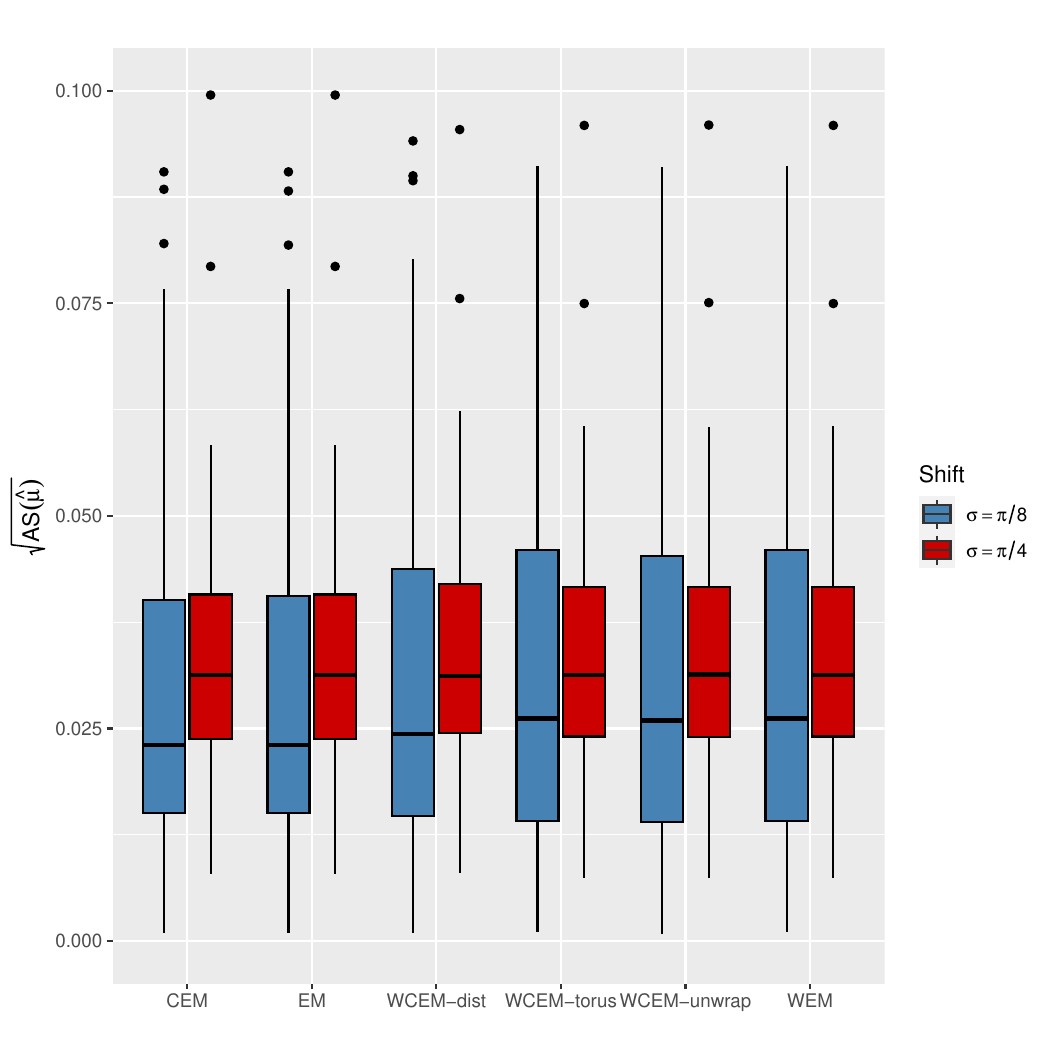}
	\includegraphics[scale=0.28]{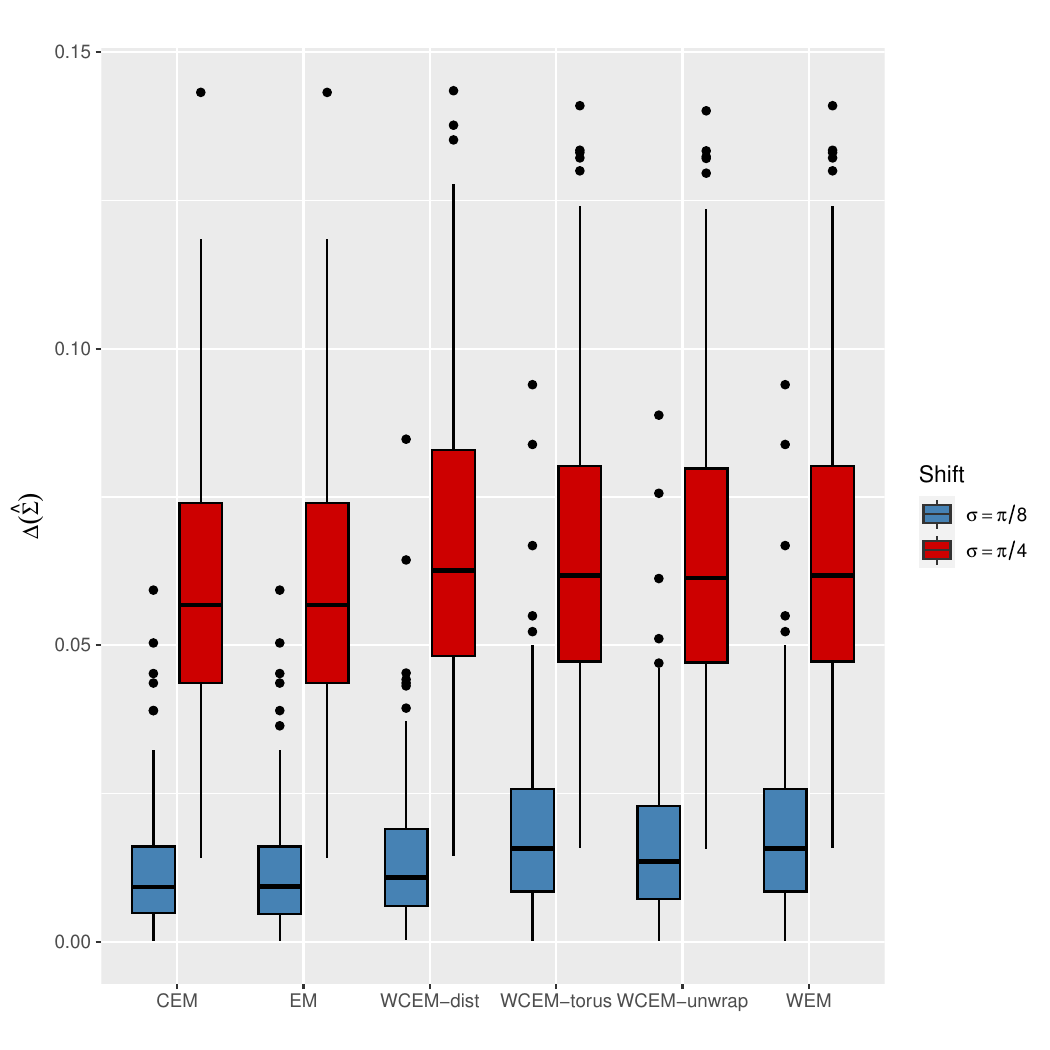}\\
	\includegraphics[scale=0.28]{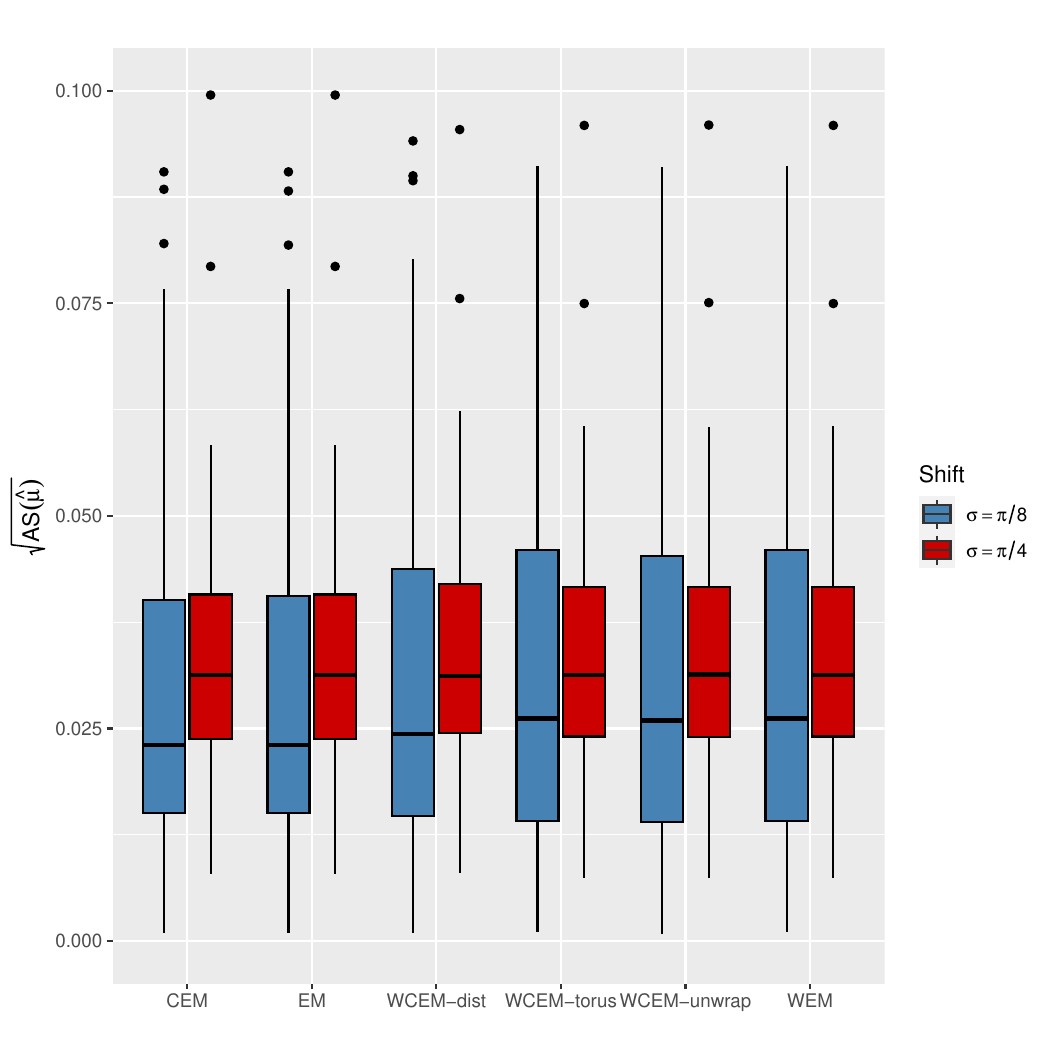}
	\includegraphics[scale=0.28]{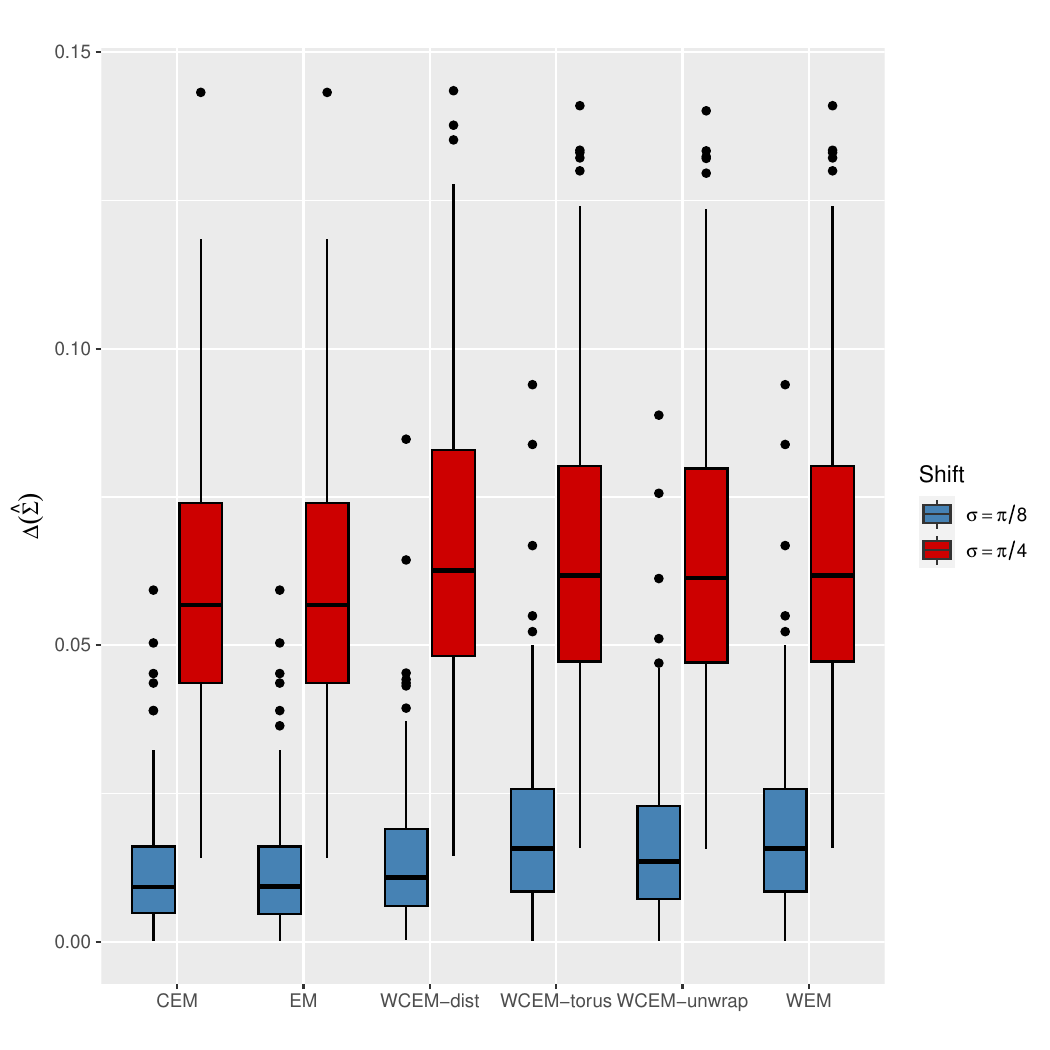}
	\caption{Box-plots for $\sqrt{AS(\hat{\vect{\mu}})}$ (left) and $\Delta(\hat\Sigma)$ (right) for $p=2$ (top) and $p=5$ (bottom),  $\sigma=\pi/8, \pi/4$ when $\epsilon=0$.} 
	\label{figtab1}
\end{figure*}

\begin{figure*}[t]
	\centering
	\includegraphics[scale=0.28]{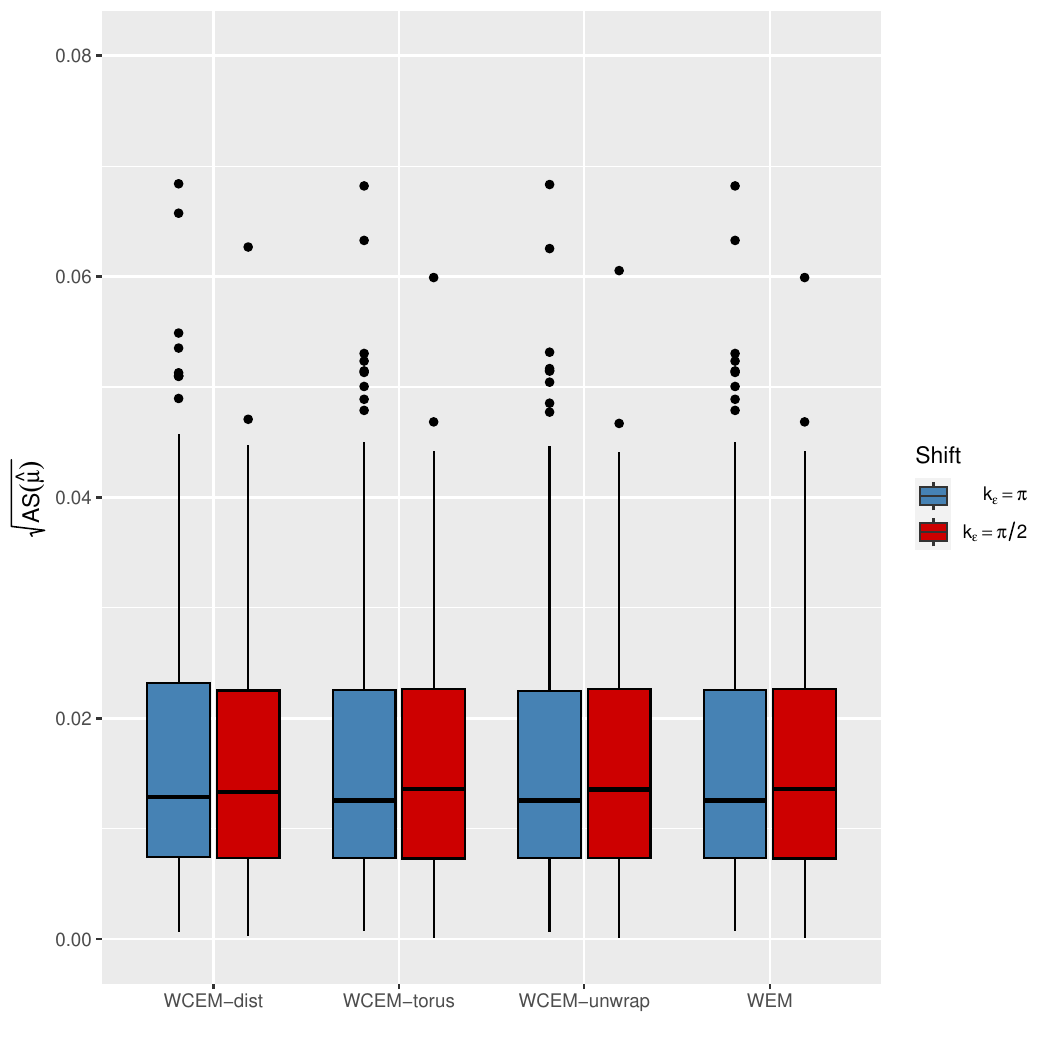}
	\includegraphics[scale=0.28]{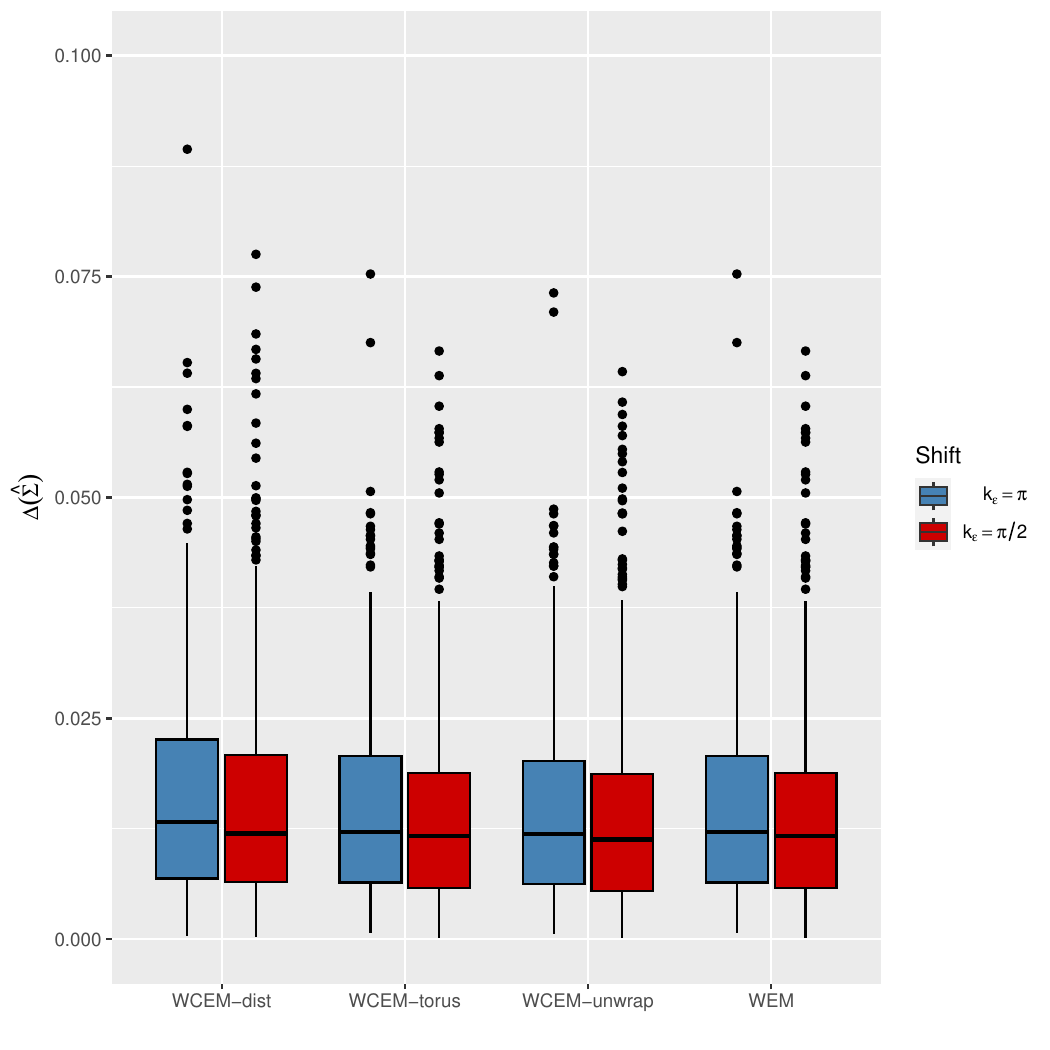}\\
	\includegraphics[scale=0.28]{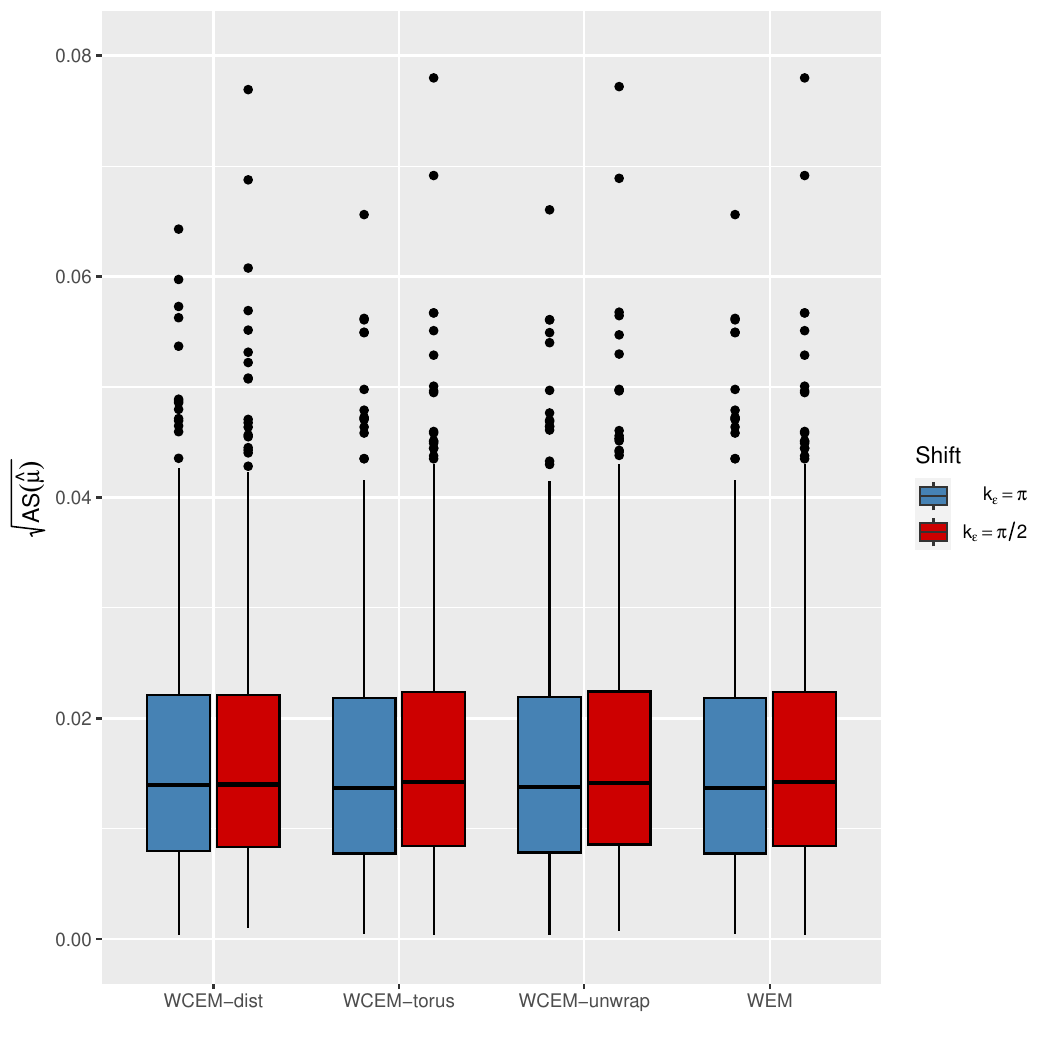}
	\includegraphics[scale=0.28]{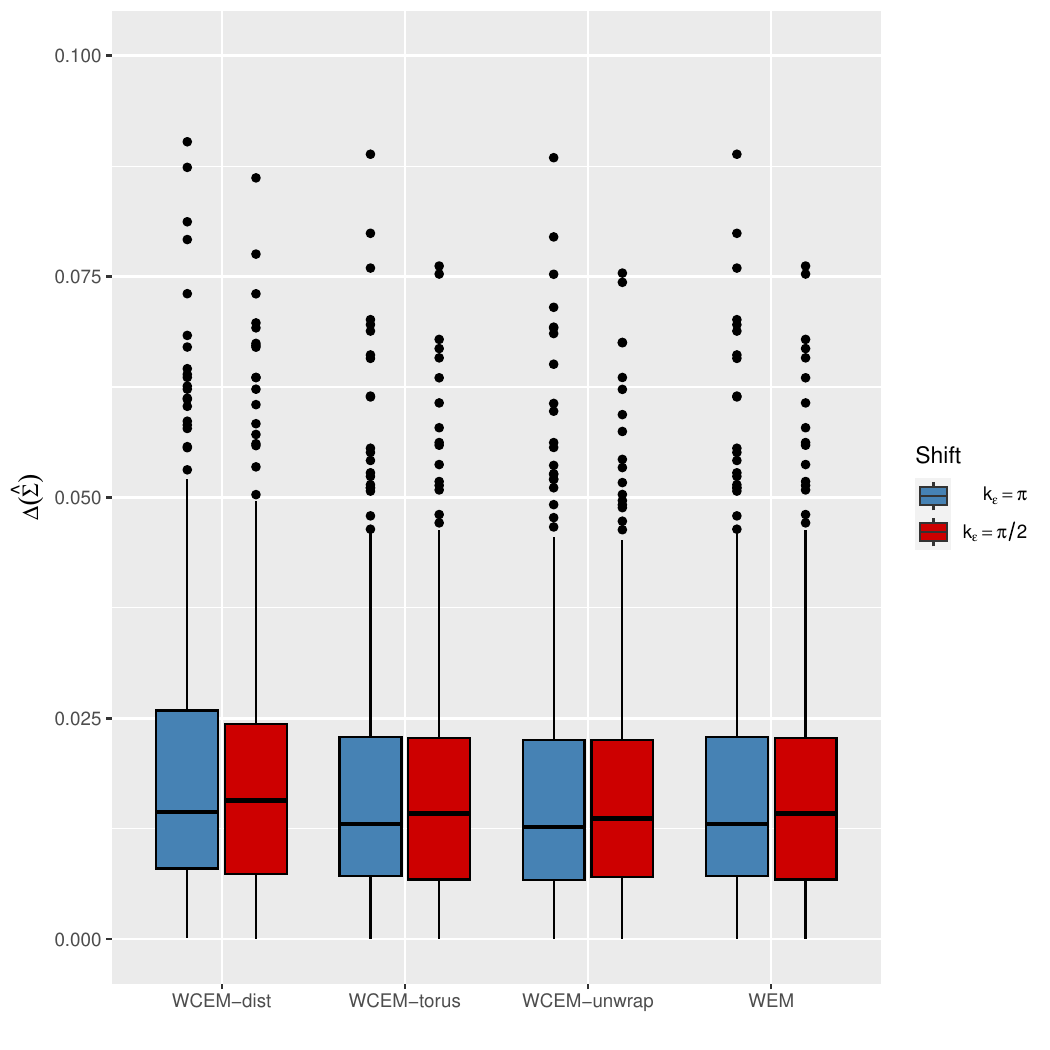}
	\caption{Box-plots for $\sqrt{AS(\hat{\vect{\mu}})}$ (left) and $\Delta(\hat\Sigma)$ (right) for $p=2$, $\sigma=\pi/8$, $k_\epsilon=\pi/2, \pi$ when $\epsilon=10\%$ (top) and $\epsilon=20\%$ (bottom).} 
	\label{figtab2a}
\end{figure*}

\begin{figure*}[t]
	\centering
	\includegraphics[scale=0.28]{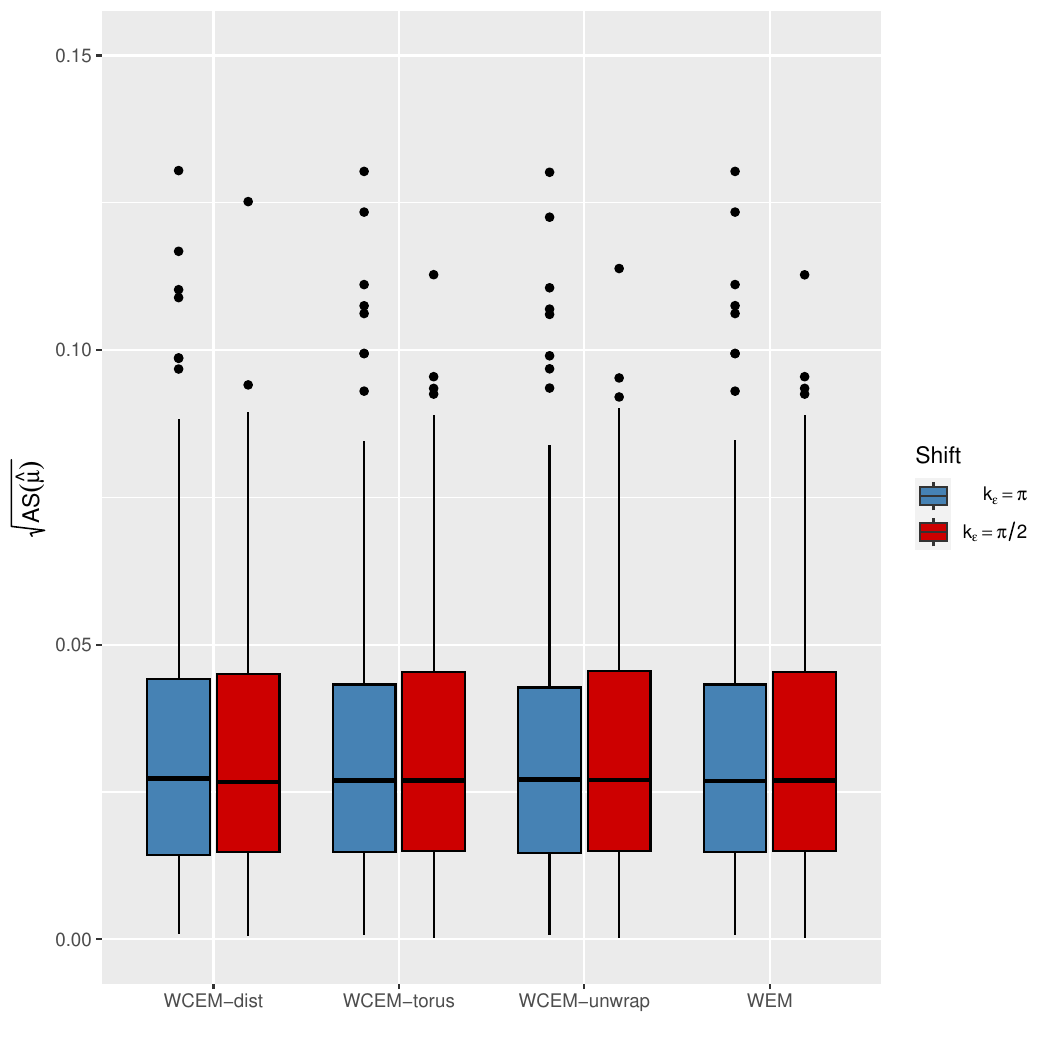}
	\includegraphics[scale=0.28]{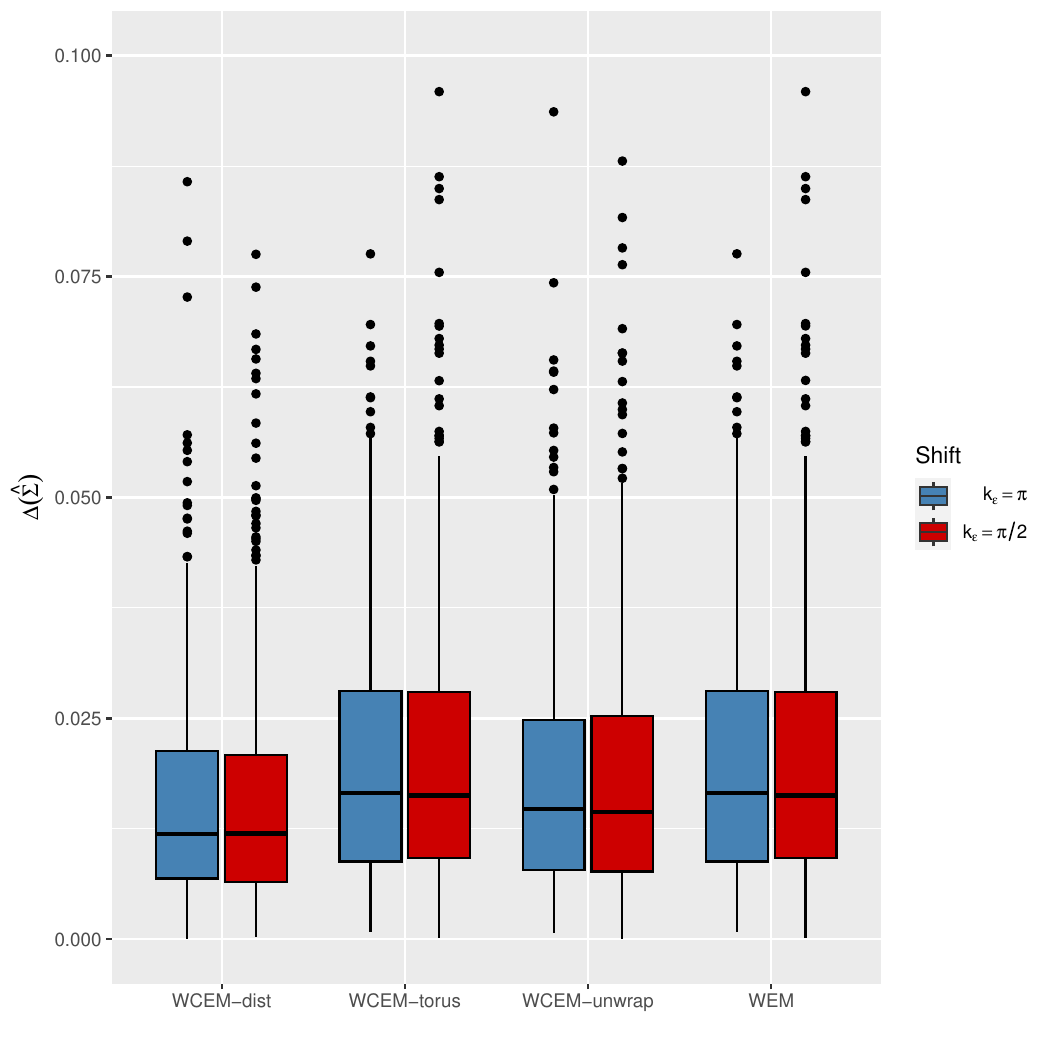}\\
	\includegraphics[scale=0.28]{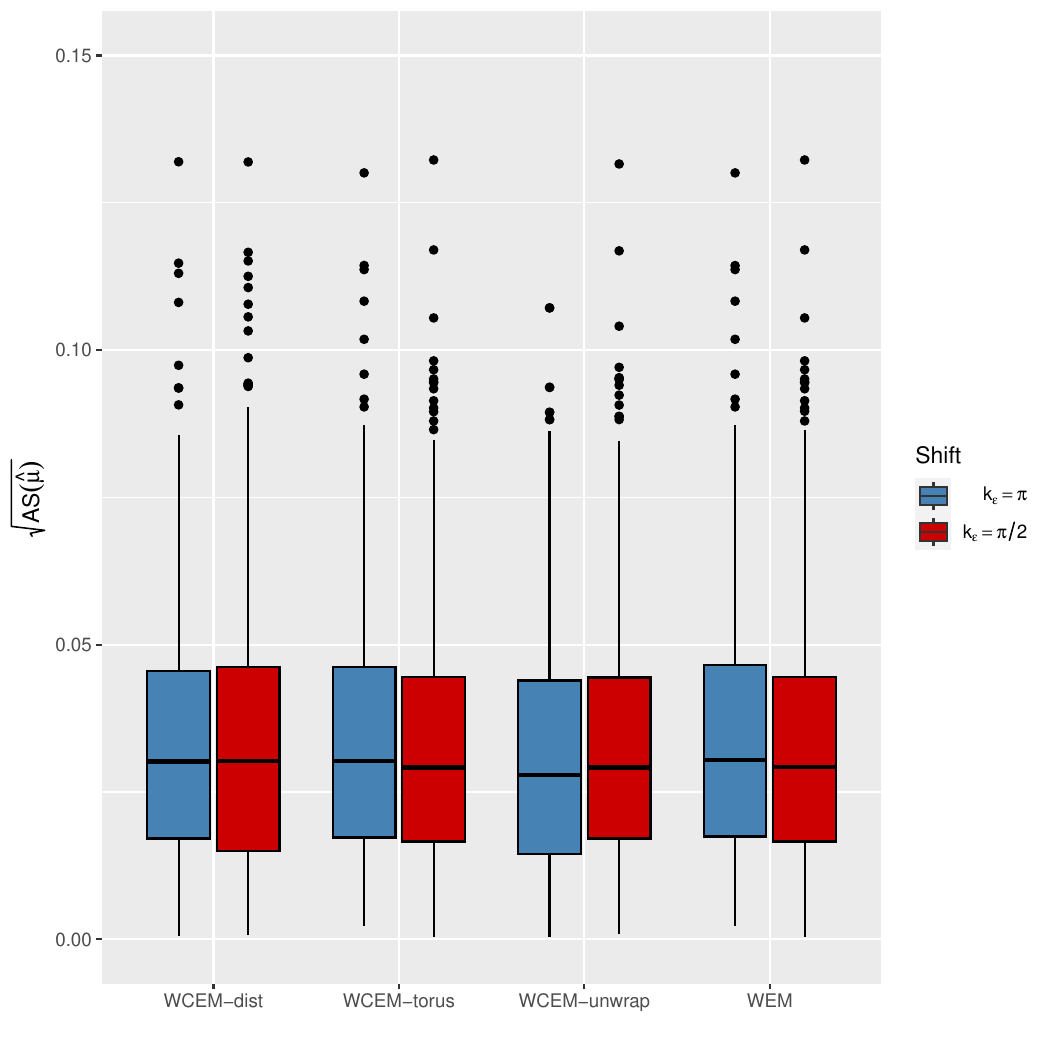}
	\includegraphics[scale=0.28]{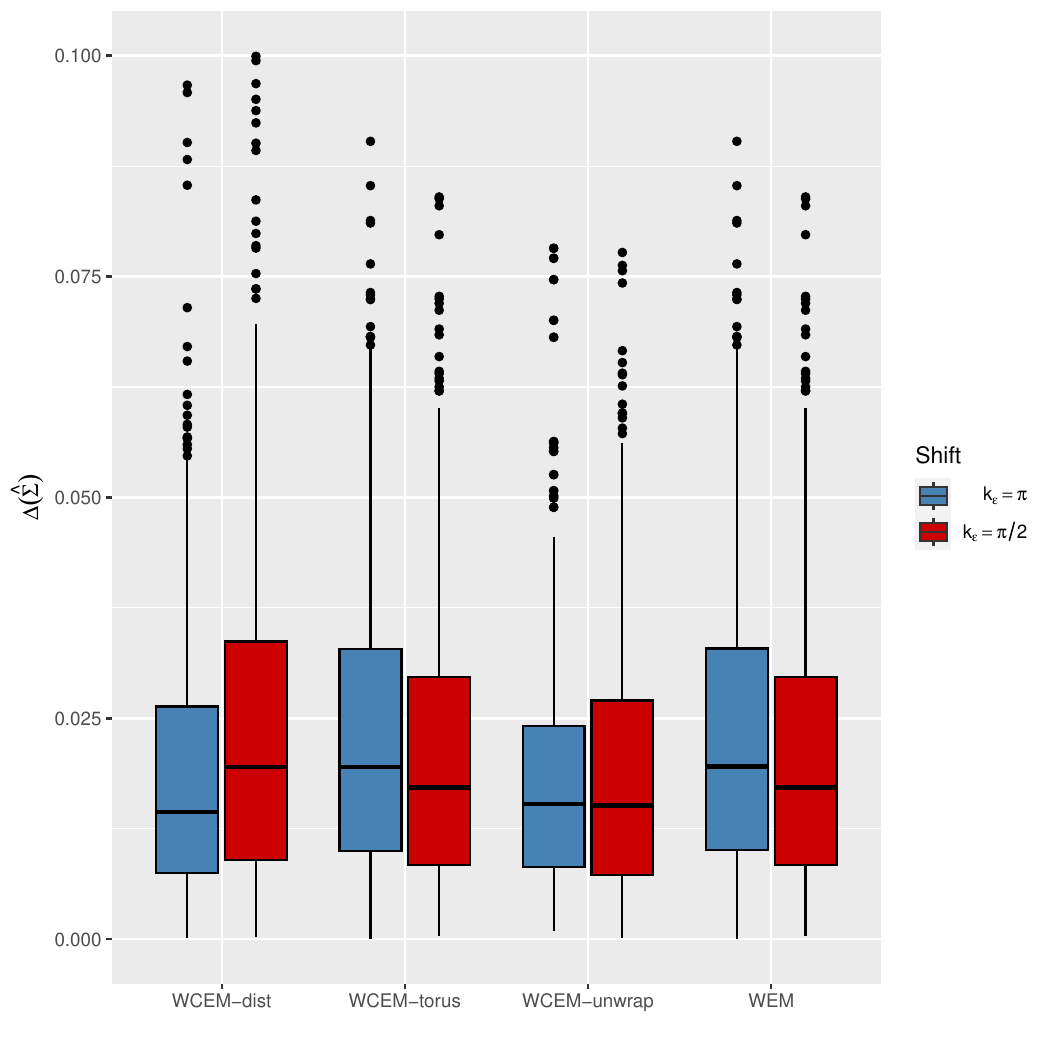}
	\caption{Box-plots for $\sqrt{AS(\hat{\vect{\mu}})}$ (left) and $\Delta(\hat\Sigma)$ (right) for $p=2$, $\sigma=\pi/4$, $k_\epsilon=\pi/2, \pi$ when $\epsilon=10\%$ (top) and $\epsilon=20\%$ (bottom).} 
	\label{figtab2b}
\end{figure*}

\begin{figure*}[t]
	\centering
	\includegraphics[scale=0.28]{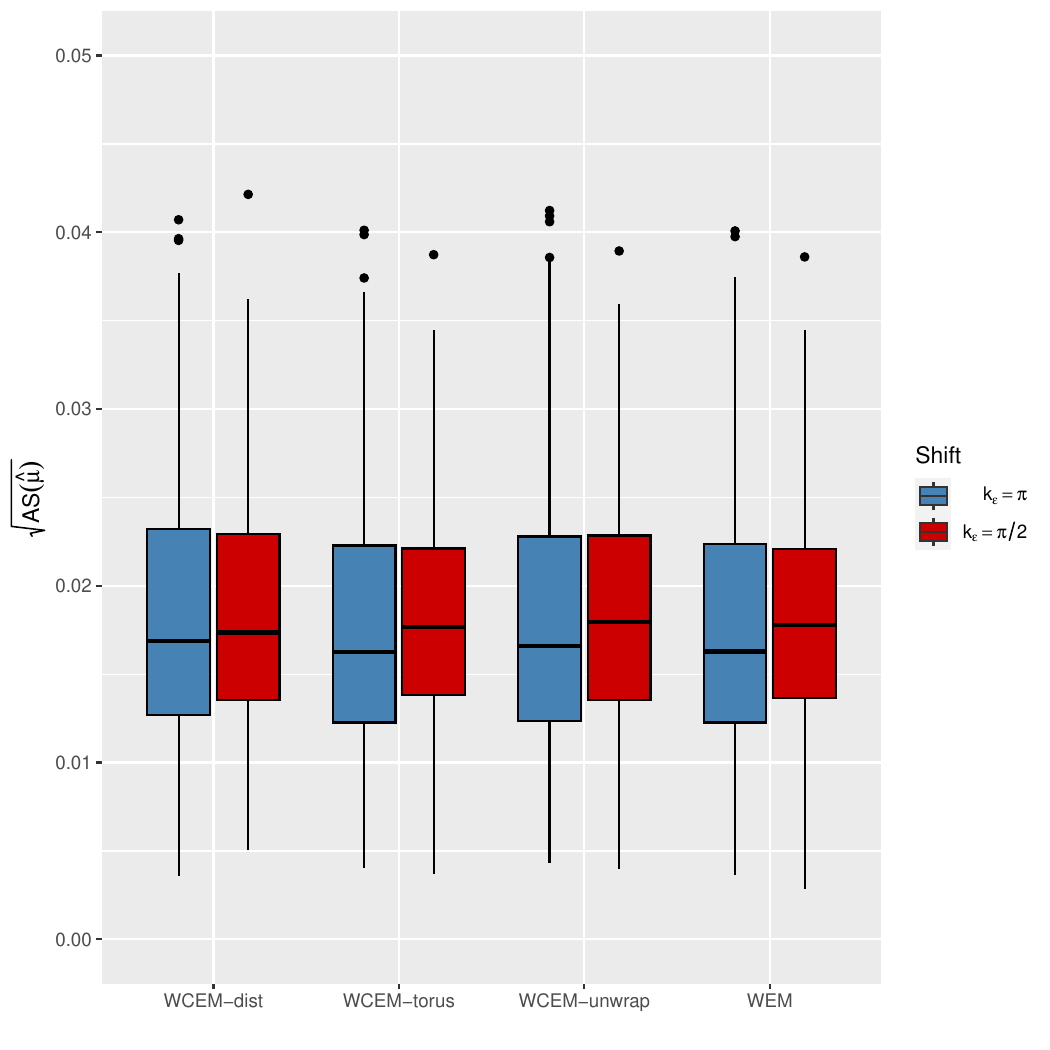}
\includegraphics[scale=0.28]{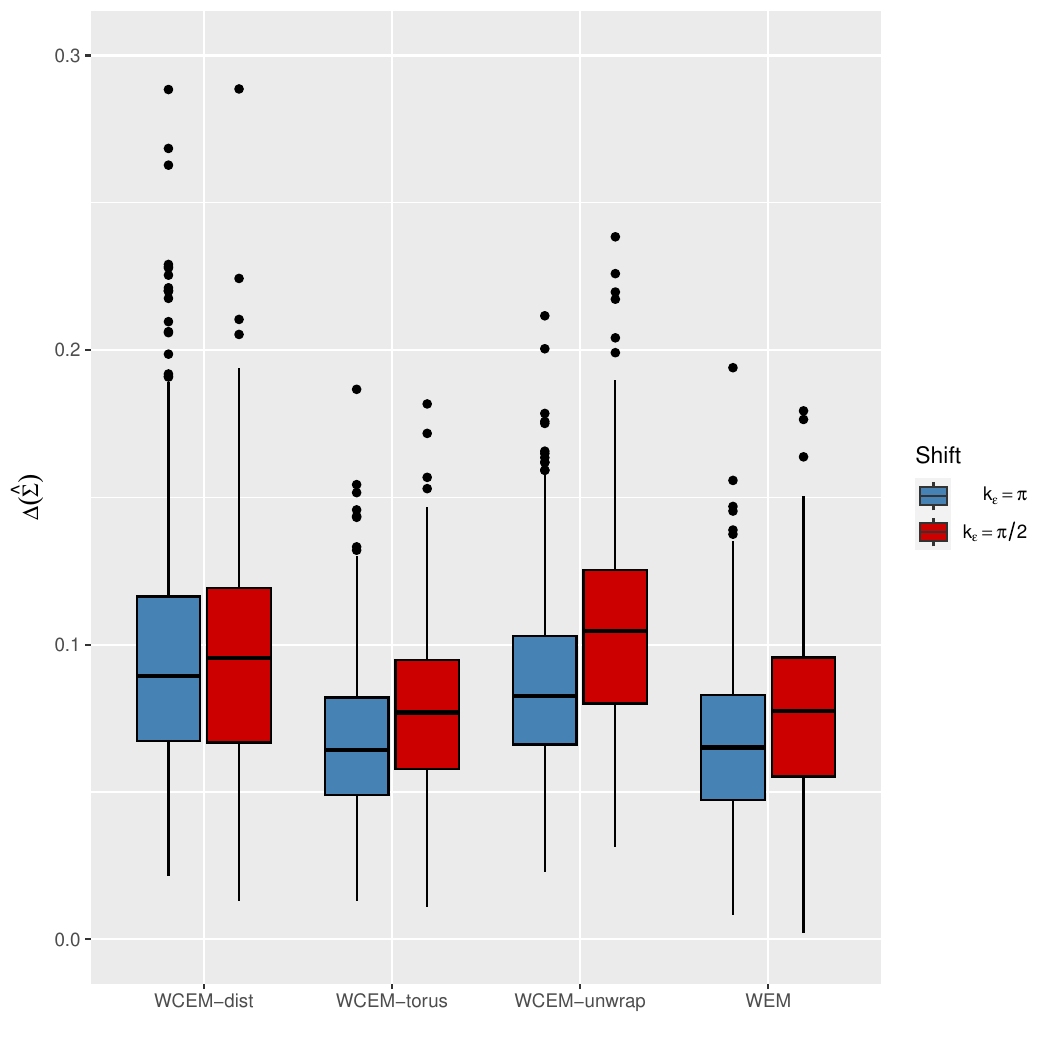}\\
\includegraphics[scale=0.28]{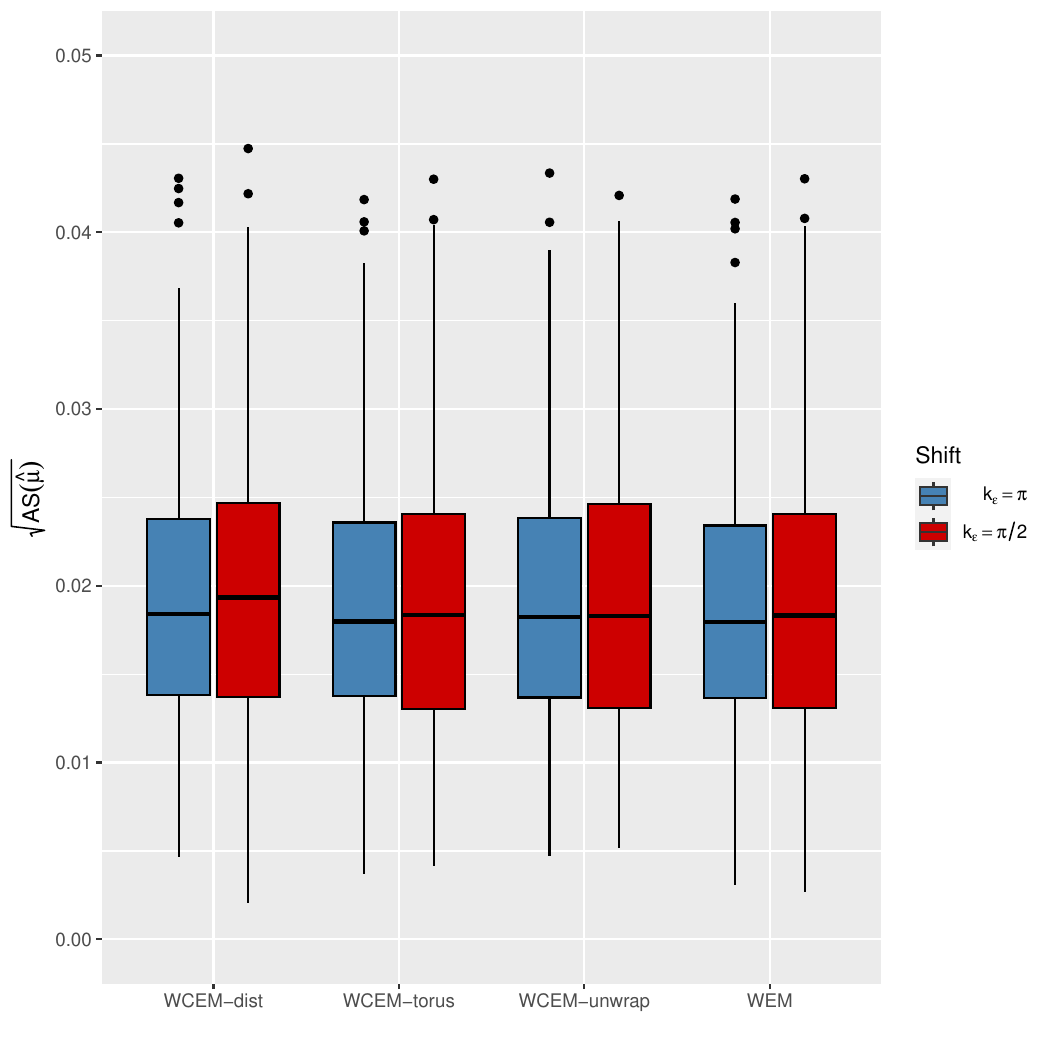}
\includegraphics[scale=0.28]{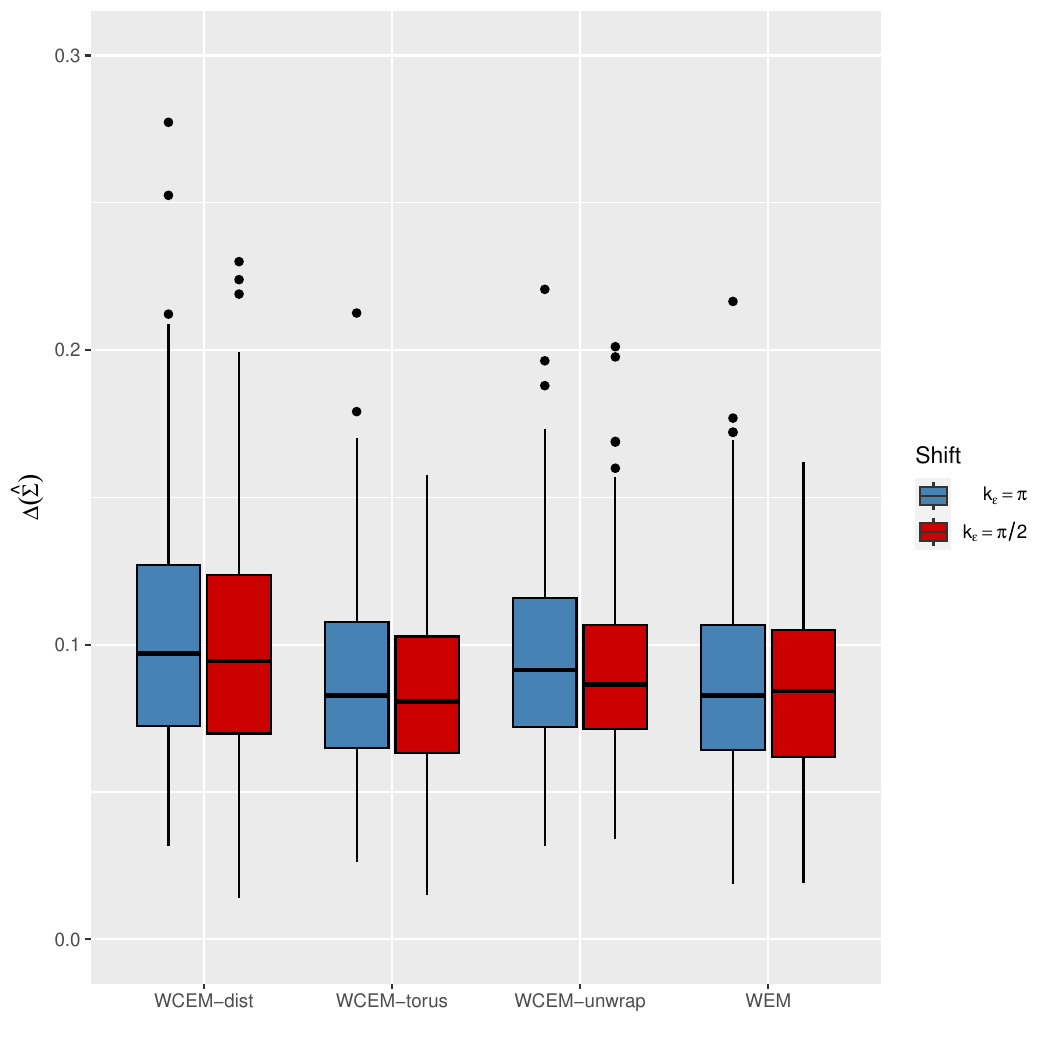}
	\caption{Box-plots for $\sqrt{AS(\hat{\vect{\mu}})}$ (left) and $\Delta(\hat\Sigma)$ (right) for $p=5$, $\sigma=\pi/8$, $k_\epsilon=\pi/2, \pi$ when $\epsilon=10\%$ (top) and $\epsilon=20\%$ (bottom).} 
	\label{figtab3a}
\end{figure*}

\begin{figure*}[t]
	\centering
	\includegraphics[scale=0.28]{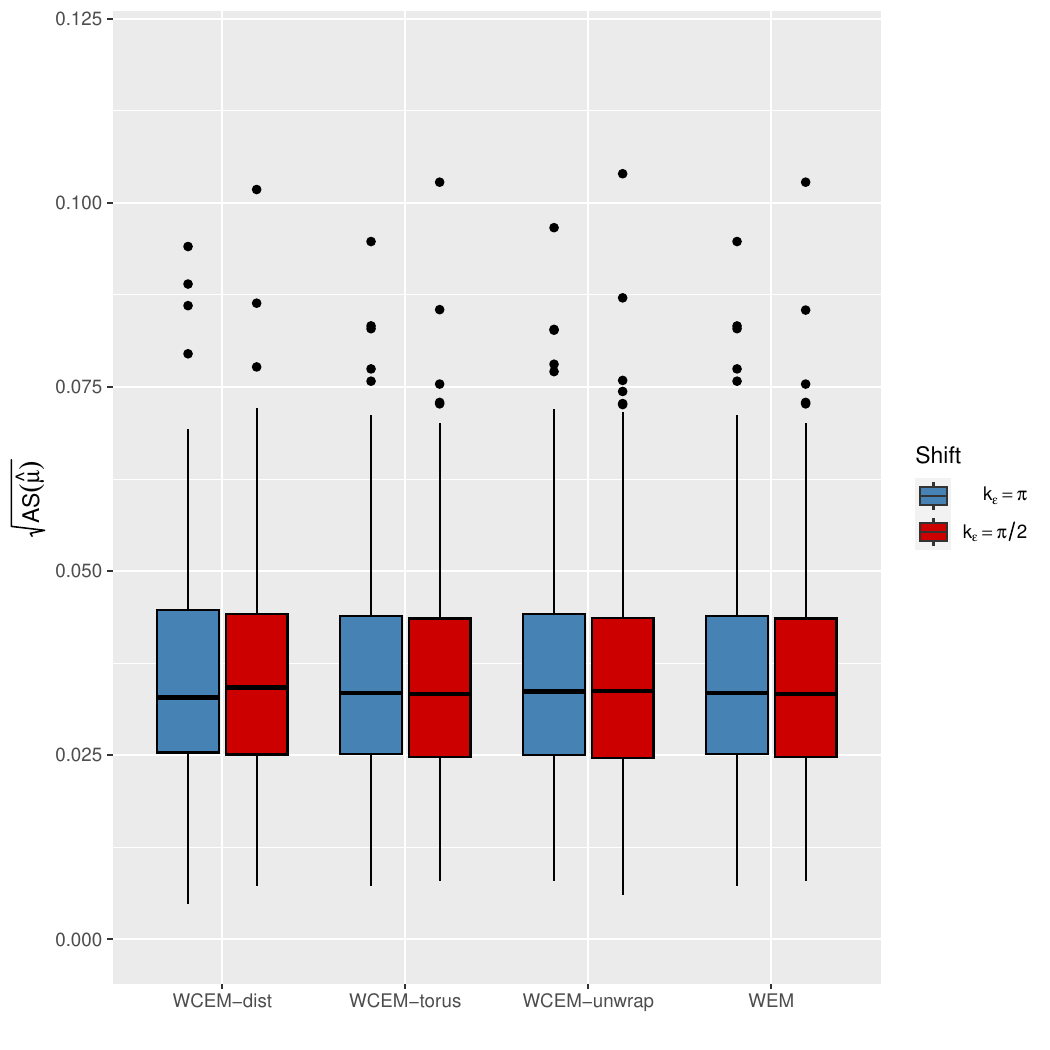}
	\includegraphics[scale=0.28]{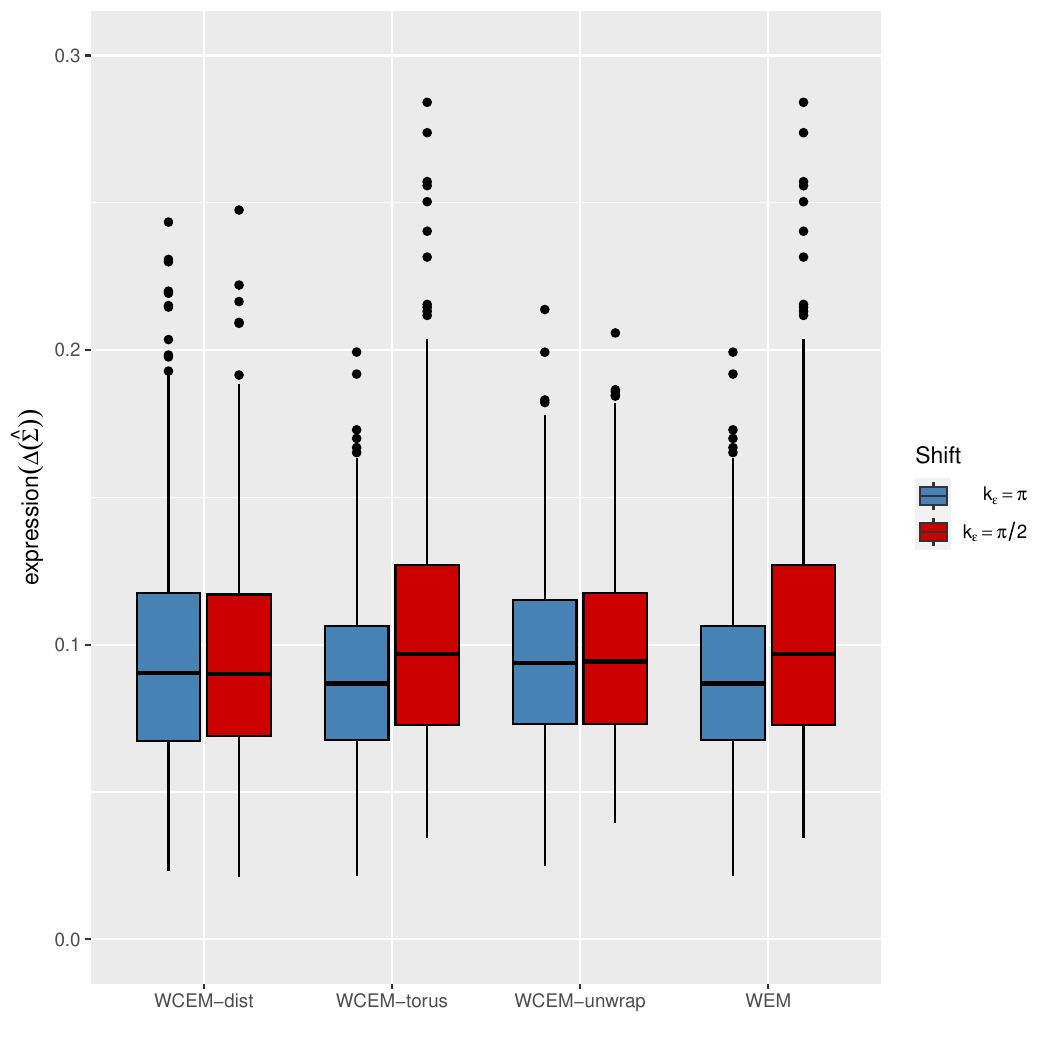}\\
	\includegraphics[scale=0.28]{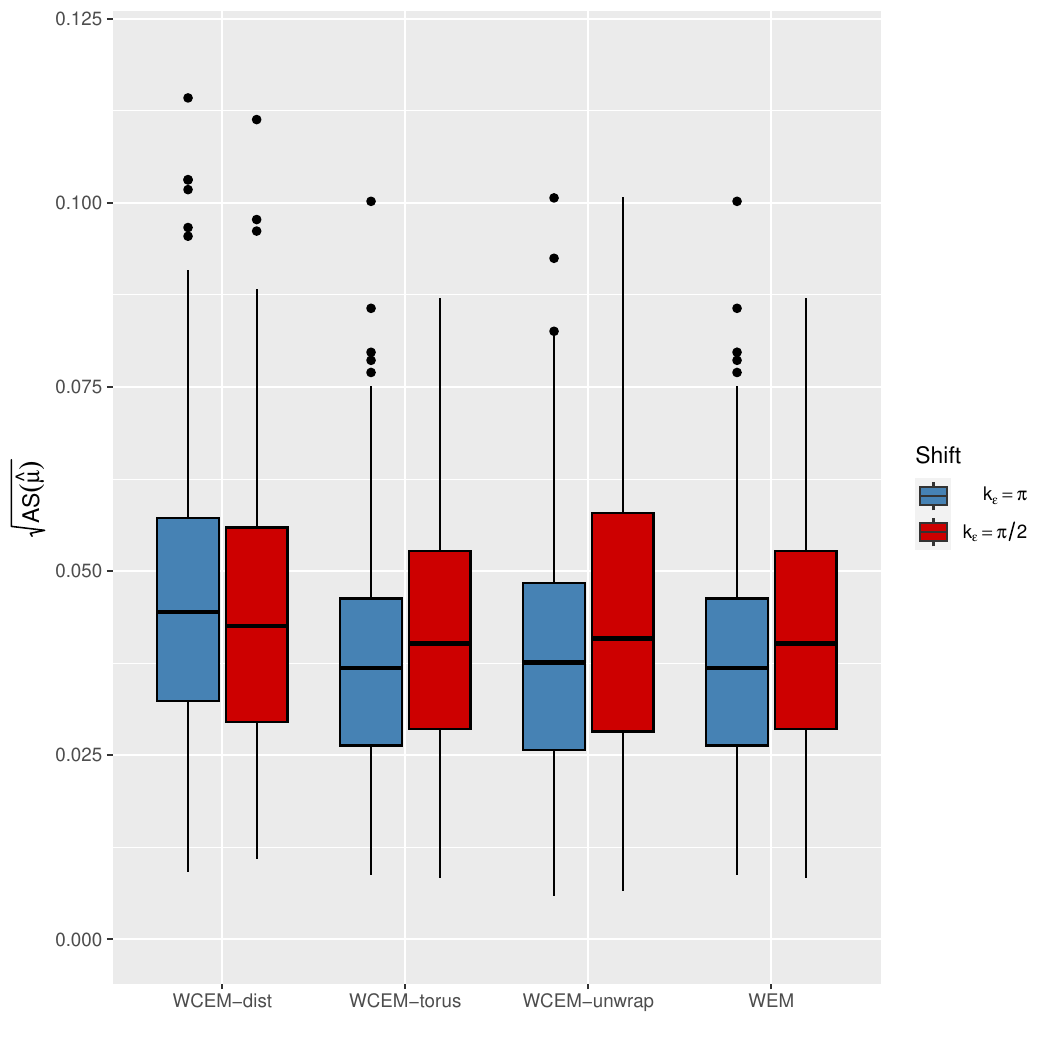}
	\includegraphics[scale=0.28]{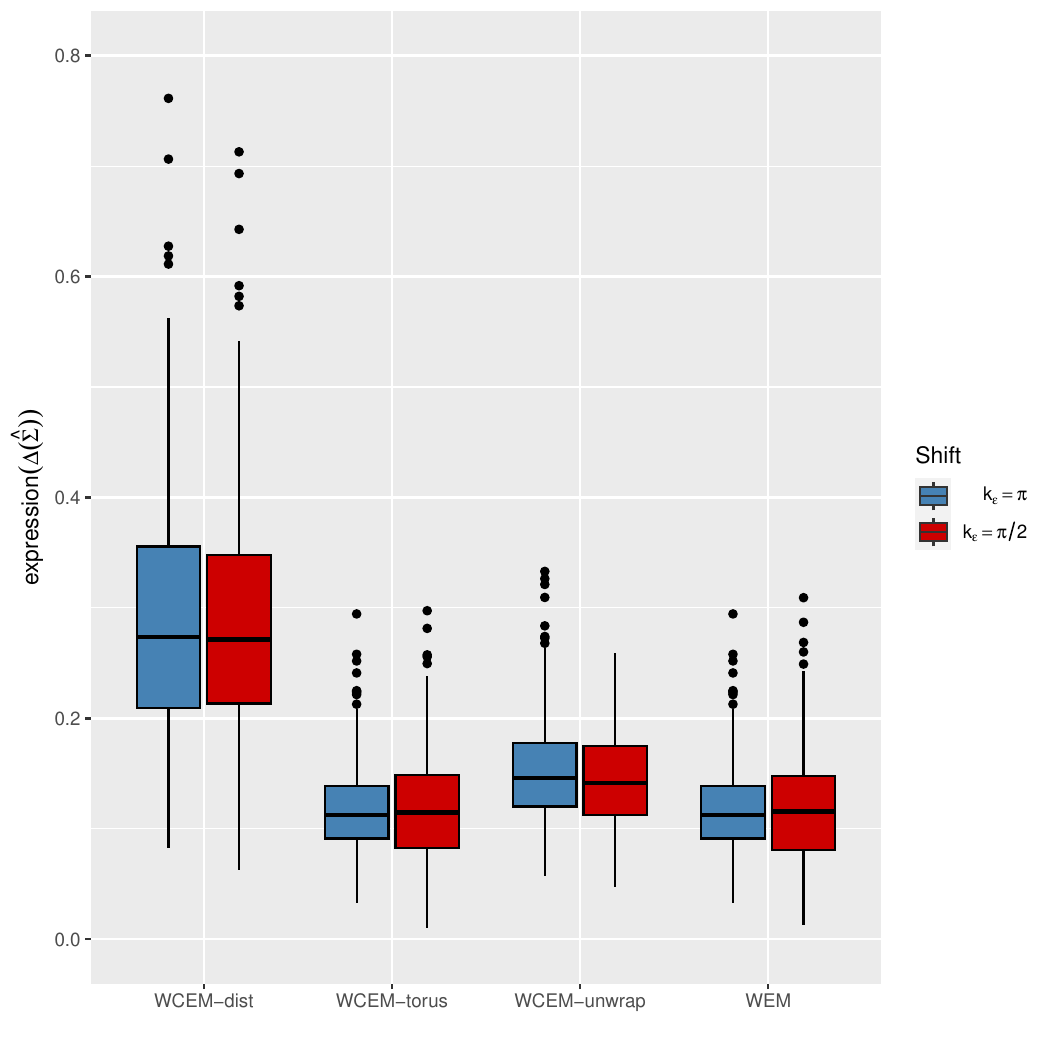}
	\caption{Box-plots for $\sqrt{AS(\hat{\vect{\mu}})}$ (left) and $\Delta(\hat\Sigma)$ (right) for $p=5$, $\sigma=\pi/4$, $k_\epsilon=\pi/2, \pi$ when $\epsilon=10\%$ (top) and $\epsilon=20\%$ (bottom).} 
	\label{figtab3b}
\end{figure*}

\section{Real data examples}
\label{sec:6}

\subsection{8TIM protein data}
Let us consider the 8TIM protein data described in Section \ref{sec:0}. 
We compare the results from maximum likelihood estimation and its robust counterparts based on weighted likelihood estimation under the WN model assumption. We use the same notation introduced in Section \ref{sec:5} to denote the different estimates.
The data and the fitted models given by the EM and WCEM-unwrap based on (\ref{residualfs2}) are shown in Figure \ref{fig:protein1}: the Ramachandran plot of the angles over $[0,2\pi)\times [0,2\pi)$ is given in the left panel, whereas data are displayed on a flat torus in the right panel, to account for their cyclic topology. The results from the WEM or WCEM-torus are indistinguishable.
In both panels the fitted models are represented through tolerance ellipses based on the $0.99-$level quantile of the $\chi^2_2$ distribution. 
The data clearly show a multi-modal clustered pattern.
Actually, the robust analyses give strong indication of the presence of several clusters: they  
all disclose the presence of different structures, otherwise undetectable by maximum likelihood estimation. The tolerance ellipses corresponding to the robustly fitted WN distribution enclose those points in  the most dense area, whereas the others are severely down-weighted. 
There is strong agreement with the findings from the analysis in \cite{bambi}.  
In the left panel of Figure \ref{fig:protein2} we displayed the weights from the WCEM-unwrap algorithm. 
According to an outliers detection testing rule performed at a significance level $\alpha=0.01$, the actual rate of contamination is about $46\%$. 
The right panel of Figure \ref{fig:protein2} shows the corresponding distance plot based on robust distances. The horizontal line gives the (square root) $\chi^2_{0.99, 2}$ cut-off. Figure \ref{fig:d} shows genuine points and outliers on the torus.

The clustered structure of the data suggested by the outcome of the robust analyses can be further explored using a monitoring plot of the weights as the bandwidth $h$ varies on a chosen grid of values. In this example, the bandwidth matrix is $H=\textrm{diag}(h^2)$. The vertical line gives the bandwidth actually used. 
The dark trajectories in Figure \ref{fig:protein3} correspond to those points receiving a large weight in the robust analysis, whereas the gray lines refer to the other points.
For small values of the bandwidth $h$, at least two groups can be detected. As $h$ increases,  we notice a transition from the robust to a non robust fit since many other observations are attached large weights and the size of global down-weighting reduces. In particular, some data points exhibit very steep trajectories, as they are no more down-weighted from some point ahead. This behavior suggest the presence of a second group of observations.  A closer look at Figure \ref{fig:protein3} also unveils a third group, that is composed by those points whose weight is still low for large values of the bandwidth on the right end part of the plot. These points highlight features that are not assimilable to the previous groups. 
Hence, the robust analysis indicates at least three groups. This finding is 
confirmed by the results stemming from a proper model based clustering of the torus data at hand \citep{greco2022mix}, whose cluster assignments are displayed in Figure \ref{fig:protein4} and Figure \ref{fig:e} .

\begin{figure*}[t]
	\centering
	\includegraphics[scale=0.28]{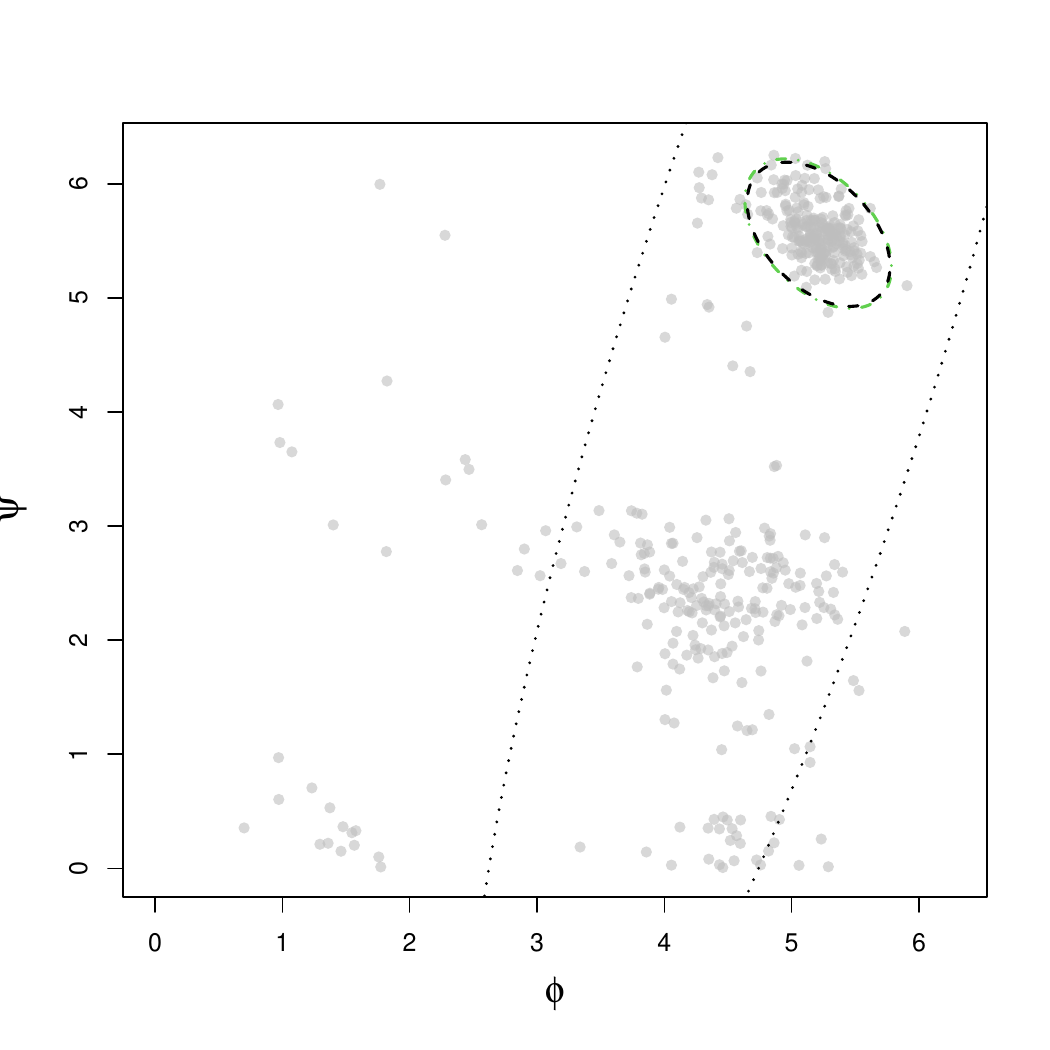}
	\includegraphics[scale=0.28]{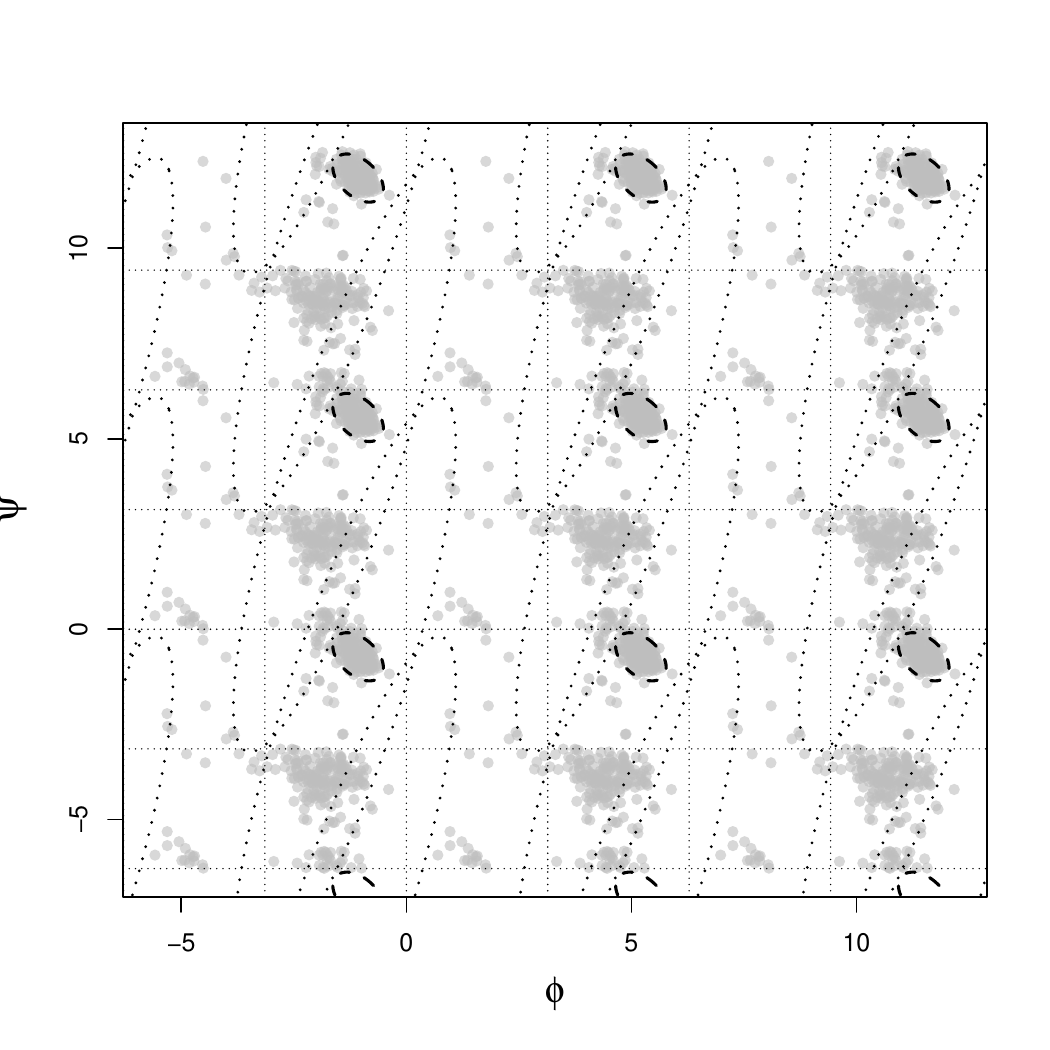}
	\caption{8TIM protein data. Left panel: Ramachandran plot. Right panel: unwrapped data on a flat torus. $99\%$ tolerance ellipses overimposed: robust fit (dashed line), maximum likelihood estimation (dotted line).} 
	\label{fig:protein1}
\end{figure*}

\begin{figure*}[t]
	\centering
	\includegraphics[scale=0.28]{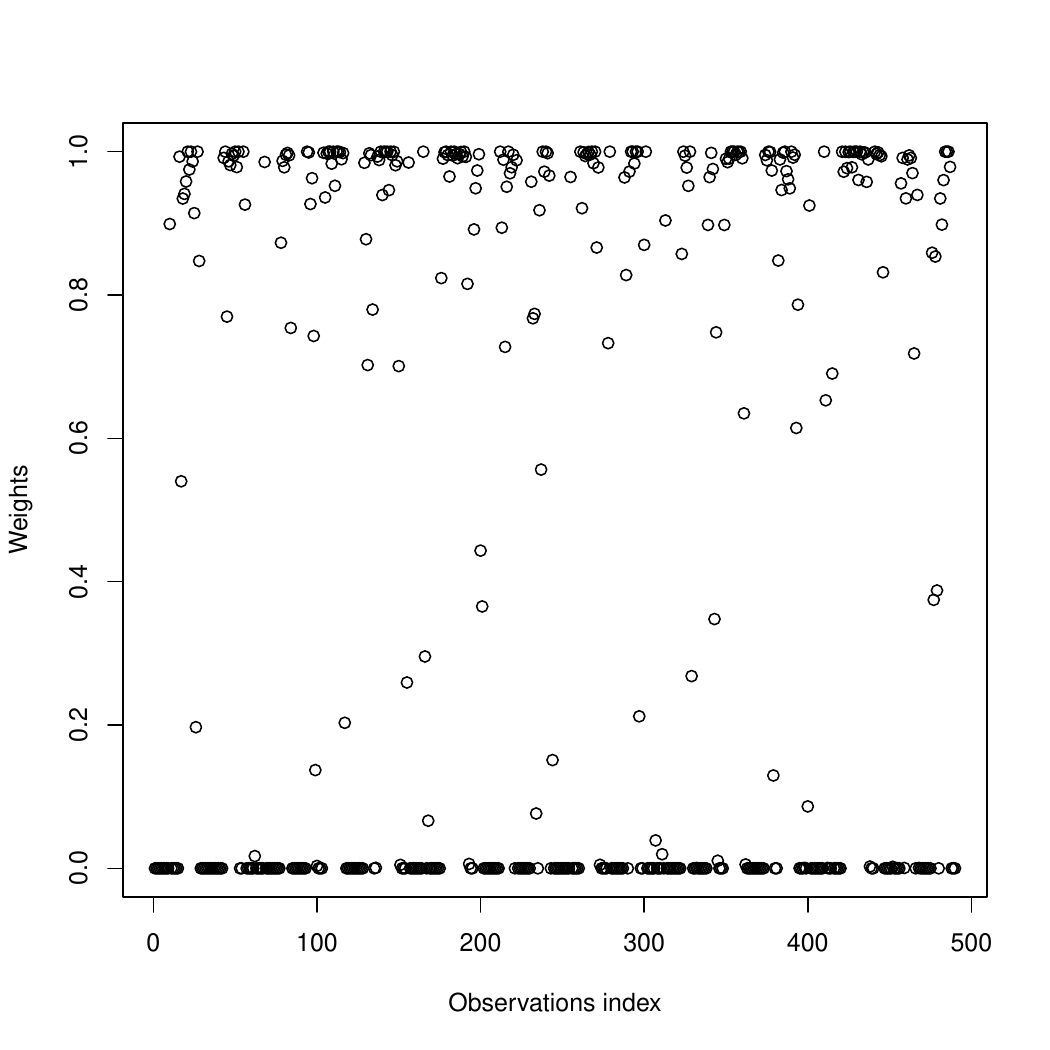}
	\includegraphics[scale=0.28]{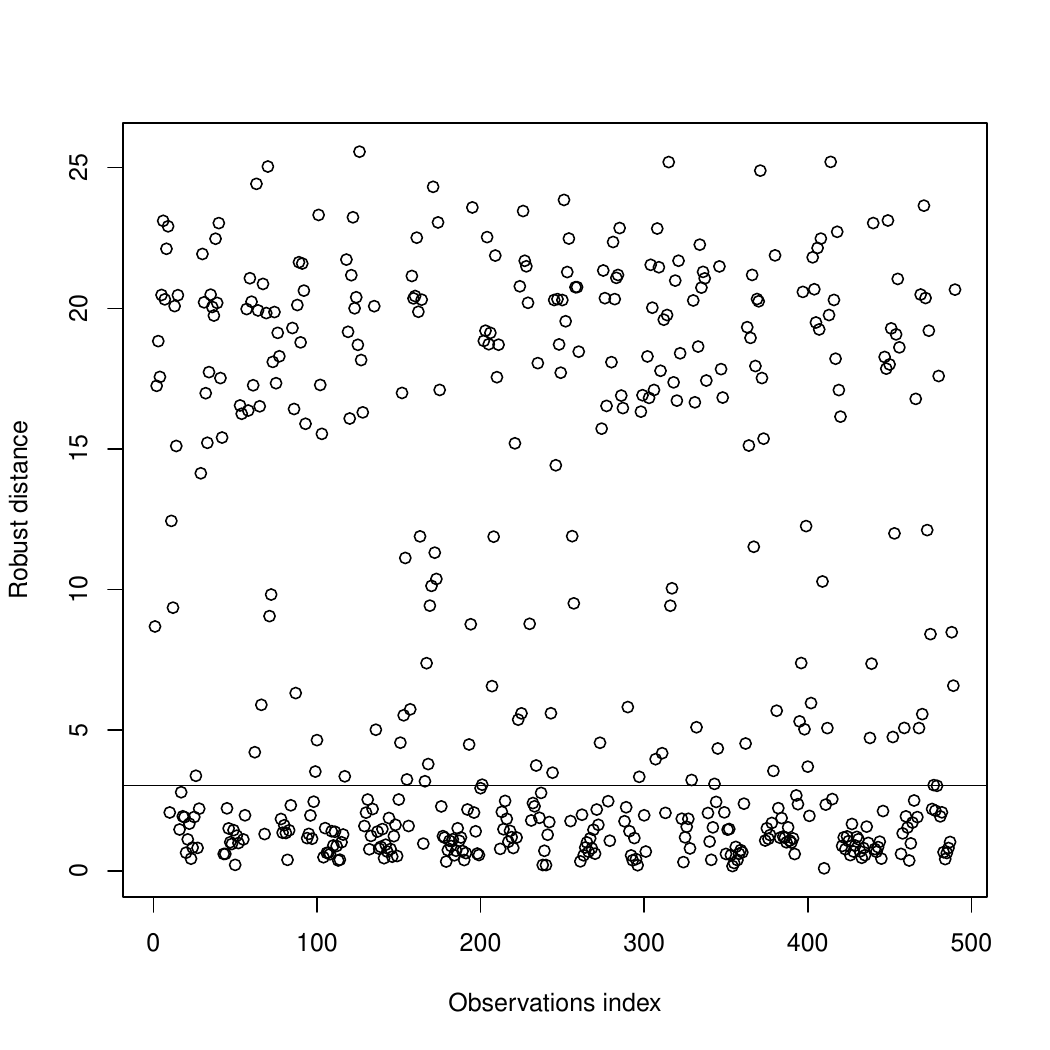}
	\caption{8TIM protein data. Left panel: weights. Right panel: robust distances. the horizontal line gives the square root of the $0.99$-level quantile of the $\chi^2_2$ distribution.}
	\label{fig:protein2}
\end{figure*}

\begin{figure*}[t]
	\centering
	\includegraphics[scale=0.65]{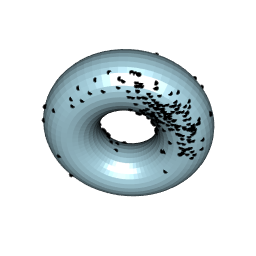}
	\includegraphics[scale=0.65]{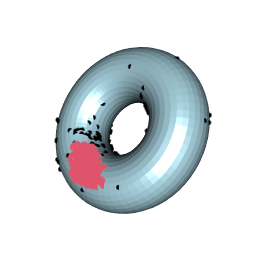}
	\caption{8TIM protein data. Bivariate angles as points on the surface of a torus from two different perspectives: genuine observations correspond to (red) larger dots, the remaining are outliers.} 
	\label{fig:d}
\end{figure*}

\begin{figure*}[t]
	\centering
	\includegraphics[scale=0.3]{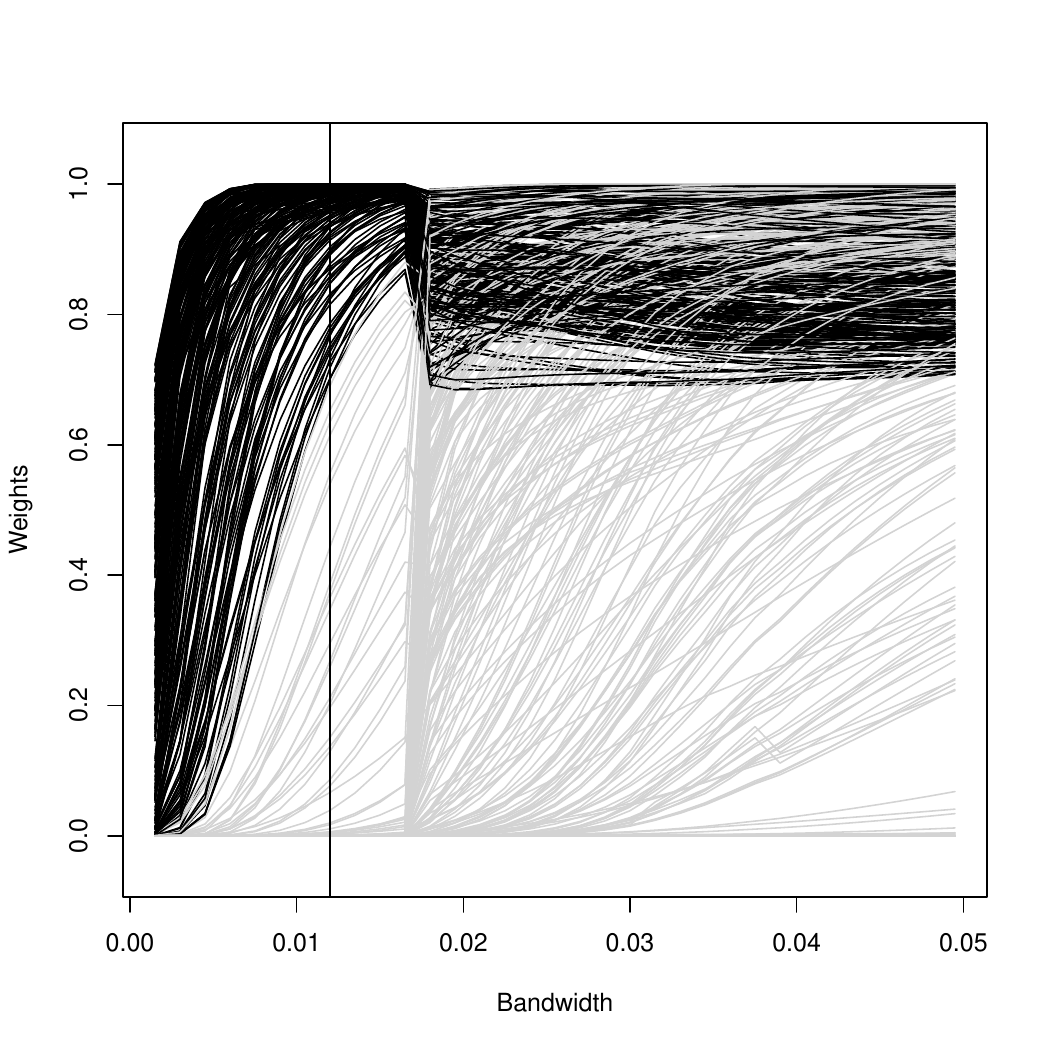}
	\caption{8TIM protein data. Monitoring plot of weights from the robust fit. The vertical line gives the selected bandwidth value.}
	\label{fig:protein3}
\end{figure*}

\begin{figure*}[t]
	\centering
	\includegraphics[scale=0.3]{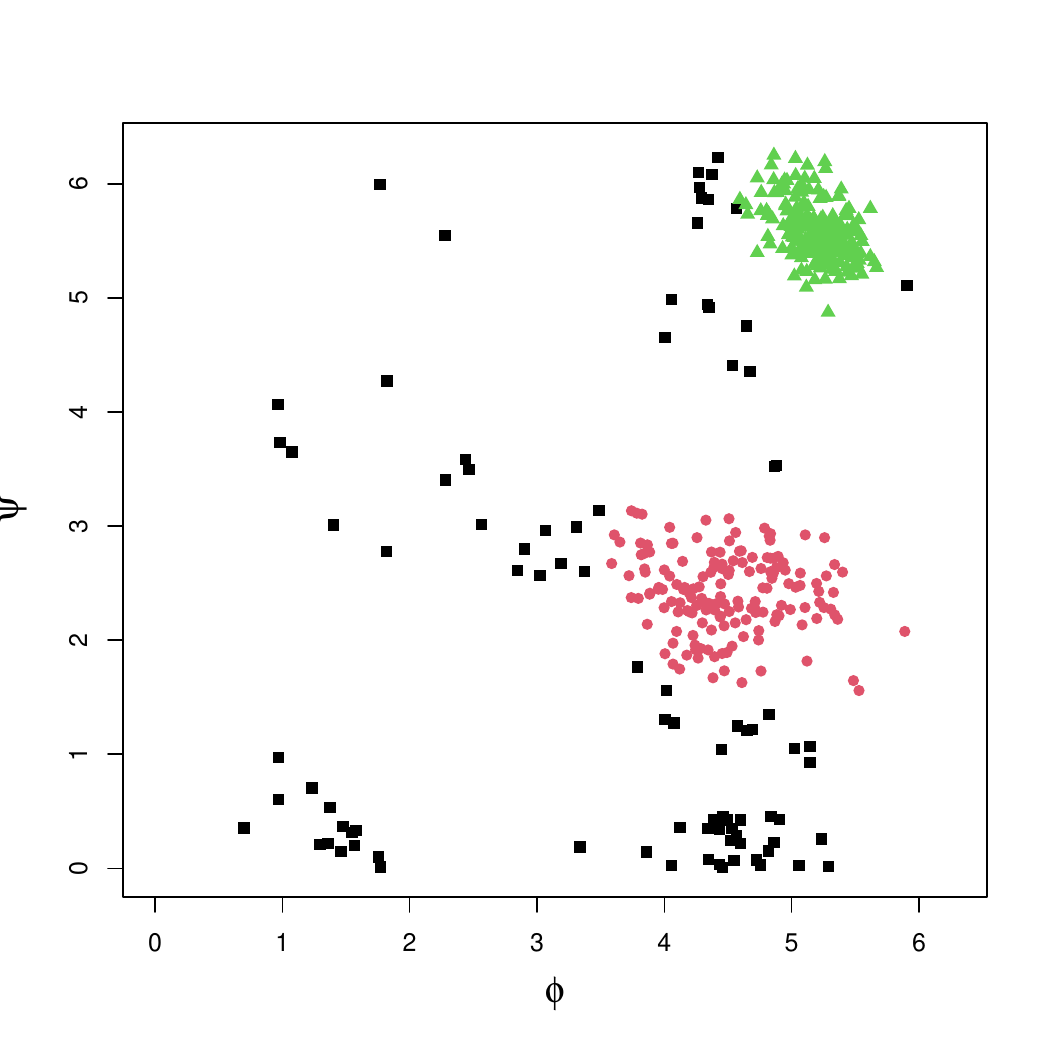}
	\caption{8TIM protein data. Model based clustering. }
	\label{fig:protein4}
\end{figure*}

\begin{figure*}[t]
	\centering
	\includegraphics[scale=0.65]{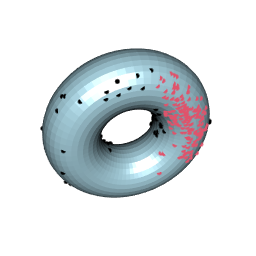}
	\includegraphics[scale=0.65]{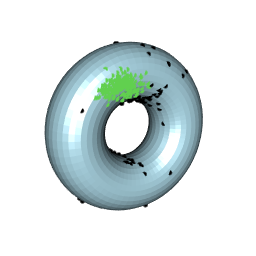}
	\caption{8TIM protein data. Model based clustering on the torus.} 
	\label{fig:e}
\end{figure*}

\subsection{RNA data}
RNA is assembled as a chain of nucleotides that constitutes a single strand folded onto itself. A nucleotide contains the five-carbon sugar deoxyribose, a nucleobase, that is a nitrogenous base,  and one phosphate group. Then, each nucleotide in RNA molecules presents seven torsion angles: six dihedral angles and one angle for the base. Data have been taken from the large RNA data set \citep{wadley2007evaluating}. Here, we consider a sub-sample of size $n=260$, obtained after joining data from two distinct clusters, whose size are 232 and 28, respectively, and we neglect the information about group labels in the fitting process.  Since the sizes of two clusters are very unbalanced, a feasible robust method is expected to fit the majority of the data belonging to the larger cluster and to lead to detect the data from the smaller cluster as outliers, as they share a different pattern. Figure \ref{fig:protein5} gives the distance plot from WCEM-torus, WCEM-unwrap and WCEM-dist, under the WN model. We do not appreciate noticeable differences among the results.
Each technique leads to detect the smaller group, denoted by black dots.
Actually, in this case, the outcome from the robust analysis allows to cope with an unsupervised classification problem and to discriminate between the two groups, with a satisfactory balance between swamping and power. \\

\begin{figure*}[t]
	\centering
	\includegraphics[scale=0.28]{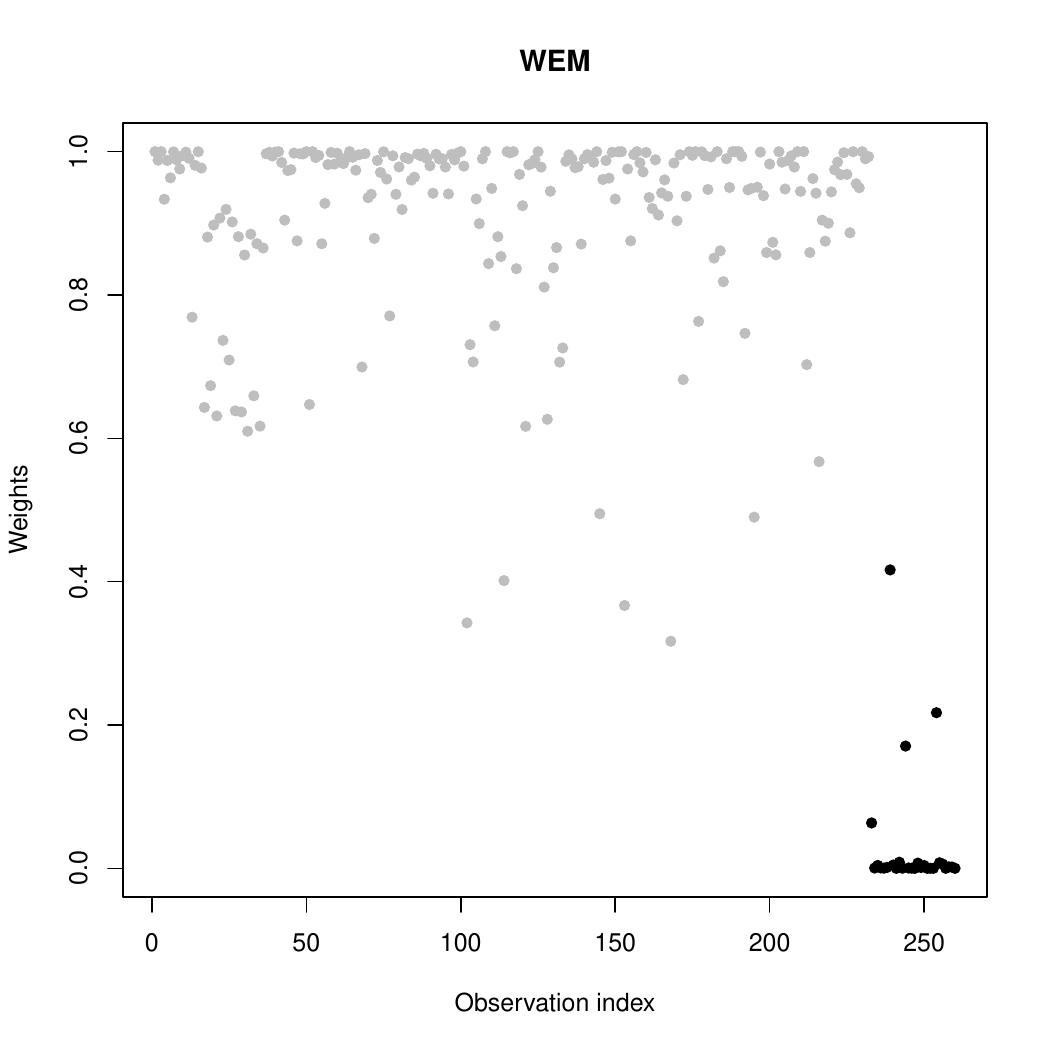}
	\includegraphics[scale=0.28]{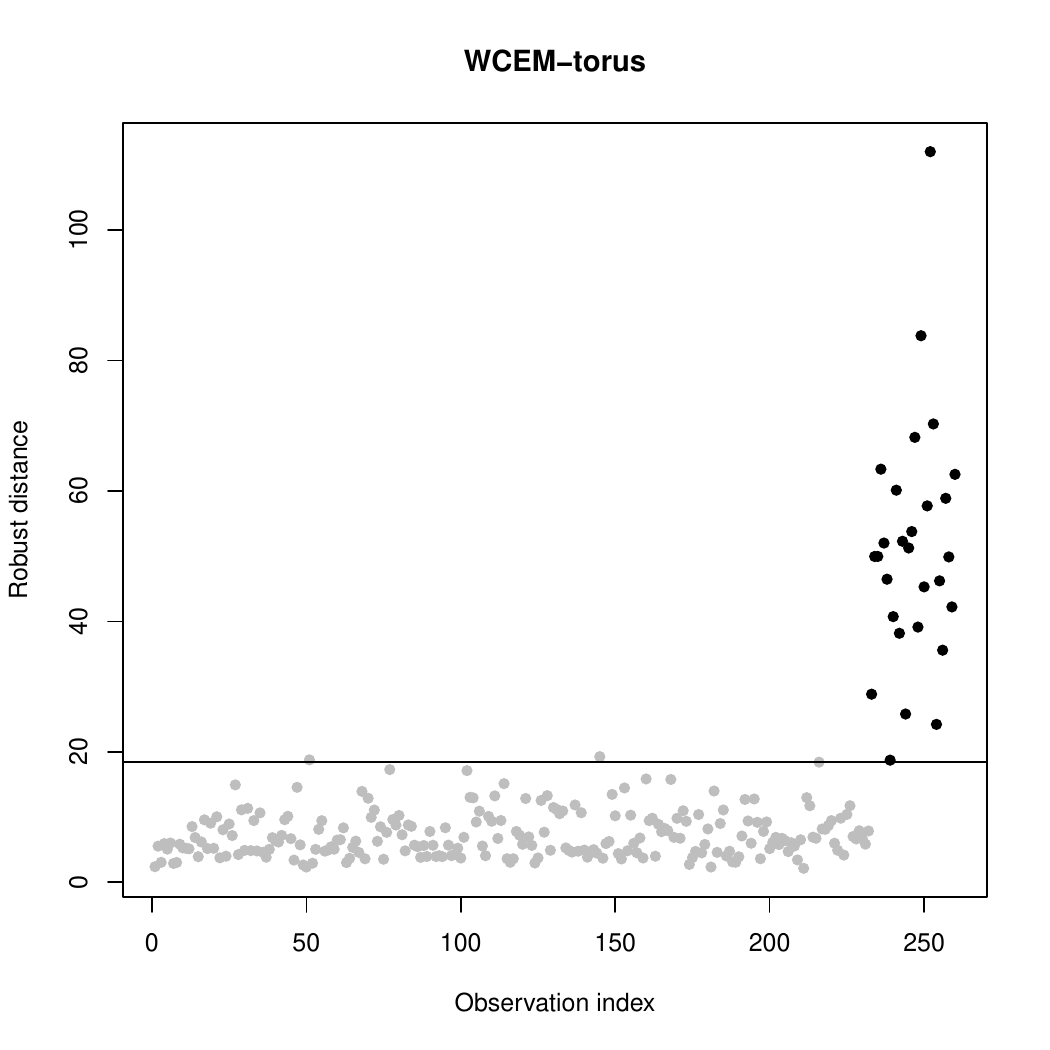}
	\includegraphics[scale=0.28]{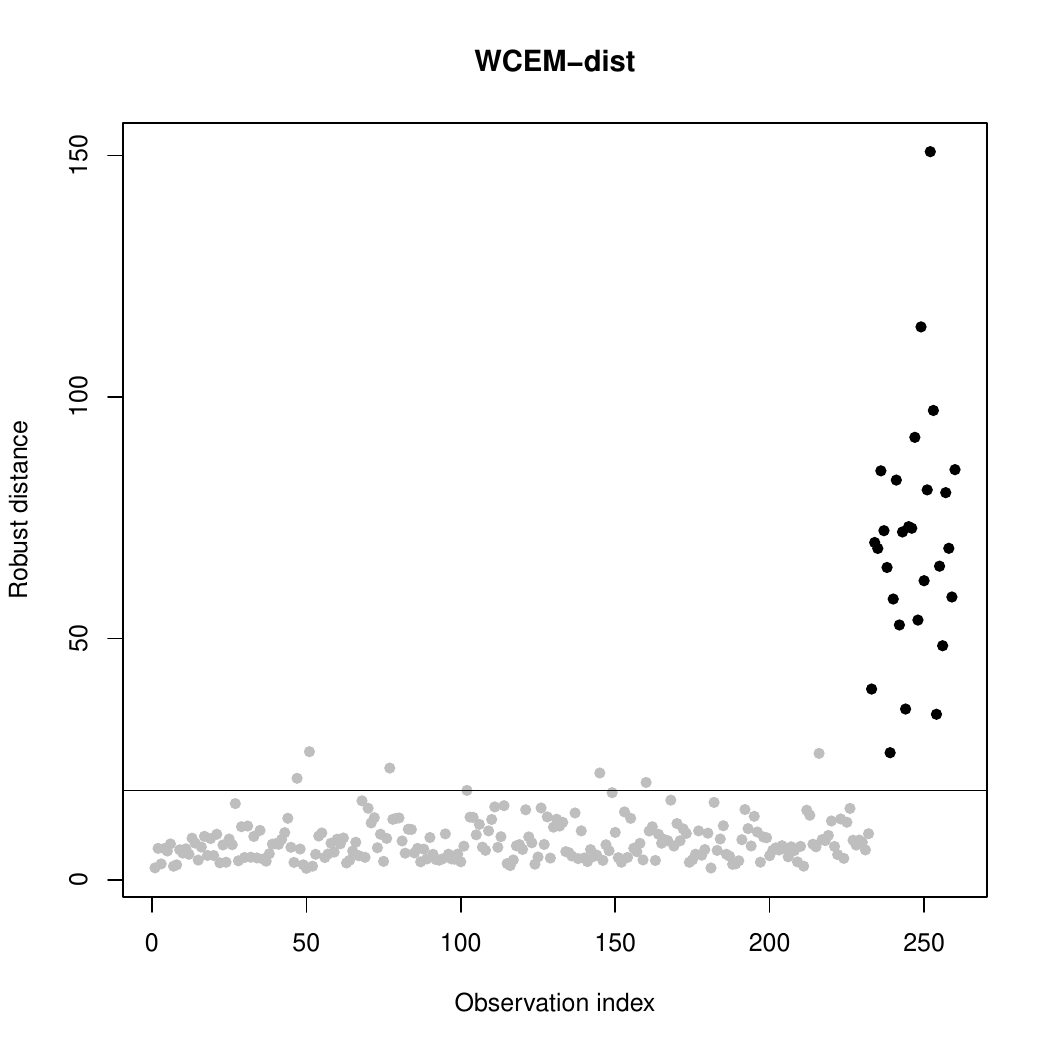}
	\includegraphics[scale=0.28]{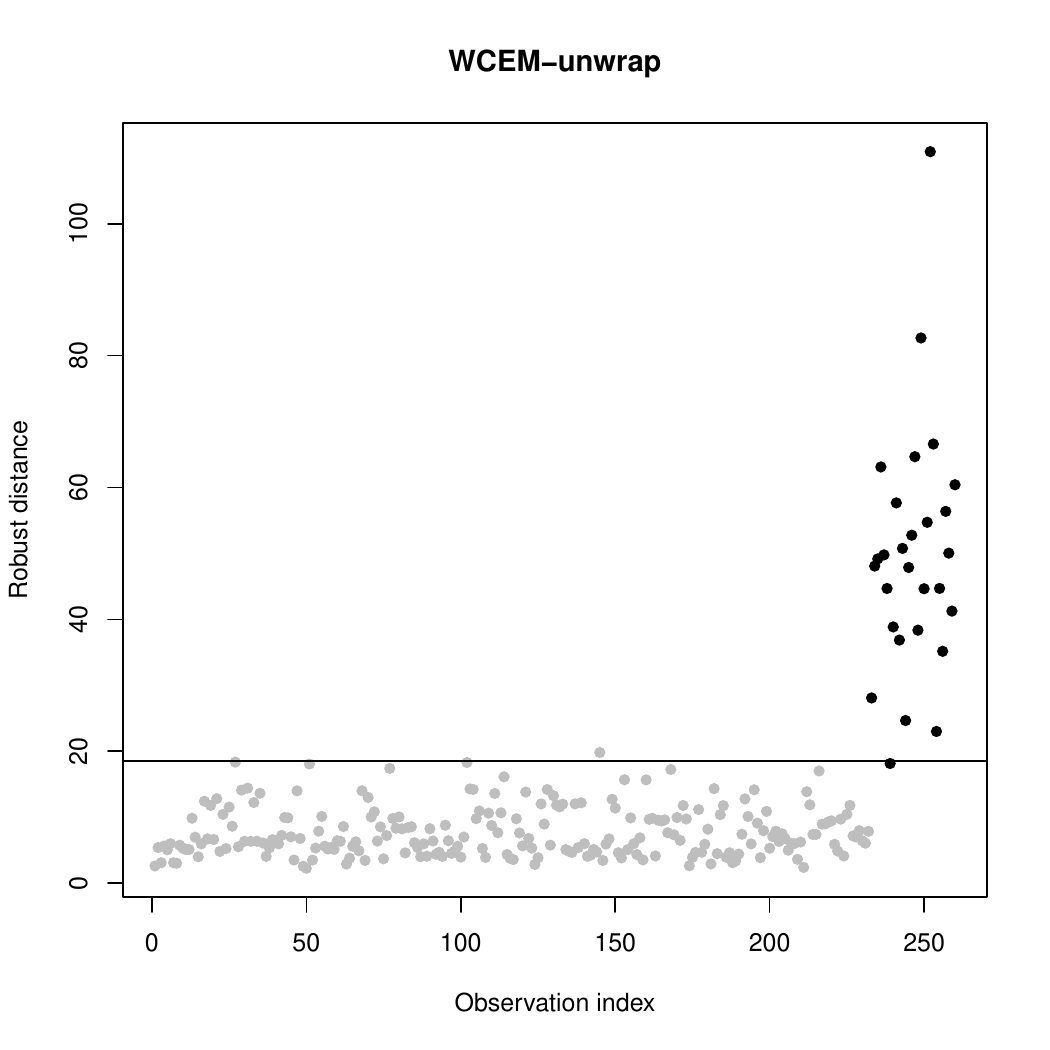}
	\caption{RNA data. Weights returned by the WEM (left-right panle). Squared distance plots for the different weighting scheme as given by the WCEM-torus, WCEM-unwrap and WCEM-dist (clockwise in the other panels). Black dots give points from the smaller {\it outlying} cluster. The horizontal line gives the $0.99$-level quantile of the $\chi^2_7$ distribution.}
	\label{fig:protein5}
\end{figure*}

\appendix

\section{MLE for wrapped unimodal elliptically symmetric distributions}

\label{appendix:mle}

Let us consider the circular model $$m^\circ(\vect{y}; \vect{\mu}, \Sigma) = \sum_{\vect{j} \in \mathbb{Z}^p} m(\vect{y} + 2\pi \vect{j}; \vect{\mu}, \Sigma)$$ where
\begin{equation*}
m(\vect{x}; \vect{\theta}) \propto \vert \Sigma \vert ^{-1/2} h\left((\vect{x} - \vect{\mu})^\top \Sigma^{-1} (\vect{x} - \vect{\mu})\right)
\end{equation*}
is a unimodal elliptically symmetric distribution. 
The log-likelihood function based on an  i.i.d. sample $\vect{y}_1, \ldots, \vect{y}_n$ is
\begin{align*}
\ell^\circ(\vect{\mu}, \Sigma) & = \sum_{i=1}^n \log m^\circ(\vect{y}_i; \vect{\mu}, \Sigma) \\
& = \sum_{i=1}^n \log \sum_{\vect{j} \in \mathbb{Z}^p} m(\vect{y}_i + 2\pi \vect{j}; \vect{\mu}, \Sigma) \\
& \propto \sum_{i=1}^n \log \sum_{\vect{j} \in \mathbb{Z}^p} \vert \Sigma\vert^{-\frac{1}{2}} h\left[(\vect{y}_i + 2\pi \vect{j} - \vect{\mu})^\top \Sigma^{-1} (\vect{y}_i + 2\pi \vect{j} - \vect{\mu})\right] \\
& = \frac{n}{2} \log \vert \Sigma^{-1} \vert + \sum_{i=1}^n \log \sum_{\vect{j} \in \mathbb{Z}^p} h\left[\trace \left((\vect{y}_i + 2\pi \vect{j} - \vect{\mu}) (\vect{y}_i + 2\pi \vect{j} - \vect{\mu})^\top \Sigma^{-1} \right) \right]\\
\end{align*}
Recall that for given square matrices $A$ and $B$, both symmetric and positive definite we have that
\begin{enumerate}
\item $\nabla_{A} \trace(BA) = B^\top$,
\item  $\nabla_A \log(\vert A \vert) = \left(A^{-1}\right)^\top$, 
\item  $\nabla_{\vect{x}} (\vect{x}^\top A \vect{x}) = 2 A \vect{x}$ .
\end{enumerate}
Let $d_{i\vect{j}}(\vect{\mu},\Sigma) = (\vect{y}_i + 2\pi \vect{j} - \vect{\mu})^\top \Sigma^{-1}(\vect{y}_i + 2\pi \vect{j} - \vect{\mu})$. 
Taking the derivatives w.r.t. $\vect{\mu}$ and $\Sigma^{-1}$, the likelihood equations are
\begin{align*}
\nabla_{\vect{\mu}} \ell^\circ(\vect{\mu}, \Sigma) & = \sum_{i=1}^n \nabla_{\vect{\mu}}  \log \sum_{\vect{j} \in \mathbb{Z}^p} h(d_{i\vect{j}}(\vect{\mu},\Sigma)) \\
& = \sum_{i=1}^n \frac{\sum_{\vect{j} \in \mathbb{Z}^p} \nabla_{\vect{\mu}} h(d_{i\vect{j}}(\vect{\mu},\Sigma))}{\sum_{\vect{k} \in \mathbb{Z}^p} h(d_{i\vect{k}}(\vect{\mu},\Sigma))} \\
& = 2 \sum_{i=1}^n \frac{\sum_{\vect{j} \in \mathbb{Z}^p} h^\prime(d_{i\vect{j}}(\vect{\mu},\Sigma)) \Sigma^{-1} (\vect{y}_i + 2\pi \vect{j} - \vect{\mu})}{\sum_{\vect{k} \in \mathbb{Z}^p} h(d_{i\vect{k}}(\vect{\mu},\Sigma))} 
\end{align*}
and
\begin{align*}
\nabla_{\Sigma^{-1}} \ell^\circ(\vect{\mu}, \Sigma) & = \frac{n}{2} \Sigma^\top + \sum_{i=1}^n \nabla_{\Sigma^{-1}} \log \sum_{\vect{j} \in \mathbb{Z}^p} h(d_{i\vect{j}}(\vect{\mu},\Sigma)) \\
& = \frac{n}{2} \Sigma + \sum_{i=1}^n \frac{\sum_{\vect{j} \in \mathbb{Z}^p} \nabla_{\Sigma^{-1}} h(d_{i\vect{j}}(\vect{\mu},\Sigma))}{\sum_{\vect{k} \in \mathbb{Z}^p} h(d_{i\vect{k}}(\vect{\mu},\Sigma))} \\
& = \frac{n}{2} \Sigma + \sum_{i=1}^n \frac{\sum_{\vect{j} \in \mathbb{Z}^p} h^\prime(d_{i\vect{j}}(\vect{\mu},\Sigma))  (\vect{y}_i + 2\pi \vect{j}-\vect{\mu})  (\vect{y}_i + 2\pi \vect{j}-\vect{\mu})^\top}{\sum_{\vect{k} \in \mathbb{Z}^p} h(d_{i\vect{k}}(\vect{\mu},\Sigma))} \\
\end{align*}
where $h^\prime(d) = \partial h(d)/\partial d$. 
Let
\begin{equation*}
v_{i\vect{j}} = \frac{h^\prime(d_{i\vect{j}}(\vect{\mu},\Sigma))}{\sum_{\vect{k} \in \mathbb{Z}^p} h(d_{i\vect{k}}(\vect{\mu},\Sigma))} \ .
\end{equation*}  
Then, the MLE $(\hat{\vect{\mu}}, \hat\Sigma)$ is the solution to the (set of) fixed point equations
\begin{align*}
\vect{\mu} & = \frac{\sum_{i=1}^n \sum_{\vect{j} \in \mathbb{Z}^p} v_{i\vect{j}}(\vect{y}_i + 2\pi \vect{j})}{\sum_{i=1}^n \sum_{\vect{k} \in \mathbb{Z}^p} v_{i\vect{k}}} \\
\Sigma & = -\frac{2}{n} \sum_{i=1}^n \sum_{\vect{j} \in \mathbb{Z}^p} v_{i\vect{j}} (\vect{y}_i + 2\pi \vect{j}-\vect{\mu}) (\vect{y}_i + 2\pi \vect{j}-\vect{\mu})^\top \ .
\end{align*}
The WN distribution corresponds to $h(t) = \exp\left(-\frac{t}{2} \right)$. Since $h^\prime(t) = -\frac{1}{2} h(d)$ then
\begin{equation*}
v_{i\vect{j}} = -\frac{1}{2} \frac{h(d_{i\vect{j}})}{\sum_{\vect{k} \in \mathbb{Z}^p} h(d_{i\vect{k}})} = -\frac{1}{2} \frac{m(\vect{y}_i + 2\pi \vect{j}; \vect{\mu}, \Sigma)}{\sum_{\vect{k} \in \mathbb{Z}^p} m(\vect{y}_i + 2\pi \vect{k}; \vect{\mu}, \Sigma)} \ .
\end{equation*}
and the estimating equations simplify to
\begin{align*}
\vect{\mu} & = \frac{1}{n} \sum_{i=1}^n \sum_{\vect{j} \in \mathbb{Z}^p} \omega_{i\vect{j}} (\vect{y}_i + 2\pi \vect{j}) \\
\Sigma & = \frac{1}{n} \sum_{i=1}^n \sum_{\vect{j} \in \mathbb{Z}^p} \omega_{i\vect{j}} (\vect{y}_i + 2\pi \vect{j} - \vect{\mu})(\vect{y}_i + 2\pi \vect{j} - \vect{\mu})^\top \ . \\
\end{align*}
with
\begin{equation*}
\omega_{i\vect{j}} = \frac{m(\vect{y}_i + 2\pi \vect{j}; \vect{\mu}, \Sigma)}{\sum_{\vect{k} \in \mathbb{Z}^p} m(\vect{y}_i + 2\pi \vect{k}; \vect{\mu}, \Sigma)} \ .
\end{equation*}

\section{EM algorithm for WN estimation}
\label{appendix:em}

Given an i.i.d. sample $(\vect{y}_1, \ldots, \vect{y}_n)$ from a WN distribution, in the EM algorithm the wrapping coefficients $\vect{j}$ are considered as latent variables and the observed torus data $\vect{y}_i$s as being incomplete, that is $\vect{y}_i$ is assumed to be one component of the pair $(\vect{y}_i,\vect{\omega}_i)$, where $\vect{\omega}_i=(\omega_{i\vect{j}} : \vect{j} \in \mathbb{Z}^p)$ is the associated latent wrapping coefficients label vector. Then, the MLE for $\vect{\theta} = (\vect{\mu}, \Sigma)$ is the result of the EM algorithm based on the complete log-likelihood function 
\begin{equation} \label{llc}
\ell_c(\vect{\theta})=\sum_{i=1}^n \sum_{\vect{j} \in \mathbb{Z}^p} \omega_{i\vect{j}} \log m(\vect{y}_i + 2 \pi \vect{j}; \vect{\theta})  \ .
\end{equation}
In the Expectation step (E-step), we evaluate the conditional expectation of (\ref{llc}) given the observed data and the current parameters value $\vect{\theta}$ by computing the conditional probability that $\vect{y}_i$ has $\vect{j}$ as wrapping coefficients vector, that is
\begin{equation*}
\omega_{i\vect{j}}=\frac{m(\vect{y}_i + 2 \pi \vect{j}; \vect{\theta})}{\sum_{\vect{k} \in \mathbb{Z}^p} m(\vect{y}_i + 2 \pi \vect{k}; \vect{\theta})}, \forall \vect{j} \in \mathbb{Z}^p \ .
\end{equation*}
Parameters estimation is carried out in the Maximization step (M-step) solving the set of (complete) likelihood equations
\begin{equation*}
\sum_{i=1}^n \sum_{\vect{j} \in \mathbb{Z}^p} \omega_{i\vect{j}} u(\vect{y}_i + 2 \pi \vect{j}; \vect{\theta}) = \vect{0} \ ,
\end{equation*}
with $u(\vect{y}_i + 2 \pi \vect{j}; \vect{\theta})=\nabla_{\vect{\theta}} \log m(\vect{y}+ 2 \pi \vect{j}; \vect{\theta})$.
An alternative estimation strategy can be based on a CEM algorithm leading to an approximated solution. At each iteration, a Classification step (C-step) is performed after the E-step, that provides crispy assignments. Let 
\begin{equation*}
\hat{\vect{j}}_i= \argmax_{\vect{j} \in \mathbb{Z}^p} \omega_{i\vect{j}}, 
\end{equation*}
then, set $\omega_{i\vect{j}} = 1$ when $\vect{j}=\hat{\vect{j}}_i$, $\omega_{i\vect{j}} = 0$ otherwise. As a result, the torus data $\vect{y}_i$ are \textit{unwrapped} to (fitted) linear data $\hat{\vect{x}}_i = \vect{y}_i + 2\pi \hat{\vect{j}}_i$. It is easy to see that the M-step simplifies to
\begin{equation*}
\sum_{i=1}^n u(\hat{\vect{x}}_{i}; \vect{\theta}) = \vect{0} \ .
\end{equation*}
Both the procedures are iterated until some convergence criterion is fulfilled, that could be based on the changes in the likelihood or in fitted parameter values \citep{nodehi2020}. 


\end{document}